\documentclass[a4paper,UKenglish,cleveref, autoref, thm-restate,authorcolumns]{lipics-v2019}

\title{Fast Preprocessing for Optimal Orthogonal Range Reporting and Range Successor with Applications to Text Indexing} 

\titlerunning{Fast Preprocessing for Orthogonal Range Reporting and Range Successor} 

\author{Younan Gao}{Faculty of Computer Science, Dalhousie University, Canada}{yn803382@dal.ca}{}{}

\author{Meng He}{Faculty of Computer Science, Dalhousie University, Canada}{mhe@cs.dal.ca}{}{}

\author{Yakov Nekrich}{Department of Computer Science, Michigan Technological University, USA}{yakov.nekrich@googlemail.com}{}{}

\authorrunning{Y. Gao, M. He and Y. Nekrich} 

\Copyright{Younan Gao, Meng He and Yakov Nekrich} 



\category{} 

\relatedversion{} 

\supplement{}

\nolinenumbers 

\usepackage{versions}
\excludeversion{REMOVED}

\DeclareMathOperator{\polylog}{polylog}

\def\idtt#1{\ensuremath{\mathtt{#1}}}
\def\pred{\idtt{pred}}
\def\succ{\idtt{succ}}
\def\rmq{\idtt{rmq}}
\def\rMq{\idtt{rMq}}
\def\rankop{\idtt{rank}}
\def\selop{\idtt{select}}
\def\prank{\idtt{rank'}}
\def\countop{\idtt{count}}

\def\noderange{\idtt{noderange}}
\def\point{\idtt{point}}

\def\lca{\idtt{lca}}

\def\occ{\idtt{occ}}

\newcommand\Sp{\mathit{S\!p}}

\begin{document}
	
	\maketitle
	
	\begin{abstract}
		
		Under the word RAM model, we design three data structures that can be constructed in $O(n\sqrt{\lg n})$ time over $n$ points in an $n \times n$ grid. 
		The first data structure is an $O(n\lg^{\epsilon} n)$-word structure supporting orthogonal range reporting in $O(\lg\lg n+k)$ time, where $k$ denotes output size and $\epsilon$ is an arbitrarily small constant.
		The second is an $O(n\lg\lg n)$-word structure supporting orthogonal range successor in $O(\lg\lg n)$ time, while the third is an $O(n\lg^{\epsilon} n)$-word structure supporting sorted range reporting in $O(\lg\lg n+k)$ time. 
		The query times of these data structures are optimal when the space costs must be within $O(n\polylog n)$ words.
		Their exact space bounds match those of the best known results achieving the same query times, and the $O(n\sqrt{\lg n})$ construction time beats the previous bounds on preprocessing.
		Previously, among 2d range search structures, only the orthogonal range counting structure of Chan and P\v{a}tra\c{s}cu (SODA 2010) and the linear space, $O(\lg^{\epsilon} n)$ query time structure for orthogonal range successor by Belazzougui and Puglisi (SODA 2016) can be built in the same $O(n\sqrt{\lg n})$ time.
		Hence our work is the first that achieve the same preprocessing time for optimal orthogonal range reporting and range successor. 
		We also apply our results to improve the construction time of text indexes. 
	\end{abstract}
	
	\newpage

\section{Introduction}
\label{sec:introduction}

Two dimensional orthogonal range search problems have been studied intensively in the communities of computational geometry, data structures and databases. 
The goal of these problems is to maintain a set, $N$, of points on the plane in a data structure 
such that one can efficiently compute aggregate information about the points contained in an axis-aligned query
rectangle $Q$.
Among these problems, {\em orthogonal range counting} and {\em orthogonal range reporting} are perhaps the most fundamental; the former counts the number of points contained in $N \cap Q$ while the latter reports them.
Another well-known problem is {\em orthogonal range successor}, which asks for the point in $N\cap Q$ with the smallest $x$- or $y$-coordinate.
Range counting, reporting and successor have many applications including text indexing~\cite{makinen2006position,bhmm2009,DBLP:journals/corr/abs-1108-3683,abs-1712-07431}, Lempel-Ziv decomposition~\cite{belazzougui2016range} and consensus trees in phylogenetics~\cite{jls2017}, to name a few. 
See \cite{Lewenstein13} for a survey on the connection between text indexing and various range searching techniques.

Most work on orthogonal range search~\cite{c1988,jms2004,chan2011orthogonal,nekrich2012sorted,zhou2016two} focuses on achieving the best tradeoffs between query time and space, and preprocessing time is often neglected.
However, the preprocessing time of a data structure matters when it is used as a building block of an algorithm processing plain data, as the total running time includes that needed to build the structure.
Furthermore, an orthogonal range search structures with fast construction time are preferred when preprocessing huge amounts of data, e.g., when used as  components of text indexes built upon large data sets from search engines and bioinformatics applications. 
The work of Chan and P\v{a}tra\c{s}cu~\cite{chan2010counting} is the first that breaks the $O(n \lg n)$ bound on the construction time of 2d orthogonal range counting structures; they designed an $O(n)$-word structure with $O({\lg n}/{\lg \lg n})$ query time that can be built in $O(n\sqrt{\lg n})$ time.
Their ideas were further extended to design an $O(n{\lg \sigma}/{\sqrt{\lg n}})$-time algorithm to build an binary wavelet trees over a string of length $n$ drawn from $[\sigma]$~\cite{munro2016fast,babenko2015wavelet}\footnote{In this paper, $\sigma$ denotes $\{0, 1, \ldots, \sigma-1\}$.}, which is a key data structure used in succinct text indexes. 
More recently, Belazzougui and Puglisi~\cite{belazzougui2016range} showed how to construct an $O(n)$-word data structure in $O(n\sqrt{\lg n})$ time to support range successor in $O(\lg^{\epsilon} n)$ time, and applied it to achieve new results on Lempel-Ziv parsing.

The previous work on constructing orthogonal range search structures in $O(n\sqrt{\lg n})$ time focuses on linear space data structure.
To achieve optimal query time for 2d orthogonal range reporting and range successor using near-linear space, however, the best tradeoffs under the word RAM model requires superlinear space~\cite{chan2011orthogonal,zhou2016two}.
The increased space costs are needed to encode more information, posing new challenges to fast construction.
We thus investigate the problem of designing data structures with optimal query times for range reporting and range successor that can be built in $O(n\sqrt{\lg n})$ time, while matching the space costs of the best known solutions.
We also consider a closely related problem called {\em sorted range reporting}~\cite{nekrich2012sorted} to achieve similar goals.
In this problem, we report all points in $N\cap Q$ in a sorted order along either $x$- or $y$-axis.
The query time should depend on the number of points actually reported even if the procedure is ended early by user.

\subparagraph*{Previous Work.} The research on 2d orthogonal range reporting has a long history~\cite{c1988,abr2000,jms2004,chan2011orthogonal}.
Researchers have achieved three best tradeoffs between query time and space costs under the word RAM model; we follow the state of the art and assume that the input points are in rank space.
The solution with optimal query time of $O(\lg\lg n +k )$ and space cost of $O(n \lg^{\epsilon} n)$ words is due to Alstrup et al.~\cite{abr2000}, while the best linear-space solution is designed by Chan et al~\cite{chan2011orthogonal} which answers a query in $O((1+k) \lg^{\epsilon} n )$ time, where $k$ is the output size and $\epsilon$ is an arbitrarily small constant.
Chan et al. also proposed an $O( \lg \lg n)$-word structure with $O((1+k) \lg \lg n )$ query time and another tradeoff matching that of Alstrup et al.~\cite{abr2000}.

The 2d orthogonal range successor problem has also been well studied. After a series of work~\cite{ls1994,kkl2007,ckwir2010,cikrtw2012,yhw2011}, Nekrich and Navarro~\cite{nekrich2012sorted} gave two solutions to this problem; the first uses $O(n)$ words and answers a query in $O(\lg^{\epsilon} n)$ time, while the second uses $O(n\lg\lg n)$ words to answer a query in $O((\lg\lg n)^2)$ time.
Zhou~\cite{zhou2016two} decreased the query time of the latter to $O(\lg\lg n)$ without increasing space costs. 
By definition, a solution to orthogonal range successor implies that to sorted range reporting. 
Furthermore, Nekrich and Navarro~\cite{nekrich2012sorted} also designed a data structure using $O(n \lg^{\epsilon} n)$ words to support sorted range reporting in $O(\lg\lg n +k )$ time.
Hence, the best three time-space tradeoffs for the original 2d orthogonal range reporting problem has also been achieved for the sorted version. 
The optimality of the $O(\lg\lg n +k)$ query time for orthogonal range reporting and the $O(\lg\lg n)$ query time for orthogonal range successor when no more than $O(n\polylog n)$ space can be used is established by a lower bound on range emptiness~\cite{pt2006}.

Alstrup et al.~\cite{abr2000} claimed that their structure for optimal orthogonal range reporting can be constructed in $O(n\lg n)$ expected time.
Even though preprocessing times are not given in \cite{chan2011orthogonal, nekrich2012sorted, zhou2016two}, straightforward analyses reveal that the other data structures we surveyed here can be built in $O(n \lg n)$ worst-case time (Bille and G{\o}rtz~\cite{DBLP:journals/corr/abs-1108-3683} also claimed that the preprocessing time of the $O(n\lg\lg n)$-word structure of Chan et al.~\cite{chan2011orthogonal} is $O(n \lg n)$). 
Hence, when faster preprocessing time is needed in their solution to Lempel-Ziv decomposition, Belazzougui and Puglisi~\cite{belazzougui2016range} had to design a new linear-space data structure for orthogonal range successor with $O(n\sqrt{\lg n})$ preprocessing time and $O(\lg^{\epsilon} n)$ query time. 
No attempts have been published to achieve similar preprocessing times for other tradeoffs.

\subparagraph*{Our Results.} Under the word RAM model, we design the following three data structures that can be constructed in $O(n\sqrt{\lg n})$ time over $n$ points in an $n \times n$ grid: 
\begin{itemize}
	\item An $O(n\lg^{\epsilon} n)$-word structure supporting orthogonal range reporting in $O(\lg\lg n+k)$ time, where $k$ denotes output size and $\epsilon$ is an arbitrarily small constant;
	\item An $O(n\lg\lg n)$-word structure supporting orthogonal range successor in $O(\lg\lg n)$ time;
	\item An $O(n\lg^{\epsilon} n)$-word structure supporting sorted range reporting in $O(\lg\lg n+k)$ time.
\end{itemize}

The query times of these structures are optimal when space costs must be within $O(n\polylog n)$ words.
Their exact space bounds match those of the best known results achieving the same query times, and the $O(n\sqrt{\lg n})$ construction time beats the previous bounds on preprocessing.
Note that even though our third result implies the first, our data structure for the first is much simpler.
In addition, our results can be used to improve the construction time of text indexes. For a text string $T$ of length $n$ over alphabet $[\sigma]$, we design 
\begin{itemize}
	\item A text index of $O(n\lg \sigma\lg^{\epsilon} n)$ bits that can be constructed in $O(n\lg \sigma/\sqrt{\lg n})$ time and can report the $\occ$ occurrences of a pattern of length $p$ in time $O({p}/{\log_{\sigma} n}+\log_{\sigma} n\lg \lg n+\occ)$, where $\epsilon$ is any small positive constant.
	This improves one result of Munro et al.~\cite{abs-1712-07431} who designed the first text indexes with both sublinear construction time and query time for small $\sigma$;
	for the same time-space tradeoff, their preprocessing time is $O(n\lg \sigma \lg^{\epsilon} n)$. 
	\item A text index of $O(n\lg^{1+\epsilon} n)$ bits for any constant $\epsilon>0$ built in $O(n\sqrt{\lg n})$ time that supports position-restricted substring search~\cite{makinen2006position} in $O({p}/{\log_{\sigma} n}+\lg p+\lg \lg \sigma+\occ)$ time. Previous indexes with similar query performance require $O(n\lg n)$ construction time. 
\end{itemize}

\subparagraph*{Overview of Our Approach.} We first discuss why some obvious approaches will not work.
The modern approach of Chan et al~\cite{chan2011orthogonal} for orthogonal range reporting is based on a problem called ball inheritance which they defined over range trees. 
This solution is well-known for its simplicity, and by choosing different parameters in their approach to ball inheritance, they obtain all three best known tradeoffs. 
One natural idea is to redesign the structures stored at range tree nodes to use bit packing to speed up construction.
However, 
even though we have achieved construction time matching the state of the art for these structures, it is still not enough to construct the data structures for the tradeoffs of ball inheritance that we need quickly enough. 
Another idea is to tune the parameters in the approach of Belazzougui and Puglisi~\cite{belazzougui2016range}, hoping to obtain the tradeoffs that we aim for, as they already showed how to construct in $O(n\sqrt{\lg n})$ time a linear space, $O((k+1)\lg^{\epsilon} n)$ query time structure for orthogonal range reporting.
Their solution uses many trees grouped into $O(\lg^{\epsilon} n)$ levels of granularity. 
If we borrow ideas from \cite{chan2011orthogonal} to set parameters to achieve different tradeoffs,
we would use $O(1/\epsilon)$ or $O(\lg\lg n)$ levels of granularity. 
However, to return a point in the answer, their query algorithm would perform operations requiring $O(\lg\lg n)$ time at each level of granularity.
Thus, at best, the former would give an $O(n\lg^{\epsilon} n)$-word structure with $O((k+1)\lg\lg n)$ query time and the latter an $O(n\lg\lg n)$-word structure with $O((k+1)(\lg\lg n)^2)$ query time.
Either solution is inferior to best known tradeoffs.
This however is fine in the original solution, as the total cost of spending $O(\lg\lg n)$ time at each of the $O(\lg^{\epsilon} n)$ levels is bounded by $O(\lg^{\epsilon'} n)$ for any $\epsilon' > \epsilon$. 

We thus design new approaches. 
For optimal orthogonal range reporting, our overall strategy is to perform two levels of reductions, making it sufficient to solve ball inheritance in special cases with fast preprocessing time.
More specifically, we first use a generalized wavelet tree and range minimum/maximum structures to reduce the problem in the general case to the special case in which the points are from a $2^{\sqrt{\lg n}}\times n'$ (narrow) grid, where $n' \le n$. 
In this reduction, we need only support ball inheritance over a wavelet tree with high fanout. 
We further reduce the problem over points in a narrow grid to that over a (small) grid of size at most $2^{\sqrt{\lg n}} \times 2^{2\sqrt{\lg n}}$.
This is done by grouping points and selecting representatives from each group, so that previous results with slower preprocessing can be used over the smaller set of representatives. 
Finally, over the small grid, we solve ball inheritance when the coordinates of each point can be encoded in $O(\sqrt{\lg n})$ bits.
The ball inheritance structures in both special cases can be built quickly by redesigning components with fast preprocessing, though the second case requires a twist to the approach of Chan et al~\cite{chan2011orthogonal}.
Our solutions to optimal range successor and sorted range reporting are based on similar strategies, though we preform more levels of reductions. 

In the main body of this paper, we describe our data structures for optimal range reporting and successor, while leaving those for optimal sorted range reporting in Appendix~\ref{app:sorted}. 

\section{Preliminaries}
\label{sect:preliminaries}
In this section, we describe and sometimes extend the previous results used in this paper. 
The proofs omitted from this section can be found in Appendix \ref{app: preliminary}.

\subparagraph*{Notation.}
We adopt the word RAM model with word size $w = \Theta(\lg n)$ bits, where $n$ often denotes the size of the given data. Our complete solutions use several sets of homogeneous components. We present a lemma to bound the costs of each different type of components, which is then applied over the entire set of these components to calculate the total cost. The size, $n'$, of the data that each component represents may be less than $n$ which is the input size of the entire problem, so when the cost of constructing the component is bounded by a function of the form $f(n')/\polylog(n)$ to take advantage of the word size, we keep both $n'$ and $n$ in the lemma statement, as commonly done in previous work on similar topics. In this case, the construction algorithm usually uses a universal table of $o(n)$ bits, whose content solely depends on the value of $n$, and hence can be constructed once in $o(n)$ time and used for all data structure components of the same type. Thus unless otherwise stated, these lemmas assume the existence of such a table without stating so explicitly in the lemma statements, and we define and analyze the table in the proof. This also applies to algorithms that manipulates sequences of size $n'$. Occasionally the query algorithms of a data structure may need a universal table as well, and we explicitly state it if this is the case.


We say a sequence $A \in [\sigma]^n$ is in $packed$ form if the bits of its elements are concatenated and stored in as few words as possible.
Thus, when packed, $A$ occupies $\lceil n\lceil{\lg \sigma}\rceil/w\rceil$ words.


\begin{REMOVED}
Under the word RAM model with word size $w = \Theta(\lg n)$, $O(\lg n / \lg \sigma)$ integers from $[\sigma]$ can be packed in one word.
Using the well-known approach of simulating external memory algorithms under word RAM~\cite{w2014}, the following algorithm for sorting packed sequences follows directly from external sorting:

\begin{lemma}
	A packed sequence $A[0..n^{\prime}-1]$ from alphabet $[\sigma]$, where $max(\sigma, n') \leq n$, can be sorted in $O(n^{\prime}\lg n'{\lg \sigma}/{\lg n})$ time with the help of a universal tables of $o(n)$ bits.
\end{lemma}

Next, we show a fast and stable sorting algorithm over a packed sequence of small integers: 
\begin{lemma}
	Let $A[0..n^{\prime}-1]$ be a packed sequence drawn from alphabet $[\sigma]$, where $max(\sigma, n') \leq n$. 
	Let $\psi$ denote the number of distinct symbols in $A$ where $\psi\leq n^{\prime}$. 
	There is an algorithm to stably sort $n^{\prime}$ elements in $A$ in $O(n^{\prime}({\lg \sigma})(\lg \psi)/{\lg n})$  time with the help of a universal tables $o(n)$ bits.
	The sorting algorithm uses $n^{\prime}\lg \sigma$ bits of working space.
\end{lemma}
\begin{proof}
	The sorting algorithm is the same as creating a balanced binary tree $T$. 
	We sort part of elements in $A$ at each tree level. 
	Eventually, all elements are sorted at the leaf level. 
	Let $r$ denote the root node. We create the sequence $A(r)=A$. 
	Among the $n^{\prime}$ elements encoded with $\lg \sigma$ bits each, we find the length $\gamma$ of the common most significant bits. 
	If $\gamma$ is equal to $\lg \sigma$, it means all elements in $A(r)$ are the same, we stop classifying elements in the next level and make it as a leaf node. 
	Otherwise, we create the left child, $r_0$, and the right child, $r_1$, of $r$, and perform a linear scan of $A(r)$.
	During the scan, for each $i\in[0, |A(r)|-1]$,  if the $(\gamma+1)$-th bit of the binary representation of $A(r)[i]$ is 0 or 1, $A(r)[i]$ is appended to into $A(r_0)$ or $A(r_1)$, respectively. 
	Afterwards, we discard the sequence $A(r)$ and recursively process the child node $r_0$ and $r_1$ in the same manner.
	In general, when generating the sequences for the child nodes of an internal node $u$, we find the length $\gamma$ of the common most significant bits among the binary expression of elements of $A(u)$.
	If $\gamma$ is equal to $\lg \sigma$, we make $u$ as a leaf node. 
	Otherwise, for each $i\in[0, |A(u)|-1]$,  if the $(\gamma+1)$-th bit of the binary representation of $A(u)[i]$ is 0 or 1, $A(u)[i]$ is appended to $A(u_0)$ or $A(u_1)$, respectively. 
	
	To speed up this process, we use a universal table $U$.
	Let $b=\lfloor \frac{\lg n}{2\lg \sigma} \rfloor$ denoting the block size. 
	For a packed sequence $\hat{S}$ of any possible $b$-element drawn from $[\sigma]$, $U$ stores an answer encoded with $\lg \lg \sigma$ bits denoting the length of the common significant bits among elements of $\hat{S}$. 
	As there are $2^{b\times \lg \sigma}\le\sqrt{n}$ entries in $U$ and each entry stores a result of $\lg \lg \sigma$ bits, $U$ occupies $O(\sqrt{n}\times\lg \lg \sigma)=o(n)$ bits of space. 
	By performing table lookups with $U$, we can retrieve the length of the common significant bits among elements of a sequence $A(u)$ in $O(|A(u)|/b+1)=O(|A(u)|\lg \sigma/\lg n+1)$ time.
	In addition, we need another universal table $U'$ of $o(n)$ bits (similar to the table $U$ in the proof of Lemma \ref{lemma:wavelet_construct_d_packed}) to divide a packed sequence of any possible $b$ elements into two packed subsequences according to the bit value of each element at a given position $\gamma$, where $\gamma\in[0, \lg \sigma-1]$. 
	With $U$ and $U'$, generating the sequences for the child nodes of an internal node $u$ can be processed in $O(|A(u)|\lg \sigma/\lg n+1)$ time.
	As $A[0, n'-1]$ has $\psi$ distinct elements, the tree $T$ would have as many as $\psi$ nodes and there are at least $\psi$ sequences to be generated.
	Hence, the sorting algorithm would take $\Omega(\psi)$ time. 
	However, $\psi$ could be as large as $n'$, which is too expensive to afford.
	Thus, we modify the structure of $T$ to decrease the $\Omega(\psi)$ term.
	We add one more condition when making a node as a leaf: If $|A(u)|$ at node $u$ are $\le b$, we make $u$ as a leaf node. 
	The modification would bring two following properties.
	First, if a leaf sequence $|A(l)|$ satisfies $|A(l)| > b$, then all elements of $A(l)$ share the same symbol. 
	In this case, we do not need to sort elements in $A(l)$.
	Second, as there are at most $\lceil n'/b \rceil$ nodes at each level, the sorting algorithm generates in total $O(n'/b\times \lg \psi)=O(n'\lg \sigma/\lg n\times \lg \psi)$ leaf sequences.
	For each of $O(\frac{n^{\prime}}{b}\times \lg \psi)$ leaf sequences $A(l)$, if $|A(l)|\le b$, we can apply a universal table $U''$ of $o(n)$ bits to sort them in constant time.
	Otherwise, all elements in $A(l)$ are the same and sorting is unnecessary.
	$U''$ has an entry for each possible pair  $(E, c)$, where $E$ is a
	sequence of length $b$ drawn from universe $[\sigma]$ and $c$ is an integer in $[0, b]$.
	This entry stores a packed sequence of the leftmost $c$ elements of $E$ in a sorted order occupying $c\lg \sigma$ bits of space.
	Similar to $U$, $U''$ uses $o(n)$ bits of space.
	After generating all leaf sequences, we linearly scan each of them.
	If a leaf sequence $A(l)$ satisfies $|A(l)|\le b$, we sort them in constant time with $U''$.
	At last, we merge all the sorted leaf sequences into one packed sequence.
	
	At each node $u$, we spend $O(|A(u)|\lg \sigma/\lg n+1)$ time on processing.
	The sum of the lengths of all the elements in the sequences at the same level  is at most $n'$.
	As at most $\lg \psi$ levels and at most $O(n'\lg \sigma/\lg n\times \lg \psi)$ sequences are constructed, the time required to construct the leaf sequences is $\sum_u O(|A(u)|\lg \sigma/\lg n+1) = O(n'\lg \sigma/\lg n\times \lg \psi)$.
	Finally, sorting each leaf sequence and merge all of them into one sequence takes another $O(n'\lg \sigma/\lg n\times \lg \psi)$ time.
	Overall, the query algorithm takes $O(n^{\prime}({\lg \sigma})(\lg \psi)/{\lg n})$ time.
\end{proof}
\end{REMOVED}

\label{sec:wavelet}

\subparagraph*{Generalized Wavelet Trees.} 
Given a sequence $A[0..n-1]$ drawn from alphabet $[\sigma]$, a $d$-ary generalized wavelet tree~\cite{FerraginaMMN07} $T_d$ over $A$ is a balanced tree in which each internal node has $d$ children, where $2\leq d \leq \sigma$. For simplicity, assume that $\sigma$ is a power of $d$.
Each node of $T_d$ then represents a range of alphabet symbols defined as follows: At the leaf level, the $i$-th leaf from left represents the integer range $[i, i]$ for each $i \in [0..\sigma-1]$. 
The range represented by an internal node is the union of the ranges represented by its children. 
Hence the root represents $[0, \sigma-1]$, and $T_d$ is a complete tree having $\log_d \sigma+1$ levels. 
Each node $u$ is further associated with a subsequence, $A(u)$, of $A$, in which $A(u)[i]$ stores the $i$-th entry in $A$ that is in the range represented by $u$. 
Thus the root is associated with the entire sequence $A$.
To save storage, $A[u]$ is not stored explicitly in \cite{FerraginaMMN07}.
Instead, each internal node $u$ stores a sequence $S(u)$ of integers in $[d]$, where $S(u)[i] = j$ if $A(u)[i]$ is within the range represented by the $j$th child of $u$.
All the $S(u)$'s built for internal nodes occupy $O(n \lg \sigma)$ bits in total.

Generalized wavelet trees share fundamental ideas with range trees but are more suitable for compact data structures over sequences which may contain duplicate values.
When we use them in this paper, we sometimes explicitly store $A(u)$ for each node $u$, and may even associate with $u$ an additional array $I(u)$ in which $I(u)[i]$ stores the index of $A(u)[i]$ in the original sequence $A$.
We call $A(u)$ the {\em value array} of $u$, and $I(u)$ the {\em index array}.
In this paper, if we construct value and/or index arrays for each node, we explicitly state so.
If not, it implies that we build a wavelet tree in which each node $u$ is associated with $S(u)$ only.
Furthermore, unless otherwise specified, we apply the standard pointer-based implementation to represent the tree structure of a wavelet tree, which is preprocessed in time linear to the number of tree nodes such that the lowest common ancestor of any two nodes can be located in $O(1)$ time \cite{bender2004level}.
We also number the levels of the tree incrementally starting from the root level, which is level $0$.
We have the following two lemmas on constructing wavelet trees:



\begin{lemma}
	\label{lemma:wavelet_construct_d_packed}
	Let $A[0..n^{\prime}-1]$ be a packed sequence drawn from alphabet $[\sigma]$ and $I[0..n^{\prime}-1]$ be a packed sequence in which $I[i]=i$ for each $i\in[0..n^{\prime}-1]$, where $n' \le n$ and $\sigma\leq 2^{O(\sqrt{\lg n})}$.
	Given $A$ and $I$ as input, a $d$-ary wavelet tree over $A$ with value and index arrays in packed form can be constructed in $O(n^{\prime}\lg \sigma(\lg n^{\prime}+\lg \sigma)/{\lg n}+ \sigma)$ time, where $d$ is an arbitrary power of $2$ with $2 \le d \le \sigma$.
        If index arrays are not constructed, the construction time can be lowered to $O(n^{\prime}\lg^2\sigma/{\lg n}+ \sigma)$; this bound still applies when neither value nor index arrays are built. 
\end{lemma}




\begin{lemma}
	\label{lemma:wavelet_construct_d_unpacked}
	Let $A[0..n-1]$ be a sequence drawn from alphabet $[\sigma]$. 
	A $d$-ary wavelet tree over $A$ with value and index arrays can be built in $O({n\lg \sigma}/{\lg d})$ time where $2\leq d\leq \sigma$.
\end{lemma}

A sequence $A[0..n-1]$ drawn from $[\sigma]$ can be viewed as a point set $N = \{(A[i], i)| 0\le i \le n-1\}$.
Let $T$ be a $d$-ary wavelet tree constructed over $A$.
Then {\em ball inheritance}~\cite{chan2011orthogonal} can be defined over $T$ which asks for the support of these operations: 
i) $\point(v, i)$, which returns the point $(A(v)[i], I(v)[i])$ in $N$ for an arbitrary node $v$ in $T$ and an integer $i$; and 
ii) $\noderange(c, d, v)$, which, given a range $[c, d]$ and a node $v$ of $T$, finds the range $[c_v, d_v]$ such that $I(v)[i] \in [c, d]$ iff $i \in [c_v, d_v]$.
If we store the value and index arrays explicitly, it is trivial to support these operations, but the space cost is high. 
To save space, we only store $S(v)$ for each node $v$ and design auxiliary structures. The following lemma presents previous results:

\begin{lemma}[{\cite[Theorem 2.1]{chan2011orthogonal}}, {\cite[Lemma 2.3]{chan2017succinct}}]
	\label{lemma:ball_intro}
	A generalized wavelet tree over a sequence $A[0..n-1]$ drawn from $[\sigma]$ can be augmented with ball inheritance data structure in $O(n\lg n f(\sigma))$ bits to support $\point$ in $O(g(\sigma))$ time and $\noderange$ in $O(g(\sigma)+\lg \lg n)$ time, where (a) $f(\sigma)=O(1) $ and $g(\sigma)=O(\lg^{\epsilon} \sigma)$; (b) $f(\sigma)=O(\lg \lg \sigma) $ and $g(\sigma)=O(\lg \lg \sigma)$; or (c) $f(\sigma)=O(\lg^{\epsilon} \sigma) $ and $g(\sigma)=O(1)$.
\end{lemma}


\subparagraph*{Data Structures for $\rankop$ and $\selop$.}
Given a sequence $A$ drawn from alphabet $[\sigma]$, a $\rankop_c(A, i)$ operation computes the number of elements equal to $c$ in $A[0..i]$, where $c\in[\sigma]$,  while a $\selop_c(A,i)$ returns the index of the entry of $A$ containing the $i$-th occurrence of $c$. We have the following two lemmas on building $\rankop$/$\selop$ structures.

\begin{lemma}
	\label{lemma_rank_prime_small}
	Let $A[0..n^{\prime}-1]$ be a packed sequence drawn from alphabet $[\sigma]$, where $n' \le n$ and $\sigma = O(\polylog n)$. A data structure of $n^{\prime}\lceil\lg \sigma\rceil+o(n^{\prime}\lg \sigma)$ bits supporting $\rankop$ in $O(1)$ time can be constructed in $O(n^{\prime}\lg^2 \sigma/\lg n+\sigma)$ time. 
\end{lemma}


\begin{lemma}[{\cite[Lemma 2.1]{babenko2015wavelet}}]
	\label{bit_sequence}
	Given a packed bit sequence $B[0..n-1]$, a systematic data structure occupying $o(n)$ extra bits can be constructed in $O(n/\lg n)$ time, which supports $\rankop$ and $\selop$ in constant time.
\end{lemma}

In the above lemma, a data structure is {\em systematic} if it requires the input data to be stored verbatim along with the additional information for answering queries. 
A restricted version of $\rankop$ is called {\em partial rank}; a partial rank operation, $\prank(A, i)$, computes the number of elements equal to $A[j]$ in $A[0..j]$. The following lemma presents a solution to supporting $\prank$, which is an easy extension of \cite[Lemma 3.5]{belazzougui2020linear}. 

\begin{lemma}
	\label{rank_select_0}
	Given a sequence $A[0..n-1]$ drawn from alphabet $[\sigma]$, a data structure of $O(n\lg \sigma)$ bits can be constructed in $O(n+\sigma)$ time, which supports $\prank$ in constant time. 
\end{lemma}

\subparagraph*{Range Minimum/Maximum.} Given a sequence $A$ of $n$ integers, a range minimum/maximum query $\rmq(i, j)/\rMq(i, j)$ with $i\leq j$ returns the position of a minimum/maximum element in the subsequence $A[i..j]$. 
Fischer and Heun~\cite{fischer2011space} considered this problem: 

\begin{lemma}[\cite{fischer2011space}]
	\label{fisher_rmq}
	Given an array $A$ of $n$ integers, a data structure of $O(n)$ bits can be constructed in $O(n)$ time, which answers $\rmq/\rMq$ in $O(1)$ time without accessing $A$.
\end{lemma}


We further build an auxiliary structure upon a packed sequence $A$ under the {\em indexing model}: after the the data structure is built, $A$ itself need not be stored verbatim; to answer a query, it suffices to provide an operator that can retrieve any element in $A$.

\begin{lemma}
	\label{minmax_index}
	Let $A[0..n^{\prime}-1]$ be a packed sequence drawn from alphabet $[\sigma]$, where $ \sigma \leq 2^{\sqrt{\lg n}}$ and $n'\le n$. 
	There is a data structure using $O(n^{\prime}\lg \lg n)$ extra bits constructed in $O(n^{\prime}{\lg \sigma}/{\lg n})$ time, which answers $\rmq/\rMq$ in $O(1)$ time and $O(1)$ accesses to the elements of $A$. 
	The query procedure uses a universal table of $o(n)$ bits.
\end{lemma}

\section{Fast Construction of $\prank$ Query Structures}
\label{section_partial_rank}

In this section we focus on how to efficiently construct data structures for $\prank$ queries over a sequence $A[0..n'-1]$ drawn from alphabet $[\sigma]$, where $n' \le n$ and $\sigma\leq 2^{\sqrt{\lg n}}$.
This is needed to solve ball inheritance in a special case that we need. Lemma \ref{lemma_rank_prime_small} already solves this problem when $\sigma\leq \lg n$, so we assume $\lg n < \sigma \leq 2^{\sqrt{\lg n}}$ in the rest of this section.

In our solution, we conceptually divide sequence $A$ into chunks of length $\sigma$. For simplicity, assume that $n^{\prime}$ is a multiple of $\sigma$. Let $A_k$ denote the $k$th chuck, where $0 \leq k \leq n^{\prime}/\sigma-1$. For each $c \in [0, \sigma-1]$, we define the following data structures:
\begin{itemize}
	\item A bitvector $B_c=1^{\rankop_c(A_0, \sigma)}01^{\rankop_c(A_1, \sigma)}0\ldots 1^{\rankop_c(A_{n^{\prime}/\sigma-1}, \sigma)}0$, which encodes the number of occurrences of symbol $c$ in each chunk in unary. $B_c$ is represented using Lemma~\ref{bit_sequence} to support $\rankop$ and $\selop$ in constant time. 
	\item A sequence $P_c[0..n^{\prime}/\sigma-1]$, in which $P_c[i] = \prank(A_i,c)$ for each $i \in [0, n^{\prime}/\sigma-1]$, i.e., $P_c[i]$ stores the answer to a partial rank query performed locally within $A_i$ at position $c$.
\end{itemize}

Note that we have one $B_c$ for each alphabet symbol $c$, while we have one $P_c$ for each relative position $c$ in the chunks of $A$. The support for $\prank$ over $A$ follows from: 
$\prank(A, j) = \selop_0(B_c, t)-(t-1)+P_{\tau}[t],  where\ \tau = j\bmod \sigma, t=\lfloor \frac{j}{\sigma} \rfloor, and \ c=A[j]$.

\begin{REMOVED}
We have the following lemma on supporting queries using these data structures, with a space analysis. 

\begin{lemma}
	\label{rank_prime_queryspace}
	The data structures in this section occupy $n^{\prime}\lg \sigma+ o(n^{\prime}\lg \sigma)$ extra bits and support $\prank$ in $O(1)$ time and $O(1)$ accesses to elements of $A$. 
\end{lemma}

\begin{proof}
	See Appendix \ref{app: space_rank_prime_queryspace} for space analysis.
	A query $\prank(A, j)$ can be answered as follows:
	$$
	\prank(A, j) = select_0(B_c, t)-(t-1)+P_{\tau}[t],  where\ \tau = j\bmod \sigma, t=\lfloor \frac{j}{\sigma} \rfloor, and \ c=A[j]
	$$
	As the $\selop$ query over $B_c$ takes constant time, answering $\prank(A, j)$ requires $O(1)$ time and a single access to $A$.
\end{proof}

Next, we consider how to construct the sequences $B_c$'s efficiently.

\begin{lemma}
	\label{lemma_B_c}
	Bitvectors $B_0, B_1, \ldots,B_{\sigma-1}$ can be constructed in $O(n^{\prime}\lg^2 \sigma/\lg n+\sigma)$ time in total. 
\end{lemma}

\begin{proof}
\end{REMOVED}
  
	To construct the sequences $B_c$'s, we first construct a sequence $M[0..n^{\prime}+n^{\prime}/\sigma-1]$ in which each element is encoded in $\lceil\lg \sigma\rceil+1$ bits.
	In $M$, $n'$ elements are {\em regular elements}, and the rest are {\em boundary elements} each of which is an integer whose binary expression simply consists of $\lceil\lg \sigma\rceil+1$ $0$-bits.
	$M$ is divided into $n^{\prime}/\sigma$ chunks, and each chunks contains $\sigma$ regular elements followed by a boundary element.
	The subsequence of the $\sigma$ regular elements in the $i$-th chunk can be obtained by appending a $1$-bit to the end of the binary expression of each element in $A_k$. 
	$M$ can be computed in $O(n^{\prime}\lg \sigma/\lg n)$ time using a universal table $U$; See Appendix \ref{app: generate_M} for details.
	
      

    We then construct a tree $T$ over $M$, in which each node $u$ is associated with a sequence $M(u)$.
    At the root node $r$ of $T$, we set $M(r)=M$, and we perform the following recursive procedure at any node $u$ at level $l$ of $T$ where $l \in [0, \lceil\lg \sigma\rceil-1]$:
    We create the left child, $u_0$, and the right child, $u_1$, of $u$, and perform a linear scan of $M(u)$.
	During the scan, for each $i\in[0, |M(u)-1|]$, if $M(u)[i]$ is a boundary element, it is appended to both $M(u_0)$ and $M(u_1)$.
	If $M(u)[i]$ is not a boundary element and its $l$th most significant bit is $0$, $M(u)[i]$ is appended to $M(u_0)$.
	If its most significant bit is $1$, it is appended to $M(u_1)$.
	After generating the sequences $M(u_0)$ and $M(u_1)$, we discard the sequence $M(u)$.
	We finish recursion after we create $\lceil \lg \sigma \rceil$ levels, i.e., we only examine the first $\lceil\lg \sigma\rceil$ bits of each element of $M$ to determine the tree structure. 
	Thus, this tree has $\sigma$ leaves, and the sequences associated with the leaves from left to right are named $M_0$, $M_1$, \ldots, $M_{\sigma-1}$.
	They form a partition of $M$.
	
	To speed up this process, we use a universal table $U'$ of $o(n)$ bits described in Appendix \ref{app: universal_table_lemma_B_c}, which allows us to process $M(u)$ in $O(|M(u)|\lg \sigma /\lg n+1)$ time. 
	Note that we assign $n^{\prime}$ regular and $2^l\times\frac{n^{\prime}}{\sigma}$ boundary elements to the nodes at tree level $l$.
	Thus the total time required to construct $T$ is $O(\sum_{l=0}^{\lceil\lg \sigma\rceil-1} ((n^{\prime}+2^l\times\frac{n^{\prime}}{\sigma})\lg \sigma/\lg n)+\sigma)=O(n'\lg^2 \sigma/\lg n+\sigma)$.

	To compute $B_c$ for any $0 \leq c \leq \sigma-1$,
	a crucial observation is that the $i$-th bit in $B_c$ is the same as the least significant bit of the $i$-th element of $M_c$. 
	Thus it takes $O(|B_c|(\lg \sigma+1)/\lg n+1)$ time to compute the content of $B_c$ using bit packing (see Lemma~\ref{lemma_split_1}).
        The time to represent $B_c$ using Lemma~\ref{bit_sequence} is less.
	Thus the time used to build $B_0$, $B_1$, \ldots, $B_{\sigma-1}$ from $M_0, M_1, \ldots, M_{\sigma-1}$ is $O(n'\lg \sigma/\lg n+\sigma)$, dominated by the time of building $T$. 

    We show how to efficiently construct $P_0, P_1,\ldots, P_{\sigma-1}$ in Appendix \ref{app: construct_p_seq}. See Appendix \ref{app: space_rank_prime_queryspace} for the space analysis of all data structure designed.
    Our work in this section combined with Lemmas~\ref{lemma_rank_prime_small} yields the following result. 

\begin{lemma}
	\label{rank_prime_1}
	Let $A[0..n^{\prime}-1]$ be a packed sequence drawn from alphabet $[\sigma]$, where $n' \le n$ and $\sigma =  O(2^{O(\sqrt{\lg n})})$. 
	A  data structure using $n^{\prime}\lceil\lg \sigma\rceil+ o(n^{\prime}\lg \sigma)$ extra bits can be constructed in $O(n^{\prime}{\lg^2 \sigma}/{\lg n}+\sigma)$ time to support $\prank$ in $O(1)$ time and $O(1)$ accesses to elements of $A$.
\end{lemma}

\section{Fast Construction of Data Structures for Ball Inheritance}
\label{sect:ball_inheritance}

We now solve, with fast preprocessing, ball inheritance for the special cases needed later to match the time and space bounds in parts (b) and (c) of Lemma~\ref{lemma:ball_intro}.
The omitted proofs are in Appendix \ref{app:ball_inheritance}.
One strategy is to construct the solution of Chan et al.~\cite{chan2011orthogonal} by replacing some of their components with those we designed with faster preprocessing. 
This yields: 


\begin{lemma}
	\label{ball_inheritance_large_d}
	Let $X[0, n-1]$ be a sequence drawn from alphabet $[\sigma]$ denoting the point set $N= \{(X[i], i)| 0\le i \le n-1\}$, where $2^{\sqrt{\lg n}} \le \sigma \le n$.
	A $2^{\sqrt{\lg n}}$-ary wavelet tree over $X$ 
	occupying $O(n \lg \sigma \cdot f(\sigma)+n\lg n)$ bits can be constructed in $O(n\lg\sigma/\sqrt{\lg n})$ time to support $\point$ in $O(g(\sigma))$ time and $\noderange$ in $O(\lg \lg n+ g(\sigma))$ time, where (a) $f(\sigma)=O(\lg (\lg \sigma/\sqrt{\lg n})) $ and $g(\sigma)=O(\lg (\lg \sigma/\sqrt{\lg n})) $; or (b) $f(\sigma)=O(\lg^{\epsilon} \sigma) $ and $g(\sigma)=O(1)$ for any constant $\epsilon >0 $.
	The $\noderange$ query requires a universal table of $o(n)$ bits. 
\end{lemma}

\begin{lemma}
	\label{ball_inheritance_point_small_02}
	Let $X[0..n^{\prime}-1]$ be a packed sequence drawn from alphabet $[\sigma]$ and $Y[0..n^{\prime}-1]$ be a packed sequence in which $Y[i]=i$ for each $i\in[0..n^{\prime}-1]$, where $\sigma = O(2^{O(\sqrt{\lg n})})$ and $n' = O(\sigma^{O(1)})$. 
	Given $X$ and $Y$ as input, a $d$-ary wavelet tree over $X$ using $O(n^{\prime}\lg \sigma\lg (\lg \sigma/\lg d)+\sigma w)$ bits of space can be constructed in $O(n^{\prime}{\lg^2 \sigma}/{\lg n}+\sigma\log_d \sigma)$ time to support $\point$ in $O(\lg (\lg \sigma/\lg d))$ time and $\noderange$ in $O(\lg \lg \sigma)$ time, where $d$ is a power of $2$ upper bounded by $min(\sigma, 2^{\sqrt{\lg n}})$. 
\end{lemma}

This strategy however cannot achieve, with the preprocessing time as in Lemma~\ref{ball_inheritance_point_small_02}, part (c) of Lemma~\ref{lemma:ball_intro} when the coordinates of points can be encoded in $O(\sqrt{\lg n})$ bits.
For this special case, we twist the approach of Chan et al.: they only store point coordinates explicitly at the leaf level of the wavelet tree, while we take advantage of the smaller grid size to store coordinates at more levels.
This allows us to build $\prank$ structures at fewer levels of the tree, decreasing the preprocessing time. 
The details are as follows.

Recall that, when used to represent the given point set $N$, each node, $u$, of the $d$-ary wavelet tree $T$ is conceptually associated with an ordered list, $N(u)$, of points whose $x$-coordinates are within the range represented by $u$, and these points are ordered by $y$-coordinate. 
Assume for simplicity that $\sigma$ is a power of $d$, and that both $1/\epsilon$ and $\tau = \log_d^{\epsilon} \sigma$ are integers.
We assign a color to each level of $T$:
Level $0$ is assigned color $0$, while any other Level $l$ is assigned color $\max\{c \ |\ l \text{ is divided by } \tau^c \text{ and } 0 \le c \le 1/\epsilon - 1\}$.
For each node $u$ of $T$ at a level assigned with color $1/\epsilon - 1$, we store the coordinates of the points in $N(u)$ explicitly.
For any other node $v$ (let $l$ be the level $l$ of $v$ and $c$ the color assigned to level $l$), we do not store $N(v)$.
Instead, for each $i \in [0, |N(v)|]$, we store a {\em skipping pointer} $\Sp(v)[i]$, which stores, at the closest level $l'$ satisfying $l' > l$ and $l'$ is a multiple of $\tau^{c+1}$, the descendant of $v$ at level $l'$ containing point $N(v)[i]$ in its ordered list of points. 
This descendant is encoded by its rank among all the descendants of $v$ at level $l'$ in left-to-right order. 
We use Lemma~\ref{rank_prime_1} to support $O(1)$-time $\prank$ over $\Sp(v)$. 
Then, since both $N(u)$ and $N(\Sp(u)[i])$ order points by $y$-coordinate, a $\prank(\Sp(u), i)$ query gives the position of the point $N(u)[i]$ in $N(\Sp(u)[i])$.
Thus, to compute $\point(v, i)$, we follow skip pointers starting from $v$ by performing $\prank$, until we reach a level with color $1/\epsilon - 1$, where we retrieve coordinates. 
With this we have:

\begin{lemma}
	\label{ball_inheritance_point}
	Let $X[0..n^{\prime}-1]$ be a packed sequence drawn from alphabet $[\sigma]$ and $Y[0..n^{\prime}-1]$ be a packed sequence in which $Y[i]=i$ for each $i\in[0..n^{\prime}-1]$, where $\sigma= O(2^{O(\sqrt{\lg n})})$ and $n'= O(\sigma^{O(1)})$. 
	Given $X$ and $Y$ as input, a $d$-ary wavelet tree over $X$ using $O(n^{\prime}\lg \sigma\log_d^{\epsilon} \sigma+\sigma w)$ bits for any positive constant $\epsilon$ can be constructed in $O(n^{\prime}{\lg^2 \sigma}/{\lg n}+\sigma\log_d \sigma)$ time to support $\point$ in $O(1)$ time and $\noderange$ in $O(\lg \lg \sigma)$ time, where $d$ is a power of $2$ upper bounded by $min(\sigma, 2^{\sqrt{\lg n}})$.
        The $\noderange$ query requires a universal table of $o(n)$ bits. 
\end{lemma}




\section{Optimal Orthogonal Range Reporting with Fast Preprocessing}
\label{sect: range_reporting}

We now design  structures that support orthogonal range reporting in optimal time and can be constructed fast.
We follow the solution overview given in Section~\ref{sec:introduction}: We first reduce the problem in the general case to the special case in which the points are from a $2^{\sqrt{\lg n}}\times n'$ (narrow) grid,
In this reduction, we need only support ball-inheritance over a wavelet tree with high fanout which is solved by part (b) of Lemma~\ref{ball_inheritance_large_d}.
We then further reduce it to the range reporting problem over a (small) grid of size at most $2^{\sqrt{\lg n}} \times 2^{2\sqrt{\lg n}}$, to which we apply Lemma~\ref{ball_inheritance_point} for ball inheritance.




The following two lemmas present our solution for small and then narrow girds:

\begin{lemma}
	\label{range_reporting_1}
	Let $N$ be a set of $\delta$ points with distinct $y$-coordinates in a $2^{\sqrt{\lg n}}\times \delta$ grid where $\delta \le 2^{2\sqrt{\lg n}}$.
	Given packed sequences $X$ and $Y$ respectively encoding the $x$- and $y$-coordinates of these points where $Y[i] = i$ for any $i \in [0, \delta-1]$,
	a data structure of $O(\delta\lg^{1/2+\epsilon} n+w \cdot 2^{\sqrt{\lg n}})$ bits  can be constructed in $O(\delta+\sqrt{\lg n}\cdot 2^{\sqrt{\lg n}})$ time to support orthogonal range reporting over $N$ in $O(\lg \lg n+\occ)$ time with the help of an $o(n)$-bit universal table, where $\epsilon$ is an arbitrary positive constant and $\occ$ is the number of reported points.
\end{lemma}

\begin{proof}
	We build a binary wavelet tree $T$ over $X$ augmented with support for ball inheritance.
	By Lemma~\ref{ball_inheritance_point}, $T$ occupies $O(\delta\lg^{1/2+\epsilon} n+w \cdot 2^{\sqrt{\lg n}})$ bits and can be built in $O(\delta+\sqrt{\lg n}\cdot 2^{\sqrt{\lg n}})$ time.
	It also supports $\point$ in $O(1)$ time and $\noderange$ in $O(\lg\lg n)$ time.
	For any internal node $v$ of $T$, its value array $A(v)$ is built at some point when augmenting $T$ to solve ball inheritance, though $A(v)$ may be discarded eventually.
	When $A(v)$ was available, we build a data structure $M(v)$ to support range minimum and maximum queries over $A(v)$ using Lemma~\ref{minmax_index}. 
	As $T$ has $\lceil \sqrt{\lg n} \rceil$ non-leaf levels and the total length of the value arrays of the nodes at each tree level is $\delta$,
	over all internal nodes, these structures use $O(\delta\sqrt{\lg n}\lg \lg n)$ bits in total and the overall construction time is $\sum_v O({|A(v)|}/{\sqrt{\lg n}}+1)=O(\delta+2^{\sqrt{\lg n}})$.
	These costs are subsumed in the storage and construction costs of $T$.
	Recall that $A(v)$ stores the $x$-coordinates of the set, $N(v)$, of points from $N$ whose $x$-coordinates are within the range represented by $v$, and the entries of $A(v)$ are ordered by the corresponding $y$-coordinates of these points.
	Thus any entry of $A(v)$ can be retrieved by $\point$ in constant time.
	Therefore, even after $A(v)$ is discarded, $M(v)$ can still support $\rmq$/$\rMq$ over $A(v)$ in $O(1)$ time.
	
	Given a query range $Q=[a, b]\times[c, d]$,  we first locate the lowest common ancestor $u$ of $l_{a}$ and $l_{b}$ in constant time, where $l_{a}$ and $l_{b}$ denote the $a$-th and $b$-th leftmost leaves of $T$, respectively. 
	Let $u_l$ and $u_r$ denote the left and right children of $u$, respectively, $[c_l, d_l]=\noderange(c, d, u_l)$ and $[c_r, d_r]=\noderange(c, d, u_r)$. 
	Then $Q \cap N= (([a, +\infty)\times[c_l, d_l]) \cap N(u_l)) \cup (([0, b]\times[c_r, d_r]) \cap N(u_r))$. 
	  In this way, we reduce a $2$-d $4$-sided range reporting in $N$ to $2$-d $3$-sided range reporting in $N(u_l)$ and $ N(u_r)$.
          Each $3$-sided range reporting can be solved in optimal time by performing $\rMq$ and $\rmq$ over $A(u_l)$ and $A(u_r)$ recursively as done by Chan et al.~\cite{chan2011orthogonal};
          for completeness, we provide the details in Appendix~\ref{app:rmqrecursive}.
          Thus the overall query time is $O(\lg \lg n+\occ)$. 
\end{proof}




\begin{lemma}
	\label{fat_rect}
	Let $N$ be a set of $n'$ points with distinct $y$-coordinates in a $2^{\sqrt{\lg n}}\times n^{\prime}$ grid where $n' \le n$. 
	Given packed sequences $X$ and $Y$ respectively encoding the $x$- and $y$-coordinates of these points where $y[i] = i$ for any $i \in [0, \delta-1]$,
	a data structure occupying $O(n^{\prime}\lg^{1/2+\epsilon} n+w (2^{\sqrt{\lg n}}+n'/2^{\sqrt{\lg n}}))$ bits  can be constructed in  $O(n^{\prime}+\sqrt{\lg n}\cdot 2^{\sqrt{\lg n}})$ time
	to support orthogonal range reporting over $N$ in $O(\lg \lg n+\occ)$ time with the help of an $o(n)$-bit universal table, where $\epsilon$ is an arbitrary positive constant and $\occ$ is the number of reported points.
\end{lemma}

\begin{proof}
	Let $b = 2^{2\sqrt{\lg n}}$.
	We need only consider the case in which $n' > b$ as Lemma~\ref{range_reporting_1} applies otherwise.
	Assume for simplicity that $n'$ is divisible by $b$.
	We divide $N$ into $n'/b$ subsets, and for each $i \in [0, n'/b-1]$, the $i$th subset, $N_i$, contains points in $N$ whose $y$-coordinates are in $[ib, (i+1)b-1]$.
	Let $p$ be a point in $N_i$. We call its coordinates $(p.x, p.y)$ {\em global coordinates}, while $(p.x', p.y') = (p.x, p.y \bmod b)$ its {\em local coordinates} in $N_i$; the conversion between global and local coordinates can be done in constant time.
With the local coordinates, We can apply Lemma~\ref{range_reporting_1} to construct an orthogonal range search structure over each $N_i$.
We also define a point set $\hat{N}$ in a $2^{\sqrt{\lg n}} \times n' / b$ grid. For each set $N_i$ where $i \in [0, n'/b-1]$ and each $j \in [0, 2^{\sqrt{\lg n}}-1]$, we store a point $(j,i)$ in $\hat{N}$ iff there exists at least one point in $N_i$ whose $x$-coordinate is $j$. Thus the number of points in $\hat{N}$ is at most $n'/b \times 2^{\sqrt{\lg n}} = n'/2^{\sqrt{\lg n}}$. $|\hat{N}|$ is small enough for us to build the optimal range reporting structure of Chan et al.~{\cite[Section 2]{chan2011orthogonal}},{\cite[Lemma 5]{DBLP:journals/corr/abs-1108-3683}}. 
In addition, for each $i \in [0, n'/b-1]$ and $j \in [0, 2^{\sqrt{\lg n}}-1]$, we store a list $P_{i,j}$ storing the local $y$-coordinates of the points in $N_i$ whose $x$-coordinates are equal to $j$.
	See Appendix \ref{app: space_construction_fat_rect} for the analysis of the space usage and the construction time.
	
	Given a query range $Q=[x_1, x_2]\times[y_1, y_2]$, we first check if $\lfloor y_1/b\rfloor$ is equal to $\lfloor y_2/b\rfloor$. If it is, then the points in the answer to the query reside in the same subset $N_{\lfloor y_1/b\rfloor}$,
and they can be reported in $O(\lg\lg n+\occ)$ time by Lemma~\ref{range_reporting_1}.
        Otherwise, we decompose $Q$ into three subranges $Q_1 = [x_1,x_2]\times[y_1,b(\lfloor y_1/b\rfloor+1)-1]$, $Q_2 = [x_1,x_2]\times[b(\lfloor y_1/b \rfloor+1), b\lfloor y_2/b \rfloor-1]$ and $Q_3 = [x_1,x_2]\times[b\lfloor y_2/b\rfloor, y_2]$.
	The points in $N \cap Q_1$ and $N \cap Q_3$ are in $N_{\lfloor y_1/b\rfloor}$ and $N_{\lfloor y_2/b\rfloor}$, respectively, and by Lemma~\ref{range_reporting_1}, they can be reported in optimal time. 
	The points in $N \cap Q_2$ are in $N_{\lfloor y_1/b\rfloor+1}, N_{\lfloor y_1/b\rfloor+2}, \ldots, N_{\lfloor y_2/b\rfloor-1}$. To retrieve them, we first perform an orthogonal range query in $\hat{N}$ with query range $\hat{Q} = [x_1,x_2] \times [\lfloor y_1/b\rfloor+1, \lfloor y_2/b\rfloor-1]$.
        For each point $(x,y) \in \hat{N}\cap\hat{Q}$, we observe that its existence means that is at least one point in $N_y \cap Q_2$ whose $x$-coordinates are equal to $x$; we retrieve the local $y$-coordinates of these points from $P_{y,x}$.
	The overall query time is thus $O(\lg\lg n+\occ)$.
\end{proof}




We now describe a solution in which the grid size is $\sigma \times n$ with $2^{\sqrt{\lg n}}\leq \sigma\leq n$. This is more general than our final result, but it will be needed for an application later.

\begin{lemma}
	\label{theorem_range_reporting_general}
	Given a sequence $X[0, n-1]$ drawn from alphabet $[\sigma]$ denoting the point set $N= \{(X[i], i)| 0\le i \le n-1\}$, a data structure of $O(n\lg^{1+\epsilon}  \sigma+n\lg n)$ bits for any constant $\epsilon>0$ can be constructed in $O(n{\lg \sigma}/{\sqrt{\lg n}})$ time to support orthogonal range reporting over $N$ in $O(\lg \lg n+\occ)$ time, where $2^{\sqrt{\lg n}}\leq \sigma\leq n$ and $\occ$ is the number of reported points.
\end{lemma}

\begin{proof}
	We build a $2^{\sqrt{\lg n}}$-ary wavelet tree $T$ upon $X[0, n-1]$ with support for ball inheritance using part (b) of Lemma~\ref{ball_inheritance_large_d}.
	As in the proof of Lemma~\ref{range_reporting_1}, for each internal node $v \in T$, we build a data structure $M(v)$ to support $\rMq/\rmq$ over its value array $A(v)$ in constant time using Lemma~\ref{fisher_rmq}, even though $A(v)$ is not be explicitly stored.
        Conceptually, let $N(v)$ denote an ordered list of points from $N$ whose $x$-coordinates are within the range represented by $v$, and these points are ordered by $y$-coordinate.
        Recall that $v$ is associated with sequence $S(v)$ drawn from alphabet $[2^{\sqrt{\lg n}}]$, where $S(v)[i]$ encodes the rank of the child of $v$ that contains $N(v)[i]$ in its ordered list.
	Let $\hat{S}(v)$ denote the point set $\{(S(v)[i], i)| 0\le i \le |S(v)|-1\}$, and we use  Lemma~\ref{fat_rect} to build a structure supporting orthogonal range reporting over $\hat{S}(v)$.
	See Appendix~\ref{app: space_construction_theorem_range_reporting_general} for the analysis of the space usage and the construction time.

	Given a query range $Q=[a, b]\times[c, d]$,  we first locate the lowest common ancestor $u$ of $l_{a}$ and $l_{b}$ in constant time, where $l_{a}$ and $l_{b}$ denote the $a$-th and $b$-th leftmost leaves of $T$, respectively.
	Let $u_i$ denote the $i$-th child of $u$, for any $i \in [0, 2^{\sqrt{\lg n}}-1]$.
	We first locate two children, $u_{a'}$ and $u_{b'}$, of $u$ that are ancestors of $l_a$ and $l_b$, respectively.
	They can be found in constant time by simple arithmetic as each child of $u$ represents a range of equal size.
	Then the answer, $Q\cap N$, to the query can be partitioned into three point sets $A_1=Q\cap N(v_{a'})$, $A_2=Q\cap (N(v_{a'+1})\cup N(v_{a'+2})\cup\ldots N(v_{b'-1}))$ and $A_3=Q\cap N(v_{b'})$.
        As in the proof of Lemma \ref{range_reporting_1}, $A_1$ and $A_3$ can be computed in in $O(\lg\lg n+|A_1|+|A_3|)$ time using $\noderange$ and $\rmq$/$\rMq$. 
	To compute $A_2$, observe that
        by performing range reporting over $\hat{S}$ to compute $S \cap ([a^{\prime}+1, b^{\prime}-1]\times[c_v, d_v])$, where $[c_v, d_v] = \noderange(c,d, v)$,
	we can find the set of points in $\hat{S}(v)$ corresponding to the points in $A_2$.
	For each point returned, we use $\point$ to find its original coordinates in $N$ and return it as part of $A_2$.
	This process uses $O(\lg\lg n+|A_2|)$ time.
	Hence we can compute $Q\cap N$ as $A_1\cup A_2\cup A_3$ in $O(\lg \lg n+\occ)$ time. 
\end{proof}

Our result on points over an $n\times n$ gird immediately follows.

\begin{theorem}
	\label{theorem_range_reporting_optimal}
	Given $n$ points in rank space, a data structure of $O(n\lg^{\epsilon}  n)$ words for any constant $\epsilon>0$ can be constructed in $O(n\sqrt{\lg n})$ time to support orthogonal range reporting in $O(\lg \lg n+\occ)$ time, where $\occ$ is the number of reported points.
\end{theorem}

\section{Optimal Orthogonal Range Successor with Fast Preprocessing}
\label{sect: fast_orthogonal_range_succ}

In this section, we assume that a range successor query asks for the lowest point in the query rectangle. The following theorem presents our result on fast construction of structures for optimal range successor; we provide a proof sketch, and the full proof is in Appendix~\ref{app: fast_orthogonal_range_succ}:
\begin{theorem}
  \label{theorem_range_successor_new}
  Given $n$ points in rank space, a data structure of $O(n\lg\lg n)$ words can be constructed in $O(n\sqrt{\lg n})$ time to support orthogonal range successor in $O(\lg \lg n)$ time. 
\end{theorem}
\begin{proof}[Proof (sketch)] Our approach is similar to that in Section~\ref{sect: range_reporting}, but more levels of reductions are required.  Let the sequence $X[0, n-1]$ denote the point set $N = \{(X[i], i)| 0\le i \le n-1\}$.	We build a $2^{\sqrt{\lg n}}$-ary wavelet tree $T$ upon $X[0, n-1]$ with support for ball inheritance using part (a) of Lemma~\ref{ball_inheritance_large_d}.
  As shown in the proof of Lemma~\ref{theorem_range_reporting_general}, a query can be answered by locating the lowest common ancestor, $u$, of the two leaves corresponding to the end points of the query $x$-range, and then performing two $3$-sided queries over the point sets represented by two children of $u$ and one $4$-sided query over $S(u)$.
  For the $3$-sided queries, Zhou~\cite{zhou2016two} already designed an indexing structure, which, with our $O(\lg\lg n)$-time support for $\point$ and $\noderange$, can answer a $3$-sided query in $O(\lg \lg n)$ time.
  The construction time is linear, but it is fine since $T$ has only $O(\sqrt{\lg n})$ levels. 
  The $4$-side query over $S(u)$ is a range successor query over $n'$ points in a $2^{\sqrt{\lg n}}\times n'$ (medium narrow) grid for any $n' \le n$.

  For such a medium narrow grid, we use the sampling strategy in Lemma~\ref{fat_rect} to reduce the problem to range successor over a set of $n'$ points in a $2^{\sqrt{\lg n}}\times n'$ grid where $n' \le 2\times2^{2{\sqrt{\lg n}}}-1$.
  The sampling is adjusted, as we need select at most $2^{\sqrt{\lg n}}$ sampled points from each subset.
  The grid size of $2^{\sqrt{\lg n}}\times n'$ with $n' \le 2\times2^{2{\sqrt{\lg n}}}-1$ is the same as that in Lemma~\ref{range_reporting_1}, so one may be tempted to apply the same strategy of building a binary wavelet tree to reduce it to the problem of building index structures for $3$-sided queries.
  However, we found that, to construct the structure of Zhou~\cite{zhou2016two} over $n'$ points whose coordinates are encoded in $O(\sqrt{\lg n})$ bits, $O(n'\lg\lg n /\sqrt{\lg n})$ time is required, which is a factor of $\lg\lg n$ more than the preprocessing time of the $\rmq$ structure needed in the proof of Lemma~\ref{range_reporting_1}.
  This factor comes from rank reduction in \cite{zhou2016two}, which requires us to sort packed sequences.
  To overcome this additional cost, we build a $\lg^{1/4} n$-ary wavelet tree over the $x$-coordinates, whose number of levels is a factor of $O(\lg\lg n)$ less than that of a binary wavelet tree.
  As discussed for the general case, this strategy reduces the current problem to orthogonal range successor over $n'$ points in an $\lg^{1/4} n \times n'$ (small narrow) grid with $n' \le n$.

  For a small narrow grid, there are two cases. If $n' > \lg n$, we build a binary wavelet tree of height $O(\lg\lg n)$. In the query algorithm, after finding the lowest common ancestor of the two leaves corresponding to the end points of the query $x$-range, we do not perform 3-sided queries. Instead, we traverse the two paths leading to these two leaves. This requires us to traverse down $O(\lg\lg n)$ levels, and at each level, we perform certain $\rankop$/$\selop$ operations in constant time, with the right auxiliary structures at each node.
  No extra support for ball inheritance is needed as we can simply go down the tree level by level to map information. 
  Finally, if $n' < \lg n$, we use sampling to reduce it to even smaller grids  of size at most $\lg^{1/4} n \times \lg^{3/4} n$, 
  over which a query can be answered using a table lookup.  
\end{proof}

\section{Applications}

We now apply our range search structures to the text indexing problem, in which we preprocess a text string $T \in [\sigma]^n$, where $\sigma \le n$.
Given a pattern string $P[0..p-1]$, a {\em counting query} computes the number of occurrences of $P$ in $T$ and a {\em listing query} reports these occurrences.


\subparagraph*{Text indexing and searching in sublinear time.} When both $T$ and $P$ are given in packed form, 
a text index of Munro et al.~\cite{abs-1712-07431} occupies $O(n\lg \sigma)$ bits, can be built in $O(n\lg \sigma/\sqrt{\lg n})$ time and supports counting queries in $O(p/\log_{\sigma} n +\lg n\log_{\sigma}n)$ time (there are other tradeoffs, but this is their main result).
Thus for small alphabet size which is common in practice, they achieve both $o(n)$ construction time and $o(p)$ query time, while previous results achieve at most one of these bounds.
To support listing queries, however, they need to increase space cost to $O(n\lg \sigma\lg^{\epsilon} n)$ bits and construction time to $O(n\lg \sigma \lg^{\epsilon} n)$, and then a listing query can be answered in $O({p}/{\log_{\sigma} n}+ \log_{\sigma} n\lg\lg n+\occ)$. 
The increase in storage and construction costs stems from one component they used which is an orthogonal range reporting structure over $t=O(n/r)$ points  in a $\sigma^{O(r)}\times t$ grid, for $r=c \log_{\sigma} n$ for any constant $c<1/4$.  
We can apply Lemma~\ref{theorem_range_reporting_general} over this point set to decrease the construction time of their index for listing queries to match that for counting queries:

\begin{REMOVED}
Their text index~\cite[Theorem 2]{abs-1712-07431} occupies $O((n/r)\lg n)$ bits, can be built in $O(n((\lg \lg n)^2/r+\lg \sigma/\sqrt{\lg n}))$ time and supports counting queries in $O(p/\log_{\sigma} n +r^2 \lg \sigma+r\lg \lg n )$ time, for any $0<r<(1/4)\log_{\sigma} n$.
Thus for small alphabet size which is common in practice, they achieve both $o(n)$ construction time and $o(p)$ query time, while previous results achieve at most one of these bounds but not both.
To support listing queries, however, they need to increase space cost to $O((n/r)\lg n+n\lg \sigma \lg^{\epsilon} n)$ bits and construction time to $O(n((\lg \lg n)^2/r+(r\lg^2 \sigma)/\lg^{1-\epsilon} n))$ for any constant $0<\epsilon<1$, and then a listing query can be answered in $O({p}/{\log_{\sigma} n}+r\lg \lg n+\occ)$ time, where $\occ$ is the number of occurrences of $P$ in $T$. 
The source of the increase in storage and construction costs stems from one component they used which is an orthogonal range reporting structure over $t=O(n/r)$ points  in a $\sigma^{O(r)}\times t$ grid.  
This structure occupies $O(n\lg \sigma\lg^{\epsilon} n)$ bits, can be constructed in $O(n(r\lg^2 \sigma)/\lg^{1-\epsilon} n)$ time and supports orthogonal range reporting queries in time $O(\lg \lg n+\occ)$.
While the query time and space cost of this data structure matches the current best result, we can improve the construction time when $r=c \log_{\sigma} n$ for any constant $c<1/4$.
This is because the $x$-coordinate of each point can be encoded with $\Theta(r\lg \sigma)=\Theta(\lg n)=\Omega(\sqrt{\lg n})$ bits. 
Thus, Lemma~\ref{theorem_range_reporting_general} applies for sufficiently large $n$,
and this means we can construct an orthogonal range reporting structure over these points in $O((n/r)\times r\lg\sigma/{\sqrt{\lg n}})=O(n\lg \sigma/\sqrt{\lg n})$ time, while the space cost and query time remain the same. 
As a result, we can improve their index for listing queries by decreasing construction time to match that for counting queries:
\end{REMOVED}

\begin{theorem}
	Given a packed text string $T$ of length $n$ over an alphabet of size $\sigma$, an index of $O(n\lg \sigma\lg^{\epsilon} n)$ bits can be built in $O(n\lg \sigma/\sqrt{\lg n})$ time for any positive constant $\epsilon$. 
	Given a packed pattern string $P$ of length $p$, this index supports listing queries in $O({p}/{\log_{\sigma} n}+\log_{\sigma} n\lg \lg n +\occ)$ time where $\occ$ is the number of occurrences of $P$ in $T$.
\end{theorem}

\subparagraph*{Position-restricted substring search.} In a position-restricted substring search~\cite{makinen2006position}, we are given both a pattern $P$ and two indices $0\leq l \leq r\leq n-1$, and we report all occurrences of $P$ in $T[l..r]$.
Makinen and Navarro~\cite{makinen2006position} solves this problem using an index for the original text indexing problem 
and a two-dimensional orthogonal range reporting structure. 
Different text indexes and range reporting structures yield different tradeoffs.
The tradeoff with the fastest query time supports position-restricted substring search in $O(p + \lg\lg n + \occ)$ time, where $\occ$ is the output size, and it uses $O(n\lg^{1+\epsilon} n)$ bits and can be constructed in $O(n\lg n)$ time. 
Again, the construction time of the range reporting structure is the bottleneck, which can be improved by Theorem~\ref{theorem_range_reporting_optimal}.
We can also use a new text index by Bille~et al.~\cite{bille2016deterministic} to achieve speedup when $P$ is given as a packed sequence.
We have: 

\begin{theorem}
	Given a text $T$ of length $n$ over an alphabet of size $\sigma$, an index of $O(n\lg^{1+\epsilon} n)$ bits can be built in $O(n\sqrt{\lg n})$ time for any constant $0<\epsilon<1/2$.
	Given a packed pattern string $P$ of length $p$, this index supports position-restricted substring search in $O({p}/{\log_{\sigma} n}+\lg p+\lg \lg \sigma+\occ)$ time, where $\occ$ in the size of the output.
\end{theorem}

\bibliography{rangereportingconstruct}

\begin{thebibliography}{10}

\bibitem{ah1997}
Susanne Albers and Torben Hagerup.
\newblock Improved parallel integer sorting without concurrent writing.
\newblock {\em Information and Computation}, 136(1):25--51, 1997.
\newblock \href {http://dx.doi.org/10.1006/inco.1997.2632}
  {\path{doi:10.1006/inco.1997.2632}}.

\bibitem{abr2000}
Stephen Alstrup, Gerth~St{\o}lting Brodal, and Theis Rauhe.
\newblock New data structures for orthogonal range searching.
\newblock In {\em 41st Annual Symposium on Foundations of Computer Science,
  {FOCS} 2000, 12-14 November 2000, Redondo Beach, California, {USA}}, pages
  198--207. {IEEE} Computer Society, 2000.

\bibitem{babenko2015wavelet}
Maxim Babenko, Pawe{\l} Gawrychowski, Tomasz Kociumaka, and Tatiana
  Starikovskaya.
\newblock Wavelet trees meet suffix trees.
\newblock In {\em 26th Annual ACM-SIAM Symposium on Discrete Algorithms}, pages
  572--591. Society for Industrial and Applied Mathematics, 2015.

\bibitem{belazzougui2020linear}
Djamal Belazzougui, Fabio Cunial, Juha K{\"a}rkk{\"a}inen, and Veli
  M{\"a}kinen.
\newblock Linear-time string indexing and analysis in small space.
\newblock {\em ACM Transactions on Algorithms (TALG)}, 16(2):1--54, 2020.

\bibitem{belazzougui2016range}
Djamal Belazzougui and Simon~J Puglisi.
\newblock Range predecessor and lempel-ziv parsing.
\newblock In {\em 27th Annual ACM-SIAM Symposium on Discrete Algorithms}, pages
  2053--2071. Society for Industrial and Applied Mathematics, 2016.

\bibitem{bender2004level}
Michael~A Bender and Mart{\i}n Farach-Colton.
\newblock The level ancestor problem simplified.
\newblock {\em Theoretical Computer Science}, 321(1):5--12, 2004.

\bibitem{DBLP:journals/corr/abs-1108-3683}
Philip Bille and Inge~Li G{\o}rtz.
\newblock Substring range reporting.
\newblock {\em Algorithmica}, 69(2):384--396, 2014.
\newblock \href {http://dx.doi.org/10.1007/s00453-012-9733-4}
  {\path{doi:10.1007/s00453-012-9733-4}}.

\bibitem{bille2016deterministic}
Philip Bille, Inge~Li G{\o}rtz, and Frederik~Rye Skjoldjensen.
\newblock Deterministic indexing for packed strings.
\newblock In {\em 28th Annual Symposium on Combinatorial Pattern Matching,
  {CPM} 2017, July 4-6, 2017, Warsaw, Poland}, pages 6:1--6:11, 2017.
\newblock \href {http://dx.doi.org/10.4230/LIPIcs.CPM.2017.6}
  {\path{doi:10.4230/LIPIcs.CPM.2017.6}}.

\bibitem{bhmm2009}
Prosenjit Bose, Meng He, Anil Maheshwari, and Pat Morin.
\newblock Succinct orthogonal range search structures on a grid with
  applications to text indexing.
\newblock In {\em 11th International Symposium on Algorithms and Data
  Structures}, volume 5664 of {\em Lecture Notes in Computer Science}, pages
  98--109. Springer, 2009.

\bibitem{chan2017succinct}
Timothy~M Chan, Meng He, J~Ian Munro, and Gelin Zhou.
\newblock Succinct indices for path minimum, with applications.
\newblock {\em Algorithmica}, 78(2):453--491, 2017.

\bibitem{chan2011orthogonal}
Timothy~M Chan, Kasper~Green Larsen, and Mihai P{\u{a}}tra{\c{s}}cu.
\newblock Orthogonal range searching on the ram, revisited.
\newblock In {\em 27th Symposium on Computational Geometry}, pages 1--10. ACM,
  2011.

\bibitem{chan2010counting}
Timothy~M. Chan and Mihai P\v{a}tra\c{s}cu.
\newblock Counting inversions, offline orthogonal range counting, and related
  problems.
\newblock In {\em 21st Annual {ACM-SIAM} Symposium on Discrete Algorithms,
  {SODA} 2010, Austin, Texas, USA, January 17-19, 2010}, pages 161--173, 2010.
\newblock \href {http://dx.doi.org/10.1137/1.9781611973075.15}
  {\path{doi:10.1137/1.9781611973075.15}}.

\bibitem{c1988}
Bernard Chazelle.
\newblock A functional approach to data structures and its use in
  multidimensional searching.
\newblock {\em SIAM Journal on Computing}, 17(3):427--462, 1988.

\bibitem{cikrtw2012}
Maxime Crochemore, Costas~S. Iliopoulos, Marcin Kubica, M.~Sohel Rahman, German
  Tischler, and Tomasz Walen.
\newblock Improved algorithms for the range next value problem and
  applications.
\newblock {\em Theoretical Computer Science}, 434:23--34, 2012.

\bibitem{ckwir2010}
Maxime Crochemore, Marcin Kubica, Tomasz Walen, Costas~S. Iliopoulos, and
  M.~Sohel Rahman.
\newblock Finding patterns in given intervals.
\newblock {\em Fundamenta Informaticae}, 101(3):173--186, 2010.
\newblock \href {http://dx.doi.org/10.3233/FI-2010-283}
  {\path{doi:10.3233/FI-2010-283}}.

\bibitem{fischer2011space}
Johannes Fischer and Volker Heun.
\newblock Space-efficient preprocessing schemes for range minimum queries on
  static arrays.
\newblock {\em SIAM Journal on Computing}, 40(2):465--492, 2011.

\bibitem{FredmanW94}
Michael~L. Fredman and Dan~E. Willard.
\newblock Trans-dichotomous algorithms for minimum spanning trees and shortest
  paths.
\newblock {\em Journal of Computer and System Sciences}, 48(3):533--551, 1994.
\newblock \href {http://dx.doi.org/10.1016/S0022-0000(05)80064-9}
  {\path{doi:10.1016/S0022-0000(05)80064-9}}.

\bibitem{grossi2009more}
Roberto Grossi, Alessio Orlandi, Rajeev Raman, and S.~Srinivasa Rao.
\newblock More haste, less waste: Lowering the redundancy in fully indexable
  dictionaries.
\newblock In {\em 26th International Symposium on Theoretical Aspects of
  Computer Science, {STACS} 2009, February 26-28, 2009, Freiburg, Germany,
  Proceedings}, pages 517--528, 2009.
\newblock \href {http://dx.doi.org/10.4230/LIPIcs.STACS.2009.1847}
  {\path{doi:10.4230/LIPIcs.STACS.2009.1847}}.

\bibitem{jms2004}
Joseph J{\'{a}}J{\'{a}}, Christian~Worm Mortensen, and Qingmin Shi.
\newblock Space-efficient and fast algorithms for multidimensional dominance
  reporting and counting.
\newblock In {\em 15th International Symposium on Algorithms and Computation},
  volume 3341 of {\em Lecture Notes in Computer Science}, pages 558--568.
  Springer, 2004.

\bibitem{jls2017}
Jesper Jansson, Zhaoxian Li, and Wing{-}Kin Sung.
\newblock On finding the adams consensus tree.
\newblock {\em Information and Computation}, 256:334--347, 2017.
\newblock \href {http://dx.doi.org/10.1016/j.ic.2017.08.002}
  {\path{doi:10.1016/j.ic.2017.08.002}}.

\bibitem{kkl2007}
Orgad Keller, Tsvi Kopelowitz, and Moshe Lewenstein.
\newblock Range non-overlapping indexing and successive list indexing.
\newblock In {\em 10th Workshop on Algorithms and Data Structures 2007,
  Halifax, Canada, August 15-17, 2007, Proceedings}, volume 4619 of {\em
  Lecture Notes in Computer Science}, pages 625--636. Springer, 2007.

\bibitem{ls1994}
Hans{-}Peter Lenhof and Michiel H.~M. Smid.
\newblock Using persistent data structures for adding range restrictions to
  searching problems.
\newblock {\em Informatique Theorique et Applications}, 28(1):25--49, 1994.
\newblock \href {http://dx.doi.org/10.1051/ita/1994280100251}
  {\path{doi:10.1051/ita/1994280100251}}.

\bibitem{Lewenstein13}
Moshe Lewenstein.
\newblock Orthogonal range searching for text indexing.
\newblock In {\em Space-Efficient Data Structures, Streams, and Algorithms -
  Papers in Honor of J. Ian Munro on the Occasion of His 66th Birthday}, pages
  267--302, 2013.
\newblock URL: \url{https://doi.org/10.1007/978-3-642-40273-9\_18}, \href
  {http://dx.doi.org/10.1007/978-3-642-40273-9\_18}
  {\path{doi:10.1007/978-3-642-40273-9\_18}}.

\bibitem{makinen2006position}
Veli M{\"a}kinen and Gonzalo Navarro.
\newblock Position-restricted substring searching.
\newblock In {\em 7th Latin American Symposium on Theoretical Informatics},
  pages 703--714. Springer, 2006.

\bibitem{FerraginaMMN07}
Veli M{\"{a}}kinen and Gonzalo Navarro.
\newblock Rank and select revisited and extended.
\newblock {\em Theoretical Computer Science}, 387(3):332--347, 2007.
\newblock \href {http://dx.doi.org/10.1016/j.tcs.2007.07.013}
  {\path{doi:10.1016/j.tcs.2007.07.013}}.

\bibitem{abs-1712-07431}
J.~Ian Munro, Gonzalo Navarro, and Yakov Nekrich.
\newblock Text indexing and searching in sublinear time.
\newblock In {\em 31th Annual Symposium on Combinatorial Pattern Matching,
  {CPM} 2020, July 17-19, 2020, Copenhagen, Denmark. To appear.}, 2020.

\bibitem{munro2016fast}
J~Ian Munro, Yakov Nekrich, and Jeffrey~S Vitter.
\newblock Fast construction of wavelet trees.
\newblock {\em Theoretical Computer Science}, 638:91--97, 2016.

\bibitem{nekrich2012sorted}
Yakov Nekrich and Gonzalo Navarro.
\newblock Sorted range reporting.
\newblock In {\em 13th Scandinavian Symposium and Workshops, Helsinki, Finland,
  July 4-6, 2012. Proceedings}, pages 271--282, 2012.
\newblock \href {http://dx.doi.org/10.1007/978-3-642-31155-0\_24}
  {\path{doi:10.1007/978-3-642-31155-0\_24}}.

\bibitem{pt2006}
Mihai Patrascu and Mikkel Thorup.
\newblock Time-space trade-offs for predecessor search.
\newblock In {\em 38th Annual {ACM} Symposium on Theory of Computing, Seattle,
  WA, USA, May 21-23, 2006}, pages 232--240. {ACM}, 2006.

\bibitem{willard1985new}
Dan~E. Willard.
\newblock New data structures for orthogonal range queries.
\newblock {\em SIAM Journal on Computing}, 14(1):232--253, 1985.
\newblock \href {http://dx.doi.org/10.1137/0214019}
  {\path{doi:10.1137/0214019}}.

\bibitem{yhw2011}
Chih{-}Chiang Yu, Wing{-}Kai Hon, and Biing{-}Feng Wang.
\newblock Improved data structures for the orthogonal range successor problem.
\newblock {\em Computational Geometry}, 44(3):148--159, 2011.
\newblock \href {http://dx.doi.org/10.1016/j.comgeo.2010.09.001}
  {\path{doi:10.1016/j.comgeo.2010.09.001}}.

\bibitem{zhou2016two}
Gelin Zhou.
\newblock Two-dimensional range successor in optimal time and almost linear
  space.
\newblock {\em Information Processing Letters}, 116(2):171--174, 2016.
\newblock \href {http://dx.doi.org/10.1016/j.ipl.2015.09.002}
  {\path{doi:10.1016/j.ipl.2015.09.002}}.

\end{thebibliography}

\newpage

\appendix

\section{Proofs Omitted From Section \ref{sect:preliminaries}}
\label{app: preliminary}

\subsection{Proof of Lemma \ref{lemma:wavelet_construct_d_packed}}

Through out the paper, we define $C(s..f)$ to be the bits between and including the $s$- and $f$-th most significant bits of $C$, where $C$ can be an integer or a character. 
Our proof requires the following lemma:

%
\begin{lemma}
	\label{lemma_split_1}
	Let $C[0..n^{\prime}-1]$ be a packed sequence of $c$-bit elements, where $n' \le n$.  Given a pair of parameters $s$ and $f$ satisfying that $0\leq s \leq f \leq c-1$, a packed sequence $A[0..n^{\prime}-1]$ of $(f-s+1)$-bit elements in which $A[i]=C[i](s..f)$ for each entry $i\in[0..n^{\prime}-1]$ can be constructed in $O(n^{\prime}{c}/{\lg n}+1)$ time. 
\end{lemma}
\begin{proof}
	Let $\delta$ denote the block size $\lfloor \frac{\lg n}{2\times c} \rfloor$. 
	We construct a universal lookup table $U$. 
	For each possible $\delta$-element packed sequence $S_1$ drawn from alphabet $[2^c]$, and each different range $[s, f]$ where $0\leq s \leq f \leq c-1$, $U$ stores a packed sequence $S_2$ in which $S_2[i]=S_1[i](s..f)$ for each $i\in[0..\delta-1]$. 
	As there are $2^{\delta\times c}\times c^2=\sqrt{n}\times c^2$ entries in $U$, and each entry stores a result of $\delta\times(f-s+1)$ bits, table $U$ occupies $O(\sqrt{n}\times c^2 \times \delta\times(f-s+1))=o(n)$ bits of space. 
	Given a pair of parameters $s$ and $f$, we can apply table $U$ to extract the bits from $\delta$ consecutive elements $C[i], C[i+1], \ldots, C[i+\delta-1]$  for each $i\in [0, n^{\prime}-1-\delta]$ in constant time. 
	Therefore, the overall processing time is $O(n^{\prime}{c}/{\lg n}+1)$.
\end{proof}

With this lemma, we now prove Lemma \ref{lemma:wavelet_construct_d_packed}.

\begin{proof}
	We only prove the result when value and index arrays are required; the other results in the lemma follow by removing the steps of constructing them. 
	The construction consists of two steps: we first build a binary wavelet tree $T_2$ and then convert it to a $d$-ary wavelet tree $T_d$.
	
	To construct $T_2$, let $r$ denote its root node, and we have $A(r) = A$ and $I(r) = I$. 
	We then create the left child, $r_0$, and the right child, $r_1$, of $r$, and perform a linear scan of $A(r)$ and $I(r)$. 
	During the scan, for each $i\in[0, |A(r)|-1]$, if the highest bit of $A(r)[i]$ is $0$, then $A(r)[i]$ is appended to $A(r_0)$, and $I(r)[i]$ is appended to $I(r_0)$. Otherwise, they are appended to $A(r_1)$ and $I(r_1)$.
	Afterwards,  we recursively process the child node $r_0$ and $r_1$ in the same manner, but we examine the second highest bit of each element of $A(r_0)$ and $A(r_1)$. 
	In general, when generating the sequences for the child nodes of an internal node $u$ at tree level $l$ where $l\in[0, \lg \sigma-1]$, we append $A(u)[i]$ and $I(u)[i]$ to $A(r_0)$ and $I(r_0)$, respectively, if the $l$-th highest bit of $A(u)[i]$ is 0. 
	Otherwise. they are appended to $A(r_1)$ and $I(r_1)$. 
	If $A(v)$ for some node $v$ is empty but $v$ is above the leaf level, then we keep $v$ as an empty node, and at next phase we create empty children $v_0$ and $v_1$ under $v$. 
	$T_2$ have been constructed completely after processing all the $\lceil \lg \sigma \rceil$ bits of each element of $A$. 
	
	To speed up this process, we use a universal table $U$.
	Let $b = \lfloor \frac{\lg n}{2 t} \rfloor$, where $t = \lceil\lg n^{\prime} \rceil+\lceil \lg \sigma\rceil$. 
	This table $U$ has an entry for each possible triple $(D, E, c)$, where $D$ is a sequence of length $b$ drawn from universe $[\sigma]$, $E$ is a sequence of length $b$ drawn from universe $[n^{\prime}]$, and $c$ is an integer in $[0, \lceil \lg \sigma \rceil-1]$.
	This entry $U[D, E, c]$ stores four packed sequences $D_0$, $D_1$, $E_0$ and $E_1$ defined as follows:
	$D_0[i]$ or $D_1[i]$ stores the $i$th element in $D$ whose $c$-th most significant bit is $0$ or $1$, respectively,
	while $E_0[i]$ or $E_1[i]$ stores $E[j]$ if $D[j]$ is the $i$th element in $A$ whose $c$-th most significant bit is $0$ or $1$, respectively. 
	Similar to the table $U$ in the proof of Lemma \ref{lemma_split_1}, $U$ uses $o(n)$ bits.
	By performing table lookups with $U$, we can process $b$ consecutive elements in $A(u)$ and $I(u)$ in constant time, 
	and hence we spend $O(|A(u)| / b + 1)$ time on each internal node $u$.
	The sum of the lengths of all the value arrays for the nodes at the same level of $T_2$ is $n'$. 
	As $T_2$ has $\lceil \lg \sigma \rceil+1$ levels and $O(\sigma)$ nodes, the total time required to construct $T_2$ is $O(n^{\prime}\lg \sigma(\lg n^{\prime}+\lg \sigma)/{\lg n}+ \sigma)$. 
	
	We then transform $T_2$ into a $d$-ary tree $T_d$. 
	For simplicity, assume that $\sigma$ is a power of $d$. 
	We first remove the nodes of $T_2$ whose levels are not multiples of $\lg d$, and add edges between each remaining node $u$ and its descendants at the next remaining level. 
	We then visit each internal node $u$ of $T_d$ and associate it with a packed sequence $S(u)$ storing $A(u)[i](l\lg d..(l+1)\lg d-1)$ for all $i\in[0, |A(u)|-1]$, where $l$ is the level of $u$ in $T_d$.
	It remains to analyze the time needed to transform $T_2$ into $T_d$.
	For each internal node $u \in T_d$, it takes $O(|A(u)|\lg \sigma/\lg n+1)$ time to construct $S(u)$.
	The sum of the lengths of all the value arrays for the nodes at the same level of $T_d$ is $n'$. 
	As $T_d$ has $\lg \sigma/\lg d+1$ levels and $O(\sigma)$ nodes, the time required to construct $S(u)$'s for all the internal node of $T_d$ is $O(n'\lg^2 \sigma/(\lg n\times \lg d)+\sigma)$.
	Therefore, the two steps of our construction algorithm use $O(n^{\prime}\lg \sigma(\lg n^{\prime}+\lg \sigma)/{\lg n}+ \sigma)$ time in total.
\end{proof}

\subsection{Proof of Lemma \ref{lemma:wavelet_construct_d_unpacked}}
\begin{proof} 
	For simplicity, assume that $\sigma$ is a power of $d$. 
	We use $O(n)$ time to create a sequence $I[0..n-1]$ in which $I[i]=i$ for each $i\in[0..n-1]$.
	At the root node $r$ of the wavelet tree, set $A(r) = A$ and $I(r) = I$. We then create an empty sequence $S(r)$, and, for each $i\in[0, n'-1]$, we append $A(r)[i](0..\lg d-1)$ to $S(r)$. 
	At the second level, there are $d$ children of $r$. 
	We linearly scan $A(r)$, $I(r)$ and $S(r)$, appending $A(r)[i]$ or $I(r)[i]$ to $A(r_{\alpha})$ or $I(r_{\alpha})$, respectively, where $\alpha=S(r)[i]$ and $r_{\alpha}$ represents the $\alpha$-th child node of $r$. 
	Next, we construct $S(v)$ for each node $v$ at the second tree level by appending $A(v)[i](\lg d..2\lg d-1)$ to $S(v)[i]$ for each $i\in[0, |A(v)|-1]$.
	
	This process continues at each successive level: in general, when generating $S(u)$ for a node $u$ at a level $\ell$ where $\ell\in[0, \frac{\lg \sigma}{\lg d}-1]$, we append $A(u)[i](\ell\times\lg d..(\ell+1)\times\lg d-1)$ to $S(u)[i]$ for each $i\in[0, |A(u)|-1]$. If $u$ is an internal node, we append $A(u)[i]$ or $I(u)[i]$ to the sequence $A(u_\alpha)$ or $I(u_\alpha)$, respectively, where $\alpha=S(u)[i]$. 
	After reaching the leaf level, $\frac{\lg \sigma}{\lg d}+1$ levels have been created on $T_d$. 
	As it uses $O(n)$ time for the non-empty nodes at each tree level and there are in total at most $O(\sigma)$ empty nodes, overall the construction time is $O(n\times{\lg \sigma}/{\lg d}+\sigma)=O(n{\lg \sigma}/{\lg d})$, as $\sigma\leq n$. 
\end{proof}

\subsection{Proof of Lemma \ref{lemma_rank_prime_small}}

Our proof requires the $\countop_c(A, j)$ operation which computes the number of elements less than or equal to $c$ in $A[0..j]$.
Clearly the support for $\countop$ implies that for $\rankop$.
Our proof alos requires the following lemma:
\begin{lemma}[{\cite[Lemma 2.3]{babenko2015wavelet}}]
	\label{lemma_rank_small}
	Let $A[0..n^{\prime}-1]$ be a packed sequence drawn from alphabet $[\sigma]$, where $n' \le n$ and $\sigma < \lg^{1/3} n$. A systematic data structure occupying $o(n^{\prime})$ extra bits supporting $\countop$ in $O(1)$ time can be constructed in $O(n^{\prime}\lg \sigma/\lg n)$ time. 
\end{lemma}

With Lemma \ref{lemma_rank_small}, we now prove Lemma \ref{lemma_rank_prime_small}.

\begin{proof}
	Lemma \ref{lemma_rank_small} already subsumes this lemma when $\sigma\leq 2^{\lceil 1/4\lg \lg n \rceil}$, so it suffices to assume that $2^{\lceil 1/4\lg \lg n \rceil}<\sigma$ in the rest of the proof. 
	
	Let $d = 2^{\lceil 1/4\lg \lg n \rceil}$.
	We first build a $d$-ary wavelet tree $T$ in $O(n^{\prime}\lg^2 \sigma/{\lg n}+\sigma)$ time by Lemma~\ref{lemma:wavelet_construct_d_packed}. 
	Then the height of the tree is $h=O(\lg \sigma/(1/4\lg \lg n)) = O(1)$.
	For each level $l$ of $T$ except the leaf level, we construct a packed sequence $S_l$ by concatenating all the $S(v)$'s for the nodes at this level from left to right.
	As $S_l$ is drawn from alphabet $[d]$ and there are at most $\sigma$ nodes at each level, it takes $O(|S_l|\lg d / \lg n + \sigma) = O(n'\lg\sigma / \lg n + \sigma)$ time to construct $S_l$.
	We then build a data structure $C_l$ with constant-time support for $\countop$ over $S_l$;
	by Lemma~\ref{lemma_rank_small}, this data structure occupies $o(n^{\prime}\lg \sigma)$ extra bits, and it can be constructed in $O(n^{\prime}\lg d/\lg n)=O(n'\lg \sigma/\lg n)$ time.
	At last, we discard all sequences $S(v)$ and the tree $T$ to save space.
	As $T$ has a constant number of levels, all the $S_l$'s and $C_l$'s occupy $n^{\prime}\lceil \lg \sigma \rceil+o(n^{\prime}\lg \sigma)$ bits in total, and their construction time, including that of $T$ which we discard later, is $O(n^{\prime}\lg^2 \sigma/{\lg n}+\sigma)$.
	The set of $S_l$'s and $C_l$'s is the data structures Bose et al.~{\cite[Theorem 4]{bhmm2009}} designed, which support $\countop$ operations in constant time over a sequence drawn from an alphabet of size $O(\polylog(n))$. This implies the support for $\rankop$.
\end{proof}

\subsection{Proof of Lemma \ref{rank_select_0}}

\begin{proof}
	Belazzougui et al.~\cite[Lemma 3.5]{belazzougui2020linear} already proved this lemma for the case in which $\sigma \leq n$. When $\sigma > n$, then the data structure we construct is simply an array $A'[0..n-1]$ storing all the answers, i.e., $A'[i] = \prank(A, i)$ for any $i \in [0, n-1]$. $A'$ occupies $O(n\lg n) = O(n\lg\sigma)$ bits. To construct $A'$, it is enough to perform a linear scan of $A$, and during the scan, we maintain an array $C[0..\sigma-1]$ in which $C[j]$, for any $j \in [0, \sigma-1]$, stores how many times symbol $j$ occurs in the portion of $A$ that we have scanned so far. This uses $O(n+\sigma)$ time.
\end{proof}

\subsection{Proof of Lemma \ref{minmax_index}} 

Belazzougui and Puglisi~\cite{belazzougui2016range} provided a systematic scheme with efficient construction for range minimum/maximum queries over an input sequence from small alphabets.
Our proof requires their result presented as follows.
\begin{lemma}[{\cite[Lemma D.1]{belazzougui2016range}}]
	\label{minimax_bela}
	Let $A[0..n^{\prime}-1]$ be a packed sequence drawn from alphabet $[\sigma]$, where $n' \le n$ and $\sigma\leq 2^{\sqrt{\lg n}}$. 
	There is a systematic data structure using $O(n^{\prime}{\lg \sigma}/{\lg n})$ extra bits constructed in $O(n^{\prime}{\lg \sigma}/{\lg n})$ time, which answers $\rmq(i, j)/\rMq(i, j)$ queries in constant time.  
	The query procedures each uses a universal table of $o(n)$ bits.
\end{lemma}

With Lemma \ref{minimax_bela}, we now prove Lemma \ref{minmax_index}:

\begin{proof}
	If $\sigma\leq \lg n$, each element in $A$ can be encoded with $O(\lg \lg n)$ bits, so explicitly storing elements of $A$ requires $O(n^{\prime}\lg \lg n)$ bits, which is affordable. 
	We then apply Lemma \ref{minimax_bela} for $\rmq/\rMq$ queries over $A$, which achieves the efficient construction time and the constant query time.
	Therefore, for the rest of the proof, it suffices to assume $\sigma > \lg n$.
	
	We only show the proof for range minimum as the support for range maximum is similar. 
	Let $b$ denote the block size $\lfloor {\lg n}/{(2\lg \lceil \sigma \rceil)}  \rfloor$. 
	The elements of $A$ are conceptually divided into blocks of $b$ elements each. 
	With a universal lookup table $U$, we can retrieve the minimum value of a block of elements in constant time. 
	For each possible $b$ elements drawn from $[\sigma]$, $U$ stores the minimum element value of these $b$ elements. 
	Similar to the table $U$ in the proof of Lemma \ref{lemma_split_1}, $U$ uses $o(n)$ bits.
	
	Next,  we store the minimum values of the blocks in a sequence $A^{\prime}$ ordered by their original position in $A$. 
	The sequence $A^{\prime}$ occupies $O(n^{\prime}/b\times \lg \sigma)=O(n^{\prime})$ bits. Over $A^{\prime}$ we build a data structure $DS_1$ of $O({n^{\prime}}/{b})$ bits in $O({n^{\prime}}/{b})$ time by Lemma \ref{fisher_rmq}. 
	
	To save storage, we do not keep the original element values in a block. Instead, each element value $e$ is replaced with its rank, i.e., the number of elements in the block that are smaller than $e$. 
	As each block has $b$ elements, the rank value can be encoded with $O(\lg b)=O(\lg \lg n)$ bits. 
	The transformation from element values to their corresponding rank values for each block can be processed in constant time by applying a universal lookup table $U'$. 
	For each possible $b$ elements drawn from $\{0,1,..., \sigma-1\}$, we store the rank values in $O(b\lg \lg n)$ bits. 
	Similar to the table $U$ in the proof of Lemma \ref{lemma_split_1}, $U'$ uses $o(n)$ bits.
	With table $U'$, we spend $O(\frac{n^{\prime}}{b})$ time on transferring $n'$ elements into their ranks. 
	
	At last, we construct a universal lookup table $U''$ in which for each possible $b$ elements drawn from $[b]$ and for each different query range $[q_1, q_2]$ where $0\leq q_1 \leq q_2\leq b-1$, we store the in-block index of the minimum value in the range $[q_1, q_2]$. 
	Similar to the table $U$ in the proof of Lemma \ref{lemma_split_1}, $U''$ uses $o(n)$ bits.
	The universal table $U''$ will be used amid the querying procedure only. 
	
	All rank values occupy $O(n^{\prime}\lg \lg n)$ bits. In addition to the $O(n^{\prime}+\frac{n^{\prime}}{b})$-bit space usage of $A^{\prime}$ and $DS_1$, the overall space cost is $O(n^{\prime}\lg \lg n)$ bits. 
	As shown above, computing the rank value for all $n^{\prime}$ elements uses $O(\frac{n^{\prime}}{b})$ time. 
	Constructing the sequence $A^{\prime}$ and building the data structure $DS_1$ over $A^{\prime}$ takes $O(\frac{n^{\prime}}{b})$ time. 
	Overall the construction time is bounded by $O({n^{\prime}}/{b})=O(n'\lg \sigma/\lg n)$.
	
	Now we show how to answer the range minimum query given a query range $[i, j]$. Let $B_s$ and $B_t$ denote the block containing $i$ and $j$, respectively, where $s=\lfloor \frac{i}{b} \rfloor$ and $t=\lfloor \frac{j}{b} \rfloor$.  We only consider the case when $s < t$; the remaining case in which $s$ is equal to $t$ can be handled similarly. Let $m_1$ denote the minimum value in $B_s[i \ mod \ b, b-1] $, $m_2$ denote the minimum value among blocks $B_{s+1}$, $B_{s+2}$,\ldots, $B_{t-1}$, and $m_3$ denote the minimum value in $B_t[0, j \ mod \ b]$. The answer is clearly $\min(m_1, m_2, m_3)$. 
	
	The value $m_2$ can be retrieved in constant time as follows: We search the data structure $DS_1$ for the index $\tau$ of the minimum value in the query range $[s+1, t-1]$ and complete with accessing $A^{\prime}[\tau]$.
	For the values $m_1$ and $m_3$ both can be answered in a similar way, and we take how to retrieve $m_1$ as an example. Given the pattern of block $B_s$ and the query range $[i \ mod \ b, b-1]$, we can apply $U''$ to retrieve the in-block index $\tau$ of the minimum value $e$ in constant time,  compute the original index $\tau^{\prime}$ of $e$ in $A$, where $\tau^{\prime}=b\times s+\tau$, and retrieve $m_1$ by accessing $A[\tau^{\prime}]$. Overall, the query requires $O(1)$ time and $O(1)$ accesses to the elements of $A$.
\end{proof}

\section{Fast Construction of Predecessor Query Structures}
\label{sect:pred_succ}
Let $A[0..n-1]$ be a sequence of integers sorted in the increasing order. Given a query integer $x$, we define operations $\pred(x)$ and $\succ(x)$:
\[\pred(x)=\max\{j\ |\ A[j] \leq x, 0\leq j \leq n-1\}\]
\[\succ(x)=\min\{j\ |\ x \leq A[j],0\leq j \leq n-1\}\]

Belazzougui et al.~\cite{belazzougui2020linear} show a data structure with deterministic linear preprocessing time for predecessor/successor queries. Their result shown as follows will be used later in our methods.
\begin{lemma}[{\cite[Lemma 3.6]{belazzougui2020linear}}]
	\label{y_fast}
	Given a sorted sequence $A$ of integers from universe $ [0, u-1]$, a data structure of $O(n\lg u)$ bits can be constructed in linear time, which answers a $\pred$ or $\succ$ query in $O(\lg \lg u)$ time. 
\end{lemma}

We also need a solution under the indexing model over packed sequences. The following result can be achieved by combining an approach of Grossi et al.~\cite{grossi2009more} with Lemma~\ref{y_fast}, while applying universal tables.

\begin{lemma}
	\label{pre_succ_index1}
	Given a sorted packed sequence $A$ of $n^{\prime}$ distinct integers from $[\sigma]$, where $n^{\prime}\leq n $ and $\sigma\leq 2^{c\sqrt{\lg n}}$ for any arbitrary positive constant $c$, a data structure using $O(n^{\prime}\lg \lg \sigma)$ extra bits of space can be constructed in $O({n^{\prime}}/{\sqrt{\lg n}})$ time, which answers a $\pred(x)$ or $\succ(x)$ query in $O(\lg \lg \sigma)$ time and $O(1)$ accesses to the elements of $A$. 
	The construction and query procedures each requires access to a universal table of size $o(n)$ bits.
\end{lemma}

\begin{proof}
	Let $b=\lfloor \sqrt{\lg n}/(2c) \rfloor$. We divide $A$ into blocks of length $b$ each.
	We retrieve the last element of each block and construct a predecessor/successor data structure $R$ over these elements using Lemma~\ref{y_fast}.
	$R$ uses $O(({n^{\prime}}/{b})\times \lg \sigma)=O(n^{\prime})$ bits of space and can be constructed in $O({n^{\prime}}/{b})$ time.
	Then, over each block, we regard each integer in it as a binary string of length $\lg \sigma$ and construct a Patricia trie over the integers in this block, as done by Grossi et al.~\cite[Lemma 3.3]{grossi2009more}.
	This is a compressed bitwise trie with a skip value of $O(\lg\lg \sigma)$ bits stored at each node.
	It stores elements at the leaves in sorted order. 
	As the trie has $b$ leaves and $b-1$ internal nodes, its tree structure can be encoded in $O(b)$ bits.
	With skip values, each trie can be encoded in $O(b\lg\lg\sigma)$ bits, without encoding the $b$ elements at its leaves. 
	To construct such a trie fast, we use a universal $U$ which has an entry for any possible packed sequence $S$ of $b$ elements drawn from $[\sigma]$.
	This entry stores the Patricia trie (without the  elements stored at leaves) of $O(b\lg \lg \sigma)$ bits constructed upon $S$. 
	As there are at most $2^{(\lg \sigma)\times b}\le \sqrt{n}$ different entries in $U$, $U$ uses $O(\sqrt{ n} \times b\lg \lg \sigma)=o(n)$ bits. 
	With $U$, a Patricia trie over any block of $b$ elements can be constructed in constant time. 
	As there are $\lceil {n^{\prime}}/{b} \rceil$ blocks in total, the overall space usage of the tries and $R$ is $O(({n^{\prime}}/{b}) \times b \lg \lg \sigma + n^{\prime})=O(n^{\prime}\lg \lg \sigma)$ bits, and  
	and the overall processing time is $O({n^{\prime}}/{\sqrt{\lg n}})$.
	
	To answer a query, given an integer $y$, where $y \in [0, \sigma-1]$, we first perform a predecessor or successor query over $R$ to find the block $B$ containing $\pred(y)$ or $\succ(y)$ in $O(\lg \lg \sigma)$ time. 
	Then, with the help of an $o(n)$-bit universal table, we query over the trie built upon $B$ using the query algorithm by Grossi et al.~\cite{grossi2009more} in $O(1)$ time and $O(1)$ accesses of elements of $A$ to retrieve $\pred(y)$ or $\succ(y)$.
	Therefore, the overall query cost is $O(\lg \lg \sigma)$ time and $O(1)$ accesses to elements of $A$.
\end{proof}

We then extend Lemma \ref{pre_succ_index1} for a sequence of integers that are not necessarily distinct. 

\begin{lemma}
	\label{pre_succ_index_duplicates}
	Given a sorted packed sequence $A$ of $n^{\prime}$ integers from $[\sigma]$, where $n^{\prime}\leq \sigma\leq 2^{c\sqrt{\lg n}}$ for any arbitrary positive constant $c$, a data structure using $O(n^{\prime}\lg \lg \sigma)$ extra bits of space can be constructed in $O({n^{\prime}}/{\sqrt{\lg n}})$ time, which answers a $\pred(x)$ or $\succ(x)$ query in $O(\lg \lg \sigma)$ time and $O(1)$ accesses to the elements of $A$. 
	The construction and query procedures each requires access to a universal table of size $o(n)$ bits.
\end{lemma}

\begin{proof}
	We create a bitvector $B[0..n']$ in which $B[i]=1$ if $i=0$ or $A[i]>A[i-1]$, and represent it by Lemma \ref{bit_sequence} to support $\rankop$ and $\selop$.
	Thus $B$ records the position of the first occurrence of each distinct integer in $A$.
	We also define a sequence $A'[0..t-1]$, in which $A'[i] = A[\selop_1(B, i)]$, where $t$ is the number of distinct elements in $A$.
	$A'[0..t-1]$ then stores the distinct elements of $A$.
	We construct a predecessor/successor structure over $A'$ using Lemma \ref{pre_succ_index1}.
	$A'$ is needed during construction but is discarded at the end of preprocessing, as each element, $A'[i]$, can be accessed by retrieving $A[\selop_1(B, i)]$.
	By Lemmas~\ref{bit_sequence} and \ref{pre_succ_index1}, $B$ and the predecessor/successor structure over $A'$ occupy $O(n'\lg \lg \sigma)$ bits in total. 
	
	Next, we give a query algorithm for $\succ(x)$ over $A$;  the support for $\pred(x)$ is similar.
	We perform a successor query over $A'$ to retrieve the successor, $x'$, in $A'$, which uses $O(\lg \lg \sigma)$ time and $O(1)$ accesses of elements of $A$.
	This will also give the position, $i$, of $x'$ in $A'$. 
	Then, we find $\succ(x)$ over $A$ in constant time, which is $\selop_1(B, i)$. 
	Overall, we require $O(\lg \lg \sigma)$ time and $O(1)$ accesses of elements of $A$.
	
	To build these data structures, we first show how to compute $B$ and $A'$. 
	Let $b=\lfloor \sqrt{\lg n}/(2c) \rfloor$.
	We use a universal table $U$ to generate $b$ elements of $B$ and $A'$ in constant time.
	$U$ has an entry for each possible packed sequence $S[0..b-1]$ drawn from $[\sigma]$ and each possible flag $f \in\{0,1\}$, which stores a bitvector $V[0..b-1]$ and  a packed sequence $S'$ of length at most $b$ defined as follows:
	If $f = 1$, we set $V[0] = 1$ and $V[0] = 0$ otherwise.
	For each $i\in[1, b-1]$, if $S[i-1] = S[i]$, then $V[i]$ is set to 0. 
	Otherwise, $V[i]$ is set to 1.
	Then, the length of $S'$ is equal to the number of $1$'s in $V$, and $S'[i] = S[\rankop_1(V, i)]$. 
	As $U$ has $O(n^{1/2})$ entries each using $O(\polylog(n))$ bits, $U$ occupies $o(n)$ bits. 
	With $U$, we can generate $b$ bits in $B$ and $b$ entries of $A'$ in one table lookup, and thus the content of $B$ and $A'$ can be computed in $O(n/b)$ time.
	Adding the construction time needed to build query structures over them, the overall preprocessing time is $O(n/b) = O({n^{\prime}}/{\sqrt{\lg n}})$. 
\end{proof}

For general integer sequences, we will use the $\pred(x)/\succ(x)$ data structure by Chan et al.~\cite{chan2011orthogonal}, which is summarized in the following lemma (even though Chan et al. did not analyze the construction time, it follows directly from previous results on the data structure components used):

\begin{lemma}[{\cite[Section 2]{chan2011orthogonal}}]
	\label{pre_succ_chan}
	Given an increasingly sorted sequence $A$ of $n^{\prime}$ distinct integers from universe $[n]$ where $n'\le n$, a data structure using extra $O(n^{\prime}\lg \lg n)$ bits of space can be constructed in linear time, which answers a $\pred(x)$ or $\succ(x)$ query in $O(\lg \lg n)$ time and $O(1)$ accesses to elements of $A$. 
	The query algorithm requires access to a universal table of $o(n)$ bits.
\end{lemma}

\section{Proofs Omitted From Section \ref{section_partial_rank}}

\subsection{Fast Construction of Sequence $M$}
\label{app: generate_M}

We show how to create $M$ efficiently with the help of a universal table $U$.
This table has an entry for each possible pair $(D, t)$, where $D$ is a sequence of length $b = \lfloor\frac{\lg n}{2\lceil\lg \sigma\rceil}\rfloor$ drawn from $[\sigma]$ and $t$ is an integer in $[0, b]$.
If $t = 0$, this entry stores a sequence of length $b$ which is obtained by appending a $1$-bit to the end of the binary expression of each element in $D$.
Otherwise, this entry stores a sequence of length $b+1$ consisting of three sections: the first section is obtained by appending a $1$-bit to the end of the binary expression of each of the first $t$ elements in $D$, the second section is a boundary element, and the third section is obtained by appending a $1$-bit to the end of the binary expression of each of the last $b-t$ elements in $D$.
As there are at most $n^{1/2}$ possible sequences of length $b$ drawn from $\sigma$ and $t$ has $b+1$ possible values, $U$ has at most $n^{1/2}(b+1)$ entries.
Since each entry is encoded in at most $(b+1)(\lceil\lg \sigma\rceil+1)=O(\polylog(n))$ bits, $U$ uses $o(n)$ bits.
With $U$, we can scan $A$ and process $b$ of its elements in constant time; whether or where a boundary element should be created when processing these $b$ elements can be inferred by keeping track of the number of elements of $b$ that we have scanned so far. Note that at most one boundary element will be created when reading $b$ elements from $A$ as $b < \lg n < \sigma$. 
The time needed to create $M$ is hence $O(n'/b) = O(n^{\prime}\lg \sigma/\lg n)$. 

\subsection{Processing $M(u)$ Efficiently with $U'$}
\label{app: universal_table_lemma_B_c}

We show how to process $M(u)$ efficiently with the help of a universal table $U'$.
Recall that $b = \lfloor\frac{\lg n}{2\lceil\lg \sigma\rceil}\rfloor$. 
$U'$ has an entry for each possible pair  $(E, c)$, where $E$ is a
sequence of length $b$ drawn from universe $[2\sigma]$ and $c$ is an integer in $[0, \lceil\lg \sigma\rceil-1]$.
This entry stores a pair of packed sequences $E_0$ and $E_1$ defined as follows: $E_0$ or $E_1$ stores the boundary elements in $E$ and the 
regular elements in $E$ whose $c$-th most significant bit is $0$ or $1$, respectively.
The elements in $E_0$ retain their relative order in $E$, and the same is true with $E_1$. 
As $U'$ has $2^{b\times(\lceil\lg \sigma\rceil+1)}\times \lceil\lg \sigma\rceil$ entries and each entry stores a pair of packed sequences occupying $O(b\lceil\lg \sigma\rceil)$ bits in total, $U'$ uses $o(n)$ bits. 
With $U'$, we can process $M(u)$ in $O(|M(u)|\lg \sigma /\lg n+1)$ time. 

\subsection{Constructing Sequences $P_0, P_1,\ldots, P_{\sigma-1}$}
\label{app: construct_p_seq}

\begin{lemma}
	\label{lemma_C_k}
	Sequences $P_0, P_1,\ldots, P_{\sigma-1}$ can be constructed in $O(n^{\prime}{\lg^2 \sigma}/{\lg n}+ \sigma)$ time in total. 
\end{lemma}
\begin{proof}
	The construction consists of two phases. 
	In the first phase, we compute the set of pairs $R_k = \{(i, \prank(A_k, i))| 0 \le i \le \sigma-1\}$ for each chunk $A_k$.
	Even though $P_i[k] =  \prank(A_k, i)$ and thus the entries of all the $P_i$'s have been computed in this phase, the pairs themselves generated for $A_k$ are not in any order that allows us to directly assign values from these pairs to entries of $P_i$'s quickly enough.
	Thus, in the second phase, we reorganize all $n'$ pairs computed from all the chunks, to compute $P_0, P_1,\ldots, P_{\sigma-1}$ efficiently.
	
	We first show how to compute the pair set $R_k$ for each $A_k$ efficiently. 
	Let $I[0, \sigma-1]$ denote a packed sequence such that $I[i]=i$ for each $i\in[0, \sigma-1]$.
	Note that $I$ can be constructed once in $O(\sigma)$ time and shared with all chunks.
	By Lemma~\ref{lemma:wavelet_construct_d_packed}, a binary wavelet tree, in which node $u$ is associated with $A(u)$ and $I(u)$ as defined before, over $A_k$ could be constructed in $O(\sigma\lg^2 \sigma/\lg n+ \sigma)$ time.
	However, the second term $O(\sigma)$, when summed over all $n'/\sigma$ chucks, is too expensive to afford.
	Thus, we modify the structure of a wavelet tree to decrease this term.
	In the modified tree, when a node $v$ satisfies $|A(v)|\leq b = \lfloor\frac{\lg n}{2\lceil\lg \sigma\rceil}\rfloor$, we make $v$ a leaf node without any descendants. 
	With this modification, we observe the following two properties. 
	First, if a leaf node $l$ satisfies $|A(l)| > b$, then the tree level of $l$ must be $\lg \sigma$ and all entries of $A(l)$ store the same symbol. 
	Second, as there are at most $\lceil \sigma/b \rceil$ nodes at each level, the modified tree has $O(\sigma/b\times \lg \sigma)=O(\sigma\lg^2 \sigma/\lg n)$ nodes.
	Following from the analysis of the algorithm in the proof of Lemma~\ref{lemma:wavelet_construct_d_packed}, the modified tree can be constructed in $O(\sigma\lg^2 \sigma/\lg n)$ time.
	After this tree is constructed, we only keep the sequences $A(l)$ and $I(l)$ for each leaf node $l$ and call them {\em leaf sequences}. We discard the rest of the tree. 
	
	To further compute $R_k$ using these leaf sequences, observe that, for any symbol $\alpha$, there exists one leaf $l$ such that $A(l)$ contains all the occurrences of $\alpha$ in $A$. 
	Thus $(I(l)[i], \prank(A_k, I(l)[i])) = (I(l)[i], \prank(A(l), i))$ holds, which we can use to reduce the problem of computing the pairs in $R_k$ to the problem of computing the answer to a partial rank query at each position of $A(l)$ for each leaf $l$. 
	Hence for each leaf $l$, we define a packed sequence $Q(l)[0..|A(l)|-1]$ in which $Q(l)[i]=\prank(A(l), i)$ to store these answers. 
	To construct $Q(l)$ efficiently, we consider two cases. 
	When $|A(l)|\leq b$, we apply a universal table $U''$ to generate $Q(l)$ in constant time. 
	$U''$ has an entry for each possible pair $(F, x)$, where $F$ is a sequence of length $b$ drawn from universe $[\sigma]$, and $x$ is an integer in $[0, b]$.
	This entry stores a packed sequence $G[0..x]$ in which $G[i]=\prank(F, i)$.
	Similar to $U$ in Appendix \ref{app: generate_M}, $U''$ uses $o(n)$ bits.
	When $|A(l)|> b$, all entries of $A(l)$ store the same symbol. 
	Thus, we have $Q(l)[i]=i$ for each $i\in[0, |A(l)|-1]$, and hence we can create $Q(l)$ by copying the first $|A(l)|$ elements from the sequence $I$ which we created before. 
	In either case, $Q(l)$ can be constructed in $O(|A(l)|\lg \sigma/\lg n+1)$ time.
	Let $l_i$ denote the $(i+1)$-st leaf visited in a preorder traversal of the tree, and $f$ the number of leaves.
	Since $\sum_{i=0}^f|Q(l_i)| = \sigma$ and $f = O(\sigma\lg^2 \sigma/\lg n)$, the total time required to build $Q(l_0), Q(l_1), \ldots, Q(l_{f-1})$ is $O(\sigma\lg^2 \sigma/\lg n)$.
	Then we construct the concatenated packed sequence $I_k = I(l_0)I(l_1)\ldots I(l_{f-1})$ and $Q_k = Q(l_0)Q(l_1)\ldots Q(l_{f-1})$.
	It requires $O(\sigma\lg^2 \sigma/\lg n)$ to concatenate these sequences if we process $\Theta(\lg n)$ bits, i.e., $O(1)$ words, in constant time by performing bit operations. 
	Since for any $i \in [0, \sigma-1]$, $(I_k[i], Q_k[i])$ is a distinct pair in $R_k$, $I_k$ and $Q_k$ store all the pairs in $R_k$.
	We perform the steps in this and the previous paragraphs for all the chunks in $A$, and the total time spent in this phase is $O(n'{\lg^2 \sigma}/{\lg n}+\sigma)$. 
	
	Next we construct $P_0, P_1,\ldots, P_{\sigma-1}$ efficiently using the pairs computed in the previous phase. 
	We first build in $O(n'\lg^2 \sigma/\lg n)$ time two concatenated packed sequences each of length $n'$:  $I'=I_0I_1\ldots I_{n'/\sigma-1}$ and $Q = Q_0Q_1\ldots Q_{n'/\sigma-1}$.
	Then we construct a binary wavelet tree over $I'$.
	Each node, $u$, of the wavelet tree is associated with two sequences, $I'(v)$ which contains all the elements of $I'$ whose values are within the range represented by $v$, retaining their relative order in $I'$, and $Q(v)$ in which $Q(v)[i]$ is the element in $Q$ corresponding to $I'(v)[i]$.
	The  algorithm of constructing a binary wavelet tree shown in Lemma~\ref{lemma:wavelet_construct_d_packed} can be modified easily to construct this wavelet tree in $O(n'\lg^2 \sigma/\lg n+\sigma)$ time. 
	Let $l'_i$ denote the $(i+1)$st leaf of this wavelet tree in preorder.
	Observe that all the entries in $I'(l'_i)$ store $i$, and $I'(l'_i)[j]$ initially came from $A_j$, i.e., $I'(l'_i)[j]$ corresponds to the $i$th position in chunk $A_j$.
	Therefore, $Q(l'_i)[j] = P_i[j]$, and we have $P_i = Q(l'_i)$.
	The processing time required for this phase is also $O(n^{\prime}{\lg^2 \sigma}/{\lg n}+\sigma)$, which is the same as the bound for the first phase.
	Therefore, the total time  required to construct all sequences $P_0, P_1,\ldots, P_{\sigma-1}$ is $O(n'{\lg^2 \sigma}/{\lg n}+\sigma)$.
\end{proof}

\subsection{Space Analysis of Data Structure in Lemma \ref{rank_prime_1}}
\label{app: space_rank_prime_queryspace}
\begin{lemma}
	The structures built in Lemma \ref{rank_prime_1} occupy $n^{\prime}\lg \sigma+ o(n^{\prime}\lg \sigma)$ extra bits of space.
\end{lemma}
\begin{proof}
	In $B_c$, each $1$ bit corresponds to an occurrence of symbol $c$ in $A$, while each $0$ corresponds to a chunk. Thus, these bit vectors have $n'$ $1$s and $n'/\sigma \times \sigma = n'$ $0$s in total. Therefore, the lengths of all these bit vectors sum up to $2n'$. By Lemma~\ref{bit_sequence}, $o(n')$ bits are needed to augment them to support $\rankop$ and $\selop$. As each chunk has $\sigma$ elements, encoding an entry of each $P_c$ requires $\lceil\lg \sigma\rceil$ bits. Thus $P_0$, \ldots, $P_{\sigma-1}$ occupy $n'\lceil\lg \sigma\rceil$ bits in total. The total space usage of all the data structures in this section is therefore $2n^{\prime} + o(n^{\prime})+n^{\prime}\lceil\lg \sigma \rceil$ bits, which is $n^{\prime}\lg \sigma+ o(n^{\prime}\lg \sigma)$ when $\sigma>\lg n$.
\end{proof}

\section{Proofs Omitted From Section \ref{sect:ball_inheritance}}
\label{app:ball_inheritance}


We first discuss how to construct the ball inheritance structures of Chan et al.~\cite{chan2011orthogonal} efficiently over generalized wavelet trees,
by replacing some of their data structure components with those we designed in previous sections to achieve faster construction time.

When used to represent the point set $N$, each node, $u$, of $T$ is conceptually associated with an ordered list, $N(u)$, of points whose $x$-coordinates are within the range represented by $u$, and these points are ordered by $y$-coordinate.
To save space, Chan et al.~\cite{chan2011orthogonal} do not encode each ordered point list explicitly.
Instead, they define a sequence, $\Sp(u)$, of {\em skipping pointers} for $u$, in which $\Sp(u)[i]$ stores, at a certain number of levels below $u$, which descendant of $u$ has $N(u)[i]$ in its ordered list of points; different choices of the distance between $u$ and its descendant give different time-space tradeoffs.
Then, since both $N(u)$ and $N(\Sp(u)[i])$ order points by $y$-coordinate, the result of a $\prank(\Sp(u), i)$ query is the position of the point $N(u)[i]$ in $N(\Sp(u)[i])$.
Thus, to compute $\point(v, i)$, we can follow these skip pointers starting from $v$ and perform $\prank$ queries along the way, until we reach the leaf level of $T$.
As Chan et al. store the coordinates of each point in the ordered lists associated with the leaves, this process will answer $\point(v, i)$. 

The following lemma gives analyses of the approach of Chan et al.~\cite{chan2011orthogonal}, in which the analysis of preprocessing time is restricted to the special cases that we need for orthogonal range searching problems. 
\begin{lemma}
	\label{ball_inheritance_point_para}
	Let $X[0..n'-1]$ be a sequence drawn from alphabet $[\sigma]$ and $Y[0..n'-1]$ be a sequence in which $Y[i]=i$ for each $i\in[0..n'-1]$, where $max(\sigma, n')\leq n$.
	A $d$-ary wavelet tree over $X$, where $d$ is a power of $2$ upper bounded by $\sigma$, can be represented using $O(n'\tau (\lg \sigma) \log_{\tau} (\log_d \sigma)+n'\lg n'+\sigma w)$ bits to support $\point$ in $O(\log_{\tau} (\log_d \sigma))$ time.
	Given $X$ and $Y$ as input, this tree can be constructed in $O(n'\tau\lg^2 \sigma/\lg n+n'\lg n'\lg \sigma/\lg n+\sigma\log_d \sigma)$ time if $\sigma=O(2^{O(\sqrt{\lg n})})$. 
	If $d\ge 2^{\sqrt{\lg n}}$, the construction time is $O((n'+\sigma)\log_d \sigma)$. The construction requires a universal table of $o(n)$ bits. 
	\begin{REMOVED}
		Given $X$ and $Y$ as input, a $d$-ary wavelet tree over $X$ using $O(n'\tau (\lg \sigma) \log_{\tau} (\log_d \sigma)+n'\lg n'+\sigma w)$ bits of space can be constructed in $P(n)$ time to support $\point$ in $O(\log_{\tau} (\log_d \sigma))$ time, where (a) $P(n)=O(n'\tau (\lg^2 \sigma/\lg n)+\sigma\log_d \sigma)$ if $\sigma=O(2^{O(\sqrt{\lg n})})$, (b) $P(n)=O(n'\log_d \sigma+n'\tau\lg \sigma/\sqrt{\lg n}+\sigma\log_d \sigma)$ if $\sigma=\Omega(2^{O(\sqrt{\lg n})})$ and $d<2^{\sqrt{\lg n}}$, \footnote{In this case, it is possible to further decrease the first term in $P(n)$ by combining a wavelet tree construction algorithm and integer sorting, but we do not go into details since this case is never encountered in our solutions to orthogonal range reporting/successor.} (c) $P(n)=O((n'+\sigma)\log_d \sigma)$, if $\sigma=\Omega(2^{O(\sqrt{\lg n})})$ and $d\ge 2^{\sqrt{\lg n}}$. 
	\end{REMOVED}
\end{lemma}
\begin{proof}
	Let $h$ denote $\log_d \sigma$.
	Each point in $N$ appears in the ordered point list associated with a node $u$ at each level, $l$, of $T$.
	At each internal node, a skipping pointer is created, which encodes the rank of the descendant of $u$ among all the descendants of $u$ at level $l' = \tau^{c+1} \lceil l/\tau^{c+1} \rceil$. 
	As $u$ has $d^{l'-l} \le d^{\tau^{c+1}}$ descendants at level $l'$, $\Sp(v)[i]$ can be encoded using at most $\tau^{c+1}\lg d$ bits.
	Since there are at most $h / \tau^{c}$ levels with color $c$, the skipping pointers created for this point across all levels of $T$ occupy at most $\sum_{c=0}^{\log_{\tau} h-1} \frac{h}{\tau^c}\times \tau^{c+1}\times \lg d=O(\tau\lg \sigma \log_{\tau} h)$ bits. 
	As there are $n^{\prime}$ points in $N$, the space usage of all skipping pointers is $O(n'\tau\lg \sigma \log_{\tau} h)=O(n'\tau (\lg \sigma) \log_{\tau} (\log_d \sigma))$ bits.
	By either Lemma~\ref{rank_prime_1} or Lemma~\ref{rank_select_0}, the extra space cost needed to build data structures to support $\prank$ in sequences of skip pointers is also $O(n'\tau (\lg \sigma) \log_{\tau} (\log_d \sigma))$ bits. 
	We know that there are in total $n'$ points at the leaf level and the coordinates of each point can be encoded in $O(\lg \sigma+\lg n')$ bits. 
	So the cost of storing point coordinates at leaf levels is $O(n'(\lg \sigma+\lg n'))$ bits.
	As $T$ has $O(\sigma)$ nodes, the implicit representation of $T$, with color and depth information, occupies $O(\sigma w)$ bits.
	Overall, the space usage of these structures is $O(n'\tau (\lg \sigma) \log_{\tau} (\log_d \sigma)+n'\lg n'+\sigma w)$ bits.
	
	Now we analyze the query time of $\point$. 
	We retrieve the depth, $l$, of $v$ to get the color, $c$, assigned to level $l$. If $v$ is a leaf node, i.e., $c=\log_{\tau} h$, then $N(v)$ is stored explicitly, and we return $N(v)[i]$ as the answer. 
	Otherwise, let $l' = \tau^{c+1} \lceil l/\tau^{c+1} \rceil$. 
	Then the point $p$ that we will eventually return as the answer to the query is also distributed to the ordered point list associated with the $\Sp(v)[i]$-th descendant, $u$, of $v$ at level $l'$.
	Node $u$ can be located in constant time using the implicit representation of $T$ as a complete $d$-ary tree. 
	Furthermore, $p$ is at position $j = \prank(\Sp(v), i)$ of $N(u)$.
	We then perform the query $\point(u, j)$ recursively to compute the answer. 
	To bound the running time, observe that this process is terminated once we reach a leaf level. 
	Hence, the process is applied recursively $O(\log_{\tau}  h)$ times, each with a cost of $O(1)$. 
	Therefore, it requires $O(\log_{\tau}  h)$ time to support $\point(v, i)$.
	
	To construct these data structures, we first build $T$ as a $d$-ary wavelet tree $T$ over $X$ with value and index arrays.
	If $\sigma=O(2^{\sqrt{\lg n}})$, $T$ can be built in $O(n'{\lg \sigma(\lg n' + \lg \sigma)}/{\lg n}+\sigma)$ time by Lemma~\ref{lemma:wavelet_construct_d_packed}.
	Otherwise, it takes $O(n'\log_d \sigma+\sigma)$ time using the algorithm shown in Lemma~\ref{lemma:wavelet_construct_d_unpacked}.
	Computing the depth of each node of $T$ and storing $T$ implicitly use $O(\sigma)$ time. 
	We then assign colors to its levels as follows:
	We first assign color $0$ to the root level.
	Then we assign color $\log_{\tau} h$ to the leaf level. 
	Among the remaining levels, we assign color $\log_{\tau} h -1$ to those that are multiples of $\tau^{\log_{\tau} h -1}=h/\tau$, and so on.
	During this process, an array of flags is used to mark those levels that have been assigned colors. 
	As we use $O(1)$ time for each level, this requires $O(\log_d \sigma)$ time.
	Observe that the value and index arrays of each node $u$ of $T$ encode the $x$- and $y$-coordinates of the points in $N(u)$, respectively.
	Therefore, at the leaf level, we keep its value and index arrays as the encoding of $N(u)$.
	Otherwise, let $c$ be the color assigned to level $l$.
	Then, for any $i \in [0, |N(u)|-1]$, $\Sp(u)[i]$ needs to store the rank of the descendant of $u$ at level $l' = \tau^{c+1} \lceil l/\tau^{c+1} \rceil)$,
	which can be computed as $A(u)[i](l\lg d..l'\lg d)$; recall that $A(u)$ is the value array of $u$ storing the $x$-coordinates of the points in $N(u)$. 
	By Lemma \ref{lemma_split_1}, all elements of $\Sp(u)$ can be generated in $O(|N(u)|{\lg \sigma}/{\lg n}+1)$ time.
	The overall time needed to generate all the skipping pointers across the entire tree $T$ is thus bounded by $O(n'{\lg^2 \sigma}/({\lg n}\times\lg d)+\sigma)$, which is subsumed by the time cost spending on building $T$.
	We discard the value and index arrays of $u$ after $\Sp(u)$ has been built.

	Next, we show how to build the data structure for $\prank$ queries upon $\Sp(u)$.
	We first consider the case in which $\sigma =O(2^{\sqrt{\lg n}})$, in which the alphabet size of $\Sp(u)$ is at most $\sigma = O(2^{\sqrt{\lg n}})$, so we apply Lemma \ref{rank_prime_1} to build a $\prank$ structure over $S(u)$. 
	$\Sp(u)$ is drawn from alphabet $d^{l'-l}$, and since $l' -l \le \tau^{c+1}$, this structure can be built in $O(|\Sp(u)|(l'-l)^2 \lg^2 d/\lg n + d^{l'-l})=O(|\Sp(u)|\tau^{2c+2}\lg^2 d/\lg n + d^{l'-l})$ time.
	Over all nodes at level $l$, observe that the term, $|\Sp(u)|\tau^{2c+2}\lg^2 d/\lg n$, sums up to $n'\tau^{2c+2}\lg^2 d/\lg n$, while the term, $d^{l'-l}$, sums up to $f \times d^{l'-l}$, where $f$ is the number of nodes at level $l$.
	To bound $f$, observe that, as each node at level $l$ has $d^{l'-l}$ descendants, there are $f\times d^{l'-l}$ nodes at level $l'$ and we have $f\times d^{l'-l} \le \sigma$. 
	Thus the sum of the term, $d^{l'-l}$, over nodes at level $l$ is bounded by $\sigma$. 
	Hence the total time required to build auxiliary data structures for $\prank$ for nodes at a level with color $c$ is $O(n'\tau^{2c+2} \lg^2 d/\lg n + \sigma)$. 
	As there are at most $h /\tau^c$ levels with color $c$, the total construction time over all levels of $T$, including the time spent building and coloring $T$ itself, is $O(n'{\lg \sigma(\lg n' + \lg \sigma)}/{\lg n}+\sigma +\sum_{c=0}^{(\log_{\tau} h)-1} (h /\tau^c) \times O(n'\tau^{2c+2} \lg^2 d/\lg n + \sigma)=O(n'\tau\lg^2 \sigma/\lg n+n'\lg n'\lg \sigma/\lg n+\sigma\log_d \sigma)$.
	Finally, we consider the case in which $d\ge 2^{\sqrt{\lg n}}$. In this case, all $\prank$ structures are built using Lemma \ref{rank_select_0}, so the total construction time is  $O(n'\log_d \sigma+\sigma)+\sum_{c=0}^{(\log_{\tau} h)-1} O(n'+ \sigma) = O(n'\log_d \sigma +\sigma\log_d \sigma)$. 
\end{proof}

Another operation of ball inheritance is $\noderange(c, d, v)$.
Recall that given a range $[c, d]$ and a node $v$ of $T$, $\noderange(c, d, v)$ finds the range $[c_v, d_v]$ such that $I(v)[i] \in [c, d]$ iff $i \in [c_v, d_v]$.
Obviously, $c_v$ or $d_v$ is equal to the positions of $\succ(c)$ or $\pred(d)$ in $I(v)$, respectively.
Hence by constructing predecessor/successor data structures over $I(v)$, we can support $\noderange$.

Lemmas~\ref{ball_inheritance_large_d} and \ref{ball_inheritance_point_small_02} addressing special cases of ball inheritance, in which either the wavelet tree has high fanout or the coordinates can be encoded in $O(\sqrt{\lg n})$ bits, can thus be obtained by choosing appropriate values for $\tau$ and applying different data structures for $\pred$/$\succ$ operations.
Next, we give the proofs of these Lemmas.

\subsection{Proof of Lemma~\ref{ball_inheritance_large_d}}
\begin{proof}
	Consider the case in Lemma \ref{ball_inheritance_point_para} for $d \ge 2^{\sqrt{\lg n}}$. In this case, set $n' = n$ and $d =2^{\sqrt{\lg n}}$. By further setting $\tau = 2$  or  $\tau=\lg^{\epsilon} \sigma$, we have the result $(a)$ or $(b)$, respectively, apart from the support of $\noderange$.
	Next, we show the data structure supporting $\noderange$, whose space cost and construction time are both subsumed by the data structure supporting $\point$.
	We apply Lemma \ref{pre_succ_chan} to construct a data structure supporting $\pred/\succ$ over $I(v)$ at each node $v$, 
	which answers $\noderange$ in $O(\lg \lg n)$ time and $O(1)$ calls to $\point$ without requiring $I(v)$ to be stored explicitly.
	This structure occupies $O(|I(v)|\lg \lg n)$ extra bits of space can be built upon $I(v)$ in $O(|I(v)|)$ time.
	As $T$ has $\lg \sigma/\sqrt{\lg n}+1$ levels and there are $n$ elements at each level, the overall time needed to construct it over all nodes is $O(n \lg \sigma/\sqrt{\lg n})$, and the overall extra space cost is $O(n\lg \lg n \times \lg \sigma/\sqrt{\lg n})=o(n\lg \sigma)$ bits.
\end{proof}


\subsection{Proof of Lemma~\ref{ball_inheritance_point_small_02}}

When coordinates of points can be encoded in $O(\sqrt{\lg n})$ bits, we can achieve faster construction time by applying the $\prank$ supporting data structure designed in Lemma~\ref{rank_prime_1} and the $\pred/\succ$ supporting data structure under the indexing model designed in Lemma~\ref{pre_succ_index_duplicates}.
\begin{proof}
	Consider the case in Lemma \ref{ball_inheritance_point_para} for $\sigma = O(2^{O(\sqrt{\lg n})})$. By setting $\tau=2$ and applying $n' = O(\sigma^{O(1)})$, we can obtain the construction time, the space cost and the query time needed to support $\point$, which match the bounds show in this lemma.
	It remains to show the support of $\noderange$.
	As $n'$ is bounded by $O(\sigma^{O(1)})$, each element of the index array $I(u)$ of each node $u$ can be encoded with $O(\lg \sigma)=O(\sqrt{\lg n})$ bits.
	We then build the predecessor/successor data structure over $I(u)$ using Lemma \ref{pre_succ_index_duplicates}.
	Given that $\point$ takes $O(\lg (\lg \sigma/\lg d))$ time, $\noderange$ can be answered in $O(\lg \lg \sigma+\lg (\lg \sigma/\lg d))=O(\lg \lg \sigma)$ time without explicitly storing $I(u)$.
	By Lemma \ref{pre_succ_index_duplicates}, this data structure occupies $O(|I(u)|\lg \lg \sigma)$ extra bits of space and can be built upon $I(u)$ in $O(|I(u)|/\sqrt{\lg n}+1)$ time.
	As $T$ has $\sigma$ nodes and $h+1$ levels, the overall time needed to construct it over all nodes is $\sum_u O(|I(u)|/\sqrt{\lg n}+1)=O(n'\lg \sigma/(\sqrt{\lg n}\times \lg d)+\sigma)$ bounded by $O(n^{\prime}{\lg^2 \sigma}/{\lg n}+\sigma\log_d \sigma)$.
	Similarly, the overall extra space cost is $O(n'\lg \lg \sigma \log_d \sigma)$ bits.
	Therefore, the overall space cost required by the data structure designed is $O(n^{\prime}\lg \sigma\lg (\lg \sigma/\lg d)+n'\lg \lg \sigma \log_d \sigma+\sigma w)=O(n^{\prime}\lg \sigma\lg (\lg \sigma/\lg d)+\sigma w)$ bits.
\end{proof}

\subsection{Proof of Lemma \ref{ball_inheritance_point}}
\label{app:twist}
One may attempt to achieve this by setting $\tau$ to $(\log_d \sigma)^{\epsilon}$ in Lemma \ref{ball_inheritance_point_para} to achieve constant-time support for $\point$, but then the construction time is $O(n' \tau \lg^2 \sigma/\lg n+\sigma \log_d \sigma)$, in which the first term is not small enough.
This term shows the time spent on building the auxiliary data structures for $\prank$.
To remove the $\tau$ factor in it, we have designed in Section \ref{sect:ball_inheritance} a variant of the solution by Chan et al.~\cite{chan2011orthogonal} by storing point coordinates at a subset of levels of $T$ instead of only at the leaf level.
This twist allows us to build $\prank$ structures at fewer tree levels to decrease the construction time, and we still achieve the query time and space bounds that  match those in part (c) of Lemma~\ref{lemma:ball_intro}.

\begin{lemma}
	\label{lem:point_space}
	The data structure in Lemma \ref{ball_inheritance_point} occupy $O(n^{\prime}\lg \sigma\log_d^{\epsilon} \sigma+\sigma w)$ bits. 
\end{lemma}
\begin{proof}
	We only discuss the space usage of all $\Sp(u)$'s and the space cost of $\prank$ data structure built over $\Sp(u)$, as the space costs of all other data structures are the same to those in Lemma \ref{ball_inheritance_point_para}.
	Each point in $N$ appears in the ordered point list associated with a node $u$ at each level, $l$, of $T$.
	When the color, $c$, of $l$ is not $1/\epsilon -1$, a skipping pointer is created for this node, which encodes the rank of the descendant of $u$ among all the descendants of $u$ at level $l' = \tau^{c+1} \lceil l/\tau^{c+1} \rceil$. 
	As $u$ has $d^{l'-l} \le d^{\tau^{c+1}}$ descendants at level $l'$, $\Sp(v)[j]$ can be encoded using at most $\tau^{c+1}\lg d$ bits.
	Since there are at most $\log_d \sigma / \tau^{c}$ levels with color $c$, the skipping pointers created for this point across all levels of $T$ occupy at most $\sum_{c=0}^{1/\epsilon -2} \frac{\log_d \sigma}{\tau^c}\times \tau^{c+1}\times \lg d=O(\tau\lg \sigma)$ bits. 
	As there are $n^{\prime}$ points in $N$, the space usage of all skipping pointers is $O(n^{\prime}\lg \sigma\log_d^{\epsilon} \sigma)$ bits.
	By Lemma~\ref{rank_prime_1}, the extra space cost needed to build data structures to support $\prank$ in sequences of skip pointers is also $O(n^{\prime}\lg \sigma\log_d^{\epsilon} \sigma)$  bits.
	In addition to the other data structures shown in Lemma \ref{ball_inheritance_point_para}, the overall space cost is $O(n^{\prime}\lg \sigma\log_d^{\epsilon} \sigma+\sigma w)$ bits. 
	
	%
\end{proof}

\begin{lemma}\label{lem:point_query}
	The new data structures can support $\point(v, i)$ in $O(1)$ time. 
\end{lemma}
\begin{proof}
	We retrieve the depth, $l$, of $v$ to get the color, $c$, assigned to level $l$. If $c=1/\epsilon -1$, then $N(v)$ is stored explicitly, and we return $N(v)[i]$ as the answer. 
	Otherwise, let $l' = \tau^{c+1} \lceil l/\tau^{c+1} \rceil$. 
	Then the point $p$ that we will eventually return as the answer to the query is also distributed to the ordered point list associated with the $\Sp(v)[i]$-th descendant, $u$, of $v$ at level $l'$.
	Node $u$ can be located in constant time using the implicit representation of $T$ as a complete $d$-ary tree. 
	Furthermore, $p$ is at position $j = \prank(\Sp(v), i)$ of $N(u)$.
	We then perform the query $\point(u, j)$ recursively to compute the answer. 
	To bound the running time, observe that this process is terminated once we reach a level with color $1/\epsilon -1$,
	and one out of every $\tau^{1/\epsilon-1}$ levels of $T$ is assigned this color. 
	Hence, the process is applied recursively $O(\log_{\tau}  \tau^{1/\epsilon-1}) = O(1)$ times, each with a cost of $O(1)$. 
	Therefore, it requires constant time to support $\point(v, i)$.
\end{proof}

\begin{lemma}\label{lem:point_construct}
	The new data structures can be constructed in $O(n^{\prime}{\lg^2 \sigma}/{\lg n}+\sigma\log_d \sigma)$ time. 
\end{lemma}
\begin{proof}
	We build $T$ and the skipping pointer sequences for its nodes as in the proof of Lemma~\ref{ball_inheritance_point_para}, which uses $O(n^{\prime}{\lg^2 \sigma}/\lg n+\sigma)$ time.
	We then apply Lemma \ref{rank_prime_1} to build the data structure for $\prank$ queries upon $\Sp(u)$.
	As $\Sp(u)$ is drawn from alphabet $d^{l'-l}$, this requires $O(|\Sp(u)|(l'-l)^2 \lg^2 d/\lg n + d^{l'-l})$ time, which is bounded by $O(|\Sp(u)|\tau^{2c+2}\lg^2 d/\lg n + d^{l'-l})$ as $l' -l \le \tau^{c+1}$.
	Hence the total time required to build auxiliary data structures for $\prank$ for nodes at a level with color $c$ is $O(n'\tau^{2c+2} \lg^2 d/\lg n + \sigma)$.
	As there are at most $\log_d \sigma/\tau^c$ levels with color $c$, the total construction time over all levels of $T$ is
	$\sum_{c=0}^{1/\epsilon-2} (\log_d \sigma/\tau^c) \times O(n'\tau^{2c+2} \lg^2 d/\lg n + \sigma)=O(n^{\prime}{\lg^2 \sigma}/{\lg n}+\sigma\log_d \sigma)$, which dominates the construction time of all data structures.
\end{proof}

\section{Proofs Omitted From Section \ref{sect: range_reporting}}

\subsection{Using Range Maximum/Minimum to Answer 3-Sided Queries in the Proof of Lemma~\ref{range_reporting_1}}
\label{app:rmqrecursive}

	To report points in $([a, +\infty)\times[c_l, d_l]) \cap N(u_l)$, we need only report the points in $N(u_l)[c_l, d_l]$ whose $x$-coordinates are at least $a$.
	This can be done by performing range maximum queries over $A(u_l)$ recursively as follows.
	We perform $\rMq(c_l, d_l)$ to get the index $m$ of the point $p$ that has the maximum $x$-coordinate in $N(u_l)[c_l, d_l]$, and retrieve its coordinates $(p.x, p.y)$ by $\point(u_l, m)$.
	If $p.x\geq a$, we report $p$ and perform the same process recursively in $N(u_l)[c_l, m-1]$ and $N(u_l)[m+1, d_l]$. 
	Otherwise we stop.
	The points in $([0, b]\times[c_r, d_r]) \cap N(u_r)]$ can be reported in a similar way. 
	To analyze the query time, observe that we perform $\noderange$ twice in $O(\lg\lg n)$ time. The recursive procedure is called $O(\occ)$ times, and each time it is performed, it uses $O(1)$ time.
	All other steps require $O(1)$ time.
	Therefore, the overall query time is $O(\lg \lg n+\occ)$.

\subsection{The Analysis of the Space Usage and Construction Time for Data Structure in Lemma \ref{fat_rect}}

Our analysis requires previous result shown as follows:
\begin{lemma}[{\cite[Section 2]{chan2011orthogonal}},{\cite[Lemma 5]{DBLP:journals/corr/abs-1108-3683}}]
	\label{range_reporting_chan}
	Given a set, $N$, of $n$ points in $[u]\times[u]$, a data structure of $O(n\lg^{1+\epsilon} n)$ bits can be constructed in $O(n\lg n)$ time, which supports orthogonal range reporting over $N$ in $O(\lg \lg u+\occ)$ time, where $\occ$ is the number of reported points.
\end{lemma}

With this lemma, we now give the analysis of the space usage and construction time for the data structure in Lemma \ref{fat_rect}.
\label{app: space_construction_fat_rect}
\begin{lemma}
	The data structure in Lemma~\ref{fat_rect} occupies $O(n^{\prime}\lg^{1/2+\epsilon} n+w(2^{\sqrt{\lg n}}+n'/2^{\sqrt{\lg n}}))$ bits of space and can be constructed in $O(n^{\prime}+\sqrt{\lg n}\cdot 2^{\sqrt{\lg n}})$ time.
\end{lemma}

\begin{proof}
	To bound the storage costs, by Lemma~\ref{range_reporting_1}, the orthogonal range reporting structure over each $N_i$ uses $O(2^{2\sqrt{\lg n}}\lg^{1/2+\epsilon} n+w 2^{\sqrt{\lg n}})$ bits.
	Thus, the range reporting structures over $N_0, N_1, \ldots, N_{n/b-1}$ occupy $O((n'/b)\times(2^{2\sqrt{\lg n}}\lg^{1/2+\epsilon} n+w\cdot 2^{\sqrt{\lg n}})) = O(n'\lg^{1/2+\epsilon} n + n'w/2^{\sqrt{\lg n}})$.
	As there are at most $n'/2^{\sqrt{\lg n}}$ points in $\hat{N}$, by Lemma~\ref{range_reporting_chan}, the range reporting structure for $\hat{N}$ occupies $O(n'\lg^{1+\epsilon} n/2^{\sqrt{\lg n}}) = o(n')$ bits.
	There are $n'$ points in all $P_{i,j}$'s and each of their local $y$-coordinates can be encoded in $\lg b=2\sqrt{\lg n}$ bits.
	In addition, each $P_{i,j}$ requires a pointer to encode its memory location, so $n'/b\times 2^{\sqrt{\lg n}} = n' / 2^{\sqrt{\lg n}}$ pointers are needed.
	Therefore, the total storage cost of all $P_{i,j}$'s is $O(n'w / 2^{\sqrt{\lg n}} + n'\sqrt{\lg n})$.
	Thus the  space costs of all structures add up to $O(n^{\prime}\lg^{1/2+\epsilon} n+n'w/2^{\sqrt{\lg n}})$ bits.
	Note that the above analysis assumes $n' > b$. Otherwise, $O(n^{\prime}\lg^{1/2+\epsilon} n+w\cdot 2^{\sqrt{\lg n}})$ bits are needed, so we use  $O(n^{\prime}\lg^{1/2+\epsilon} n+w\cdot 2^{\sqrt{\lg n}}+n'w/2^{\sqrt{\lg n}})$ as the space bound on both cases. 
	
	Regarding construction time, when $n' > b$, observe that the point sets $N_0, N_1, \ldots, N_{n'/b-1}$ and $\hat{N}$, as well as the sequences $P[i,j]$ for $i = 0,1,\ldots,n'/b-1$ and $j=0,1,\ldots,2^{\sqrt{\lg n}}-1$, can be computed in $O(n')$ time.
	By Lemma~\ref{range_reporting_chan}, The range reporting structure for $\hat{N}$ can be built in $O(n'/b \times \lg n) = o(n')$ time.
	Finally, the total construction time of the range reporting structures for $N_0, N_1, \ldots, N_{n/b-1}$ is $O({n^{\prime}}/{2^{2\sqrt{\lg n}}}\times(2^{2\sqrt{\lg n}}+\sqrt{\lg n}\times 2^{\sqrt{\lg n}}))=O(n^{\prime})$, which dominates the total preprocessing time of all our data structures.
	When $n' \le b$, the construction time is $O(n^{\prime}+\sqrt{\lg n}\cdot 2^{\sqrt{\lg n}})$ by Lemma~\ref{range_reporting_1}, so we use $O(n^{\prime}+\sqrt{\lg n}\cdot 2^{\sqrt{\lg n}})$ as the upper bound on construction time in both cases.
\end{proof}

\subsection{The Analysis of Space Usage and Construction Time for Data Structure in Lemma \ref{theorem_range_reporting_general}}
\label{app: space_construction_theorem_range_reporting_general}
\begin{lemma}
	The data structure in Lemma~\ref{theorem_range_reporting_general} occupies $O(n\lg^{1+\epsilon}  \sigma+n\lg n)$ bits of space for any constant $\epsilon>0$  and can be constructed in $O(n{\lg \sigma}/{\sqrt{\lg n}})$ time.
\end{lemma}

\begin{proof}
	Now we analyze the space costs.
	$T$ with support for ball inheritance uses $O(n\lg^{1+\epsilon}\sigma+n\lg n)$ bits for any positive $\epsilon$.
	For each internal node $v$, since $w = \Theta(\lg n)$, the data structure for range reporting over $\hat{S}$ uses $O(|S(u)|\lg^{1/2+\epsilon'} n+2^{\sqrt{\lg n}}\lg n+|S(u)|\lg n/2^{\sqrt{\lg n}})$ bits for any positive $\epsilon'$. 
	This subsumes the cost of storing $M(u)$ which is $O(|A(u)|)$ bits. 
	As $T$ has $O(\sigma/2^{\sqrt{\lg  n}})$ internal nodes,
	the total cost of storing these structures at all internal nodes is $\sum_u O(|S(u)|\lg^{1/2+\epsilon'} n+2^{\sqrt{\lg n}}\lg n+|S(u)|\lg n/2^{\sqrt{\lg n}}) = O(n\lg\sigma/\sqrt{\lg n} \times \lg^{1/2+\epsilon'}n + \sigma \lg n) = O(n\lg\sigma\lg^{\epsilon'} n + \sigma\lg n)$.
	As $\lg n \le \lg^2\sigma$ and $\sigma \le n$, this is bounded by $O(n\lg^{1+2\epsilon'}\sigma)$.
	Setting $\epsilon' = \epsilon/2$, the space bound turns to be $O(n\lg^{1+\epsilon}  \sigma)$ bits. 
	Overall, the data structures occupy $O(n\lg^{1+\epsilon}\sigma+n\lg n)$ bits.
	
	Finally, we analyze the construction time.
	As shown in Lemma~\ref{ball_inheritance_large_d}, $T$ with support for ball inheritance can be constructed in $O(n\lg \sigma/\sqrt{\lg n})$ time.
	For each internal node $u$ of $T$, constructing $M(u)$ and the range reporting structure over $\hat{S}(v)$ requires $O(|A(u)|+|S(u)|+\sqrt{\lg n}\cdot 2^{\sqrt{\lg n}})=O(|S(u)|+\sqrt{\lg n}\cdot 2^{\sqrt{\lg n}})$ time.
	As $T$ has $O(\sigma/2^{\sqrt{\lg  n}})$ internal nodes, these structures over all internal nodes can be built in $\sum_u O(|S(u)|+\sqrt{\lg n}\times2^{\sqrt{\lg n}}) = O(n\lg\sigma/\sqrt{\lg n}+\sigma\sqrt{\lg n}) = O(n\lg \sigma/\sqrt{\lg n})$ as $\sigma\leq n$.
	The preprocessing time of all data structures is hence $O(n\lg \sigma/\sqrt{\lg n})$.
\end{proof}

\section{Optimal Orthogonal Range Successor with Fast Preprocessing}
\label{app: fast_orthogonal_range_succ}

We now design data structures over $n$ points in 2d rank space that support an orthogonal range successor query in optimal time and can be constructed fast.
Previously, using a solution to three-sided next point problem defined in Section \ref{sect: three_side_next_point} and ball inheritance, Zhou~\cite{zhou2016two} solved the orthogonal range successor problem within optimal query time.
As their solution relies on auxiliary structures on a binary wavelet tree, the pre-processing time requires $O(n\lg n)$.
Our data structure is constructed upon a $2^{\sqrt{\lg n}}$-ary wavelet tree to reduce the problem in the general case to the three-sided next point query problem and the orthogonal range successor problem in the special case in which the points are from a $2^{\sqrt{\lg n}}\times n'$ {\em medium narrow} grid.
And our solutions to ball inheritance upon a generalized wavelet tree with high fanout can apply, which reduces the processing time from $O(n\lg n)$ to $O(n\sqrt{\lg n})$.
We further design data structures with fast construction time supporting the three-sided next point problem and the reduced orthogonal range successor problem in the special cases.
Hence, we describe our solutions in this order: in Section \ref{sect: three_side_next_point}, we introduce the methods to solve the three sided next point query, and in Section \ref{sect:successor}, we describe our solutions to the orthogonal range successor problem over a small narrow, medium narrow and general grid, respectively.

\subsection{Fast Construction of the Three-Sided Next Point  Structures}
\label{sect: three_side_next_point}
In this subsection, we show how to efficiently construct data structures for three-sided next point queries, defined as follows.
Given a set of points  $N$, of $n$ points in the rank space, a three-sided next point query to be the problem of retrieving the point with the smallest $y$-coordinate among all points in $N\cap Q$ where $Q=[a, +\infty]\times[c, d]$.
We assume that points are in the rank space. 

The methods shown in Lemmas \ref{lemma_three_sided_next_point_big} and \ref{lemma_three_sided_next_point_small} are under the indexing model: after the construction of the data structure, each query operation needs to access some points and report them.
The point set $N$ itself need not be stored explicitly; it suffices to provide an operator supporting the access to an arbitrary point of $N$. 
The operator is implemented by $\point(v, i)$ from ball inheritance.
Our solutions will use the previous results as follows:

\begin{lemma}[{\cite[Lemma 5]{nekrich2012sorted}}]
	\label{lemma-three-sided-base}
	There exists a data structure of $O(n\lg^3 n)$-bit space constructed upon a set of $n$ points in rank space in $O(n\lg^2 n)$ time, which supports three-sided next point query in $O(\lg \lg n)$ time.
\end{lemma}

Both Lemmas \ref{lemma_three_sided_next_point_block} and \ref{lemma_three_sided_next_point_big} are originally designed by Zhou~\cite{zhou2016two}.
But they did not mention the construction time before.
Here, we only give the analysis of the construction time.

\begin{lemma}[{\cite[Lemma 3.2]{zhou2016two}}]
	\label{lemma_three_sided_next_point_block}
	Let $N$ be a set of $\lg^3 n$ points in rank space. 
	Given packed sequences $X$ and $Y$ respectively encoding the $x$- and $y$-coordinates of these points where $Y[i] = i$ for any $i \in [0, \lg^3 n-1]$, a data structure using $O(\lg^3 n\lg \lg n)$ bits of space constructed over $N$ in $o(\lg^3 n/\sqrt{\lg n})$ time that answers the three-sided next point query in $O(\lg \lg n)$ time. 
	The query procedure requires access to a universal table of $o(n)$ bits.
\end{lemma}
\begin{proof}
	We divide each consecutive $\lg^{3/4} n$ points along $y$-axis of $N$ into blocks.
	The dividing operation can be done in $O(\lg^3 n/\lg^{3/4} n)=o(\lg^3 n/\sqrt{\lg n})$ time by bit-wise operations.
	As each point requires $6\lg \lg n$ bits of space, each block uses $6\lg \lg n\times\lg^{3/4} n$ bits less than a word.
	From each block, we apply a universal table $U$ of $o(n)$ bits to retrieve the point with maximum $x$-coordinate in constant time.
	$U$ has an entry for each possible triple $(\alpha, \beta, \gamma)$, where $\alpha$ or $\beta$ is a packed sequence of length at most $\lg^{3/4} n$ drawn from $[\lg^3 n]$ denoting the $x$- or $y$-coordinates of the points, respectively, and $\gamma$ is an integer $\in [0..(\lg^{3/4} n)-1]$ denoting the number of points.
	This entry stores the point with the maximum $x$-coordinate among the point set denoted by $\alpha$ and $\beta$.
	As $U$ has $O(2^{(\lg^{3/4} n)\times(6\lg \lg n)}\times \lg^{3/4} n)$ entries and each entry stores a point of $6\lg \lg n$ bits, $U$ uses $o(n)$ bits of space.
	Let $\hat{N}$ denote the set of the selected points and $|\hat{N}|=\lceil n^{\prime}/\lg^{3/4} n \rceil$.
	We use Lemma \ref{lemma-three-sided-base} to build a data structure $DS(\hat{N})$ over $\hat{N}$ for the the three-sided next point query. 
	The data structure $DS(\hat{N})$ using $O(|\hat{N}|\lg^3 |\hat{N}|)=o(n')$ bits of space can be built in $O(|\hat{N}|\lg^2 |\hat{N}|)=O(|\hat{N}|\lg \lg n)$ time bounded by $o(n^{\prime}/\sqrt{\lg n})$.
	Overall, the data structure uses $O(n'\lg \lg n+o(n'))=O(n'\lg \lg n)$ bits of space and can be constructed in $o(n^{\prime}/\sqrt{\lg n})$ time.
\end{proof}

\begin{lemma}[{\cite{zhou2016two}}]
	\label{lemma_three_sided_next_point_big}
	Let the sequence $A[0..n'-1]$ of distinct elements drawn from $[n]$ denote a point set $N = \{(A[i], i)| 0\le i \le n'-1\}$, where $ n^{\prime}\leq n$. 
	There exists a data structure using $O(n^{\prime}\lg \lg n)$ bits of extra space constructed over $N$ in $O(n^{\prime})$ time that answers three-sided next point query in $O(\lg \lg n)$ time and $O(1)$ access to $A$. 
\end{lemma}
\begin{proof}
	We divide $N$ into $n'/\lg^3 n$ blocks, and for each $i \in [0, n'/\lg^3 n-1]$, the $i$-th block, $N_i$, contains points in $N$ whose $y$-coordinates are in $[i \lg^3 n, (i+1)\lg^3 n-1]$.
	Assume for simplicity that $n'$ is divisible by $\lg^3 n$.
	We linearly scan the points of each block and retrieve the one with maximum $x$-coordinate from each block. 
	Let $\hat{N}$ denote the selected points and $|\hat{N}|= n^{\prime}/\lg^3 n$. 
	We apply Lemma \ref{lemma-three-sided-base} to build the data structure $DS(\hat{N})$ over $\hat{N}$ for three-sided next point queries. 
	As shown in Lemma \ref{lemma-three-sided-base}, the data structure $DS(\hat{N})$ uses $O(\hat{N}\lg^3 \hat{N})=O(n'/\lg n)$ bits of space and can be built in $O(\hat{N}\lg^2 \hat{N})=O(n^{\prime}/\lg n)$ time.
	
	We apply the general rank reduction technique \cite{willard1985new} to reduce the points of each block to the rank space, which can be accomplished by sorting the points once with respect to each of $x$- and $y$-coordinate.
	As there are only $\polylog(n)$ points within each block, it is well-known that an atomic heap \cite{FredmanW94} can be used to sort them in linear time.
	As each point can be encoded in $O(\lg \lg n)$ bits after rank reduction, we take linear time to store the $x$- and $y$- coordinates of points of each block $N_i$ in a packed sequences ${X'}(N_i)$ and ${Y'}(N_i)$, respectively.
	Note that the $y$-coordinates of points in ${Y'}(N_i)$ denote the in-block indexes.
	Afterwards, we build data structure $TS(N_i)$ over ${X'}(N_i)$ and ${Y'}(N_i)$ of each block $N_i$ by Lemma \ref{lemma_three_sided_next_point_block} in $o(\lg^3 n/\sqrt{\lg n})$ time for three-sided next point query within a block.
	As each point in $N$ has a distinct $y$-coordinate represented by its index $j$ in $A[0..n'-1]$, we can use the block index $i$ and in-block index $i'$ to compute $j$, i.e., $j=i\times\lg^3 n+i'$, and then apply $\point(v, j)$ to retrieve the original $x$-coordinate of that point.
	Thus, we do not need to store the coordinates of points in $N$ for saving space.
	
	As defined above, let $Q$ denote the query range $[a, +\infty]\times[c,d]$.  
	When a query happens upon some block $N_i$, the query range $[a, +\infty]$ along $x$-axis need to be reduced to $[\hat{a}, +\infty]$ in rank space. 
	Before discarding the $x$-coordinates of points in $N_i$, we sort them in linear time using an atomic heap \cite{FredmanW94}.
	Let $S(N_i)$ denote the sequence storing all sorted $x$-coordinates. 
	If $S(N_i)$ is available at the querying procedure, we can apply $\succ(a)$ over $S(N_i)$ to find $\hat{a}$. 
	However, storing $S(N_i)$ will overflow the total space usage. 
	Instead, all points of $N_i$ are still sorted by $x$-coordinate and each point $e$ after sorting is specified by its in-block index $i'$ using $O(\lg \lg n)$ bits of space. 
	As all $x$-coordinates are distinct in $N$, we can use Lemma \ref{pre_succ_chan} to build the predecessor/successor data structure $PS(N_i)$ of $O(\lg^3 n\lg \lg n)$ bits in linear time over $S(N_i)$. 
	Afterwards $S(N_i)$ can be discarded. 
	The $\succ(a)$ query can be retrieved in $O(\lg \lg n)$ time and $O(1)$ calls to $\point$ without storing the sequence $S(N_i)$. 
	For all $n'/\lg^3 n$ blocks, the data structure can be constructed in $O(n'+n^{\prime}/\lg n+(n'/\lg^3 n)\times \lg^3 n+ (n'/\lg^3 n)\times o(\lg^3 n/\sqrt{\lg n})+(n'/\lg^3 n)\times \lg^3 n)=O(n')$ time.
\end{proof}

More interestingly, when the $x$- and $y$-coordinates of the points are stored in the packed form, we can solve the three-side queries with a data structure built in sublinear time. 
Our method requires a fast sorting algorithm for performing rank reduction over a small set of points.
When a sequence of $n'$ integers from $[\sigma]$ is bit packed into $O(n'\lg \sigma/ \lg n)$ words, it can be sorted using a bit-packed version of mergesort:
\begin{lemma}[{\cite{ah1997}}]
	\label{packed_sort}
	A packed sequence $A[0..n^{\prime}-1]$ from alphabet $[\sigma]$, where $max(\sigma, n') \leq n$, can be sorted in $O(n^{\prime}\lg n'{\lg \sigma}/{\lg n})$ time with the help of a universal tables of $o(n)$ bits.
\end{lemma}

Note one  difference between Lemmas  \ref{lemma_three_sided_next_point_big} and \ref{lemma_three_sided_next_point_small}:  Lemma~\ref{lemma_three_sided_next_point_small} allows multiple points with the same  $x$-coordinate.
\begin{lemma}
	\label{lemma_three_sided_next_point_small}
	Let $N$ be a set of $n'$ points with distinct $y$-coordinates in a $2^{\sqrt{\lg n}}\times n'$ grid, where $n^{\prime}=O(2^{c\sqrt{\lg n}})$ for any constant integer $c$. 
	Given packed sequences $X$ and $Y$ respectively encoding the $x$- and $y$-coordinates of these points where $Y[i] = i$ for any $i \in [0, n'-1]$, a data structure using $O(n^{\prime}\lg \lg n)$ bits of extra space constructed over $N$ in $O({n^{\prime}}\lg \lg n/{\sqrt{\lg n}})$ time that answers three-sided next point query in $O(\lg \lg n)$ time and $O(1)$ access to $A$. 
\end{lemma}
\begin{proof}
	As shown in the proof of Lemma \ref{lemma_three_sided_next_point_big}, the linear construction time is bounded by the rank reduction operation and building the predecessor/successor index data structure upon sorted $x$-coordinates of points for each block. 
	As the $x$- and $y$-coordinates of each point are encoded with $O(\sqrt{\lg n})$ bits and the coordinates of points are stored in packed sequences, we can sort a block of $\lg^{3} n$ points in $O((\lg^3 n)\times({\sqrt{\lg n}}/{\lg n})\times\lg \lg n)$ time by applying Lemma \ref{packed_sort}. 
	Meanwhile, the predecessor and successor data structure for each block $N_i$ can be constructed in $O(\lg^3 n/\sqrt{\lg n})$ time by applying Lemma \ref{pre_succ_index_duplicates}. 
	Overall, the whole data structure over $N$ can be built in $O({n^{\prime}}\lg \lg n/{\sqrt{\lg n}})$ time. 
\end{proof}

\subsection{Fast Construction of the Orthogonal Range Successor  Structures}
\label{sect:successor}

Now we consider the solution of  the orthogonal range successor problem with optimal time.
We describe our solution first for a  small narrow grid of size $\lg^{1/4}\times n'$, then for  a  medium narrow grid of size $2^{\sqrt{\lg n}}\times n'$, and finally for  an $n\times n$ grid.  Every step in our construction relies on the previous one.



\subsubsection{Orthogonal Range Successor Queries in a Small Narrow Grid}
First, we consider a special case such that the number of points is less than $\lg n$.
\begin{lemma}
	\label{lemma_orthogonal_range_successor_base}
	Let $N$ be a set of $n'$ points with distinct $y$-coordinates in a $\lg^{1/4} n \times n'$ grid where $n' <\lg n$.
	Given packed sequences $X$ and $Y$ respectively encoding the $x$- and $y$-coordinates of these points where $Y[i] = i$ for any $i \in [0, n'-1]$, a data structure of $O(n^{\prime}\lg \lg n)$ bits can be built in $o(n^{\prime}/\sqrt{\lg n})$ time over $N$ to answer orthogonal range successor query in $O(1)$ time. 
	The construction and query procedure each requires access to a universal table of $o(n)$ bits.
\end{lemma}

\begin{proof}
	When $n'\leq \lg^{3/4} n$, we can apply a universal table $U$ of $o(n)$ bits to retrieve in constant time the point with the smallest $y$-coordinate in the query range.
	$U$ has an entry for each possible set $(\alpha, \beta, \gamma, a', b', c', d')$, where $\alpha$ (or $\beta$, respectively) is a packed sequence of length at most $\lg^{3/4} n$ drawn from $[\lg^{1/4} n]$ (or $[\lg^{3/4} n]$, respectively) denoting the $x$-coordinate (or $y$-coordinate, respectively), $\gamma$ is an integer $\in [0..(\lg^{3/4} n)-1]$ denoting the number of points, and $a', b', c', d'$ each is an integer $\in [0, (\lg n)-1]$ and all together denotes the query range.
	This entry stores the point with the smallest $y$-coordinate in the point set denoted by $\alpha$ and $\beta$.
	As $U$ has $O(2^{(\lg^{3/4} n)\times(\lg \lg n)}\times \lg^{3/4} n \times \lg^4 n)$ entries and each entry stores a point of at most $\lg \lg n$ bits, $U$ uses $o(n)$ bits of space.
	
	Assume for simplicity that $n'$ is divisible by $\lg^{3/4} n$.
	We divide $N$ into $n'/\lg^{3/4} n$ subsets, and for each $i \in [0, n'/\lg^{3/4} n-1]$, the $i$-th subset, $N_i$, contains points in $N$ whose $y$-coordinates are in $[i \lg^{3/4} n, ((i+1)\lg^{3/4} n)-1]$.
	The division of $N$ into $N_i$  can be implemented in $O(n'/\lg^{3/4} n)$ time using bitwise operations.
	We also define a point set $\hat{N}$ in a $\lg^{1/4} n \times n'$ grid. 
	For each set $N_i$ where $i \in [0, n'/\lg^{3/4} n-1]$ and each $j \in [0, \lg^{1/4} n-1]$, if there exists at least one point in $N_i$ whose $x$-coordinate is $j$, we store the one with the smallest $y$-coordinate among them in $\hat{N}$.
	Thus the number of points in $\hat{N}$ is at most $n'/\lg^{3/4} n \times \lg^{1/4} n = n'/\sqrt{\lg n}<\sqrt{\lg n}$. 
	As each block of points occupies in total $(1/4\lg \lg n+\lg n')\times\lg^{3/4} n$ bits, creating points for $\hat{N}$ from each block can be implemented in $O(1)$ time with a universal table $U'$ of $o(n)$ bits.
	$U'$ has an entry for each possible triple $(\alpha, \beta, \gamma)$, where $\alpha$ (or $\beta$, respectively) is a packed sequence of length at most $\lg^{3/4} n$ drawn from $[\lg^{1/4} n]$ (or $[(\lg n)-1]$, respectively) denoting the $x$-coordinate (or $y$-coordinate, respectively), and $\gamma$ is an integer $\in [0..(\lg^{3/4} n)]$ denoting the number of points.	This entry stores a packed sequence of at most $\lg^{1/4} n$ points for $\hat{N}$ occupying at most $(\lg^{1/4} n)\times(1/4\lg \lg n+\lg \lg n)$ bits.
	Similar to the universal table $U$, $U'$ uses of $o(n)$ bits.
	Therefore, constructing $\hat{N}$ takes $O(n'/\lg^{3/4} n)$ time.
	Obviously, storing all points in $N$ and $\hat{N}$ occupies $O(n'\lg \lg n)$ bits of space in total.
	
	Let $Q=[a, b]\times[c,d]$ denote the query range and $N_i,..., N_j$ denote the blocks intersecting the range $[c,d]$ such that $i=\lfloor c/ \lg^{3/4} n\rfloor$ and $j=\lfloor d/ \lg^{3/4} n\rfloor$. 
	If $i=j$, then the query range $Q$ is within a single block and  we can apply $U$ to retrieve the answer in constant time. 
	Otherwise, we sequentially check $B_i\cap Q$, $(B_{i+1}\cup B_{i+2} \cup...\cup B_{j-1})\cap Q$, $B_j \cap Q$, and stop querying once the lowest point is retrieved. 
	The second case can also be answered in constant time by querying over $U$ with the range $\hat{N}\cap[a,b]\times[i\times b+b, j\times b]$. 
	Overall, the query time is $O(1)$.
\end{proof}

Next, we consider the orthogonal range successor problem on a larger number of points.
Our method requires the previous result shown as follows:

\begin{lemma}[{\cite[Lemma C.3]{belazzougui2016range}}]
	\label{lem:select}
	Let $A[0..n^{\prime}-1]$ be a packed sequence  drawn from alphabet $[\sigma]$, where $n' \le n$. 
	A data structure of $O(n'\lg \sigma)$ bits supporting $\selop$ in $O(1)$ time can be constructed in $O(n^{\prime}\lg^2 \sigma/\lg n+\sigma)$ time.
\end{lemma}

\begin{lemma}
	\label{theorem_orthogonal_range_successor_small}
	Given packed sequence $X[0..n'-1]$ drawn from $[\lg^{1/4} n]$ denote a point set $N = \{(A[i], i)| 0\le i \le n'-1\}$, where $n'\le n$, a data structure of $O(n^{\prime}\lg^2 \lg n+w\times\lg^{1/4} n)$ bits can be built in $O(n^{\prime}/\sqrt{\lg n})$ time over $N$ to answer orthogonal range successor query in $O(\lg \lg n)$ time. 
\end{lemma}

\begin{proof}
	Lemma \ref{lemma_orthogonal_range_successor_base} already achieves this result for $n'< \lg n$, so it suffices to consider the case $n' \ge \lg n$ in the rest of the proof.
	%
	
	We construct a binary wavelet tree $T$ upon $X[0..n'-1]$ by Lemma~\ref{lemma:wavelet_construct_d_packed} together with the value array $A(v)$ in packed form at each node $v$ and the bit sequence $S(v)$ if $v$ is an internal node.
	Recall that $A(v)$ stores the $x$-coordinates of the ordered list, $N(v)$, of points from $N$ whose $x$-coordinates are within the range represented by $v$, and these points are ordered by $y$-coordinate. 
	The tree $T$ has $\lceil 1/4\lg \lg n \rceil+1$ levels and $\lg^{1/4} n$ nodes.
	Over the sequences associated with each internal node $u$, we build the following data structures:
	\begin{itemize}
		\item $RK_{ds}(u)$ supports $O(1)$-time $\rankop$ queries over $A(u)$ by Lemma \ref{lemma_rank_small};
		\item $SL_{ds}(u)$ supports $O(1)$-time $\selop$ queries over $A(u)$ by Lemma \ref{lem:select};
		\item  $B_{ds}(u)$  supports $O(1)$-time $\rankop$ queries over $S(u)$ by Lemma \ref{bit_sequence}.
	\end{itemize}
	
	As shown in Lemma~\ref{lemma:wavelet_construct_d_packed}, $T$ uses $O(n'\lg^2 \lg n+w \times \lg^{1/4} n)$ bits of space and can be constructed in $O(n'\lg^2 \lg n/\lg n+\lg^{1/4} n)=o(n'/\sqrt{\lg n})$ time as $n'\ge \lg n$.
	Both $RK_{ds}(u)$ and $SL_{ds}(u)$ use $O(|A(u)|\lg \lg n)$ bits of space, while $B_{ds}(u)$ only requires $o(|S(u)|)$ bits of space.
	As there are $\lceil 1/4\lg \lg n \rceil$ non-leaf levels in $T$ and  $n'$ elements across each level, all data structures $RK_{ds}(u)$, $SL_{ds}(u)$ and $B_{ds}(u)$ use $O(n'\lg^2 \lg n)$ bits.
	Constructing $RK_{ds}(u)$ takes $O(|A(u)|\lg \lg n/\lg n+1)$ time, $SL_{ds}(u)$ uses $O(|A(u)|\lg^2 \lg n/\lg n+\lg^{1/4} n)$ time to build, and building $B_{ds}(u)$ requires $O(|S(u)|/\lg n+1)$ time.
	As $T$ has less than $\lg^{1/4} n$ internal nodes, the overall construction time for these data structures is $\sum_u (O(|A(u)|\lg \lg n/\lg n+1) + O(|A(u)|\lg^2 \lg n/\lg n+\lg^{1/4} n) +  O(|S(u)|/\lg n+1))=O(n' \lg^3 \lg n/\lg n+\sqrt{\lg n})=O(n'/\sqrt{\lg n})$ as $n'\ge \lg n$.
	Therefore, this data structure requires $O(n'\lg^2 \lg n+w \times \lg^{1/4} n)$ bits of space and takes $O(n'/\sqrt{\lg n})$ time to construct.
	With $RK_{ds}(u)$ and $SL_{ds}(u)$, we can implement the operation $\point(u,i)$ in constant time, 	as we have $\point(u,i) = (A(u)[i], \selop_{A(u)[i]}(A(r), \rankop_{A(u)[i]}(A(u), i)))$, where $r$ is the root node.

	Given a query range $Q=[a, b]\times[c, d]$,  we first locate the lowest common ancestor $v$ of $l_{a}$ and $l_{b}$ in constant time, where $l_{a}$ and $l_{b}$ denote the $a$-th and $b$-th leftmost leaves of $T$, respectively.
	Let $\pi_a$ and $\pi_b$ denote the paths from $v$ to the $a$-th leaf and from $a$ to the $b$-th leaf respectively.
	For each node $u$ on $\pi_a$ we mark the right child of $u$ if it exists and is not  on the path $\pi_a$. 
	For each node $u$ on $\pi_b$ we mark the left child of $u$ if it exists and is not on the path $\pi_b$. 
	In addition, we mark the $a$-th and $b$-th leaves. 
	The points on the marked node have the $x$-coordinate in the range $[a, b]$. 
	As the height of $T$ is $O(\lg \lg n)$, there are in total $O(\lg \lg n)$ nodes marked. 
	
	The points at all marked nodes within the query range $Q$ can be identified in total $O(\lg \lg n)$ time. 
	Let $[c_v, d_v]$ denote the range such that $I(v)[c_v..d_v]$ within $[c, d]$.
	Recall that $I(v)$ is the index array that is not explicitly stored in our data structure.
	Clearly, the range $[c_v, d_v]$ can be retrieved by answering $\rankop$ query over $S(u)$ where $u$ is the parent of $v$, i.e., $[c_v, d_v]=[\rankop_0(S(u), c_u),\rankop_0(S(u), d_u)]$ if $v$ is the left child of $u$. Otherwise, $[c_v, d_v]=[\rankop_1(S(u), c_u),\rankop_1(S(u), d_u)]$.
	As we move down  the path from the root node to the $a$-th leaf ($b$-th leaf, respectively), we answer $\rankop$ queries at the visited nodes. And if a marked node $v$ is identified, we can find the index range $[c_v, d_v]$ by $\rankop$ queries over the bit sequence $S(u)$ where $u$ is the parent of $v$.
	
	Obviously, within each marked node $v$ the point represented by ($A(v)[c_v], I(v)[c_v]$) carries the ``locally'' smallest $y$-coordinate in $Q$, where $I(v)[c_v]=\selop_{A(v)[c_v]}(A(r), \rankop_{A(v)[c_v]}(A(v), c_v))$. 
	Therefore, the lowest point in $Q$ can be retrieved by comparing the $O(\lg \lg n)$ locally lowest points at all marked nodes. Overall, the query time is $O(\lg \lg n)$.
\end{proof}

\subsubsection{Orthogonal Range Successor Queries in a Medium Narrow  Grid}
Our solution for points in a $2^{\sqrt{\lg n}}\times n'$ grid for any $2^{\sqrt{\lg n}}\le n' \le n$ uses the following previous result:

\begin{lemma}[{\cite[Theorem 3.3]{zhou2016two}}]
	\label{theorem_range_successor}
	There exists a data structure of $O(n\lg n \lg \lg n)$ bits constructed upon a set of $n$ points in rank space in $O(n\lg n)$ time that answers orthogonal range successor queries in $O(\lg \lg n)$ time.
\end{lemma}


The following lemma presents our solution for a medium narrow grid.  

\begin{lemma}
	\label{ors_on_block}
	Let $N$ be a set of $n'$ points with distinct $y$-coordinates in a $2^{\sqrt{\lg n}} \times n'$ grid where $2^{\sqrt{\lg n}} \leq n^{\prime} \le 2^{2{\sqrt{\lg n}}}$.
	Given packed sequences $X$ and $Y$ respectively encoding the $x$- and $y$-coordinates of these points where $Y[i] = i$ for any $i \in [0, n'-1]$, a data structure of $O(n'\sqrt{\lg n}\lg \lg n+w\times 2^{\sqrt{\lg n}})$ bits can be built over $N$ in $O(n'+2^{\sqrt{\lg n}}\times\sqrt{\lg n}/\lg \lg n)$ time to answer orthogonal range successor query in $O(\lg \lg n)$ time. 
\end{lemma}
\begin{proof}
	We build a $\lg^{1/4}$-ary wavelet tree $T$ upon $X[0, n'-1]$ and $Y[0, n'-1]$ with support for ball inheritance using Lemma~\ref{ball_inheritance_point_small_02}.
	Recall that each node $u$ of $T$ is associated with (but does not explicitly store)
	the value array $A(u)$ and the index array $I(u)$, in which $A(u)$ and $I(u)$ store the $x$- and $y$-coordinates of the ordered list, ${N}(u)$, of points from $N$ whose $x$-coordinates are within the range represented by $u$, and these points are ordered by $y$-coordinate. 
	Furthermore, $u$ is associated with another sequence $S(u)$ drawn from alphabet $[\lg^{1/4} n]$, in which $S(u)[i]$ encodes the rank of the child of $u$ that contains $N(u)[i]$ in its ordered list.
	Let $\hat{S}(u)$ denote the point set $\{(S(u)[i], i)| 0\le i \le |S(u)|-1\}$, and we use Lemma~\ref{theorem_orthogonal_range_successor_small} to build a structure $RS_{ds}(u)$  supporting orthogonal range successor queries over $\hat{S}(u)$.
	Let $\hat{N}(u)$ denote the point set $\hat{N}(u)= \{(A(u)[i], i)| 0\le i \le |A(u)|-1\}$, and we use Lemma~\ref{lemma_three_sided_next_point_small} to build a structure $TS_{ds}(u)$  supporting three sided next point queries over $\hat{N}(u)$.
	Note that as shown in Lemma \ref{ball_inheritance_point_small_02}, both $\point(v,i)$ and $\noderange(c,d,v)$ can be answered in $O(\lg \lg n)$ time on $T$.
	
	Given a query range $Q=[a, b]\times[c, d]$,  we first locate the lowest common ancestor $v$ of $l_{a}$ and $l_{b}$ in constant time, where $l_{a}$ and $l_{b}$ denote the $a$-th and $b$-th leftmost leaves of $T$, respectively.
	Let $v_i$ denote the $i$-th child of $v$, for any $i \in [0, \lg^{1/4} n-1]$.
	We first locate two children, $v_{a'}$ and $v_{b'}$, of $v$ that are ancestors of $l_a$ and $l_b$, respectively.
	They can be found in constant time by simple arithmetic as each child of $v$ represents a range of equal size.
	Then the answer, $Q\cap N$, to the query can be reduced to retrieving the lowest point among three point sets $A_1=Q\cap N(v_{a'})$, $A_2=Q\cap (N(v_{a'+1})\cup N(v_{a'+2})\cup\ldots N(v_{b'-1}))$ and $A_3=Q\cap N(v_{b'})$.
	To find the lowest point in $A_1$, we need only retrieve the point $p'_1$ with the smallest $y$-coordinate in $[a, +\infty]\times[c_{v_{a'}}, d_{v_{a'}}]$ where $[c_{v_{a'}}, d_{v_{a'}}]= \noderange(c,d, v_{a'})$ and then use $\point(v_{a'}, p'_1.y)$ to find its original coordinates $p_1$ in $N$.
	The point $p'_1$ can be found by querying over $TS_{ds}(u)$ in $O(\lg \lg n)$ time using the algorithm shown in the proof of Lemma~\ref{lemma_three_sided_next_point_small}.
	With $O(\lg \lg n)$-time support for $\noderange$ and $\point$, $p_1$ can be retrieved in $O(\lg \lg n)$ time.
	Similarly, we can find the lowest point $p_3$ in $A_3$ in $O(\lg \lg n)$ time.
	To compute $A_2$, observe that any entry, $\hat{S}(v)[i]$, can be obtained by replacing the $x$-coordinate of point $N(v)[i]$ with the rank of the child whose ordered list contains $N(v)[i]$.
	Hence, by performing an orthogonal range successor query over $RS_{ds}(v)$ to compute $\hat{S}(v) \cap ([a^{\prime}+1, b^{\prime}-1]\times[c_v, d_v])$, where $[c_v, d_v] = \noderange(c,d, v)$,
	we can find in $O(\lg \lg n)$ time the lowest point $p'_2$ in $\hat{S}(v) \cap ([a^{\prime}+1, b^{\prime}-1]\times[c_v, d_v])$.
	Again, we use $\point$ to find its original coordinates $p_2$ in $N$.
	Obviously, the lowest point in $Q\cap N$ is the point with the smallest $y$-coordinate among $p_1$, $p_2$, and $p_3$.
	Therefore, the overall query time required is $O(\lg \lg n)$.

	Now we analyze the space costs.
	$T$ with support for ball inheritance uses $O(n'\sqrt{\lg n}\lg \lg n+ w\times 2^{\sqrt{\lg n}})$ bits by Lemma \ref{ball_inheritance_point_small_02}.
	For each internal node $u$, $RS_{ds}(u)$ over $\hat{S}(u)$ uses $O(|S(u)|\lg^2 \lg n+w\times \lg^{1/4} n)$ bits of space as shown in Lemma \ref{theorem_orthogonal_range_successor_small}.
	This subsumes the cost of storing $TS_{ds}(u)$ over $\hat{N}(u)$, which is $O(|S(u)|\lg \lg n)$ bits.
	As $T$ has $O(2^{\sqrt{\lg n}}/\lg^{1/4} n)$ internal nodes and $4\sqrt{\lg n}/\lg \lg n$ tree levels,
	the total cost of storing these structures at all internal nodes is 
	$\sum_u O(|S(u)|\lg^2 \lg n+w\times\lg^{1/4} n)= O(n'\sqrt{\lg n}\lg \lg n+w\times 2^{\sqrt{\lg n}})$ bits of space.
	Therefore, all the data structures occupy $O(n'\sqrt{\lg n}\lg \lg n+w\times 2^{\sqrt{\lg n}})$ bits of space.

	Finally, we analyze the construction time.
	As shown in Lemma~\ref{ball_inheritance_point_small_02}, $T$ with support for ball inheritance can be constructed in $O(n'+2^{\sqrt{\lg n}}\times \sqrt{\lg n}/\lg \lg n)$ time.
	At each internal node $u$ of $T$, constructing $TS_{ds}(u)$ requires $O({|A(u)|}/{\sqrt{\lg n}}\times\lg \lg n +1 )$ time using the algorithm in the proof of Lemma \ref{lemma_three_sided_next_point_small} and $RS_{ds}(v)$ requires $O(|S(u)|/\sqrt{\lg n}+1)$ time by Lemma \ref{theorem_orthogonal_range_successor_small}.
	As $T$ has ${\sqrt{\lg n}}/(1/4\lg \lg n)$ non-leaf levels and $O(2^{\sqrt{\lg n}}/\lg^{1/4} n)$ internal nodes, these structures over all internal nodes can be built in $\sum_u O(|A(u)|/{\sqrt{\lg n}}\times\lg \lg n+1)=O(n^{\prime})$ time.
	Therefore, the overall construction time is $O(n'+2^{\sqrt{\lg n}}\times \sqrt{\lg n}/\lg \lg n)$.
\end{proof}

\begin{lemma}
	\label{lemma_range_successor_core_summary}
	Let $N$ be a set of $n'$ points with distinct $y$-coordinates in a $2^{\sqrt{\lg n}} \times n'$ grid where $2^{\sqrt{\lg n}}\leq n'\leq n$.
	Given packed sequences $X$ and $Y$ respectively encoding the $x$- and $y$-coordinates of these points where $Y[i] = i$ for any $i \in [0, n'-1]$, a data structure of $O(n^{\prime}\sqrt{\lg n}\lg \lg n+w(n'/2^{\sqrt{\lg n}}+2^{\sqrt{\lg n}}))$ bits can be built over $N$ in $O(n^{\prime}+\sqrt{\lg n}\cdot 2^{\sqrt{\lg n}}/\lg \lg n)$ time to answer orthogonal range successor query in $O(\lg \lg n)$ time.
\end{lemma}

\begin{proof}
	Let $b$ denote $2^{2\sqrt{\lg n}}$.
	We need only consider the case in which $n' > b$ as Lemma~\ref{ors_on_block} applies otherwise.
	Assume for simplicity that $n'$ is divisible by $b$.
	We divide $N$ into $n'/b$ subsets, and for each $i \in [0, n'/b-1]$, the $i$th subset, $N_i$, contains points in $N$ whose $y$-coordinates are in $[i b, (i+1)b-1]$.
	The dividing procedure can be achieved in linear time.
	Let $p$ be a point in $N_i$. We call its coordinates $(p.x, p.y)$ {\em global coordinates}, while $(p.x', p.y') = (p.x, p.y \bmod b)$ its {\em local coordinates} in $N_i$; the conversion between global and local coordinates can be done in constant time.
	Hence the points in $N_i$ with their local coordinates can be viewed as a point set in a $2^{\sqrt{\lg n}}\times 2^{2\sqrt{\lg n}}$ grid, and we apply Lemma~\ref{ors_on_block} to construct an orthogonal range search structure $RS(N_i)$ over $N_i$.
	We also define a point set $\hat{N}$ in a $2^{\sqrt{\lg n}} \times n'$ grid. 
	For each set $N_i$ where $i \in [0, n'/b-1]$ and each $j \in [0, 2^{\sqrt{\lg n}}-1]$, if there exists at least one point in $N_i$ whose $x$-coordinate is $j$, we store the one among them with the smallest $y$-coordinate in $\hat{N}$.
	Thus the number of points in $\hat{N}$ is at most $2^{\sqrt{\lg n}} \times n' /b = n'/2^{\sqrt{\lg n}}$. 
	Finally, we build the data structure $\hat{RS}_{ds}$ for orthogonal range successor over $\hat{N}$ by Lemma~\ref{theorem_range_successor}.
	
	Given a query range $Q=[x_1, x_2]\times[y_1, y_2]$, we first check if $\lfloor y_1/b \rfloor$ is equal to $\lfloor y_2/b\rfloor$. If it is, then the points in the answer to the query reside in the same subset $N_{\lfloor y_1/b\rfloor}$, and we can retrieve the lowest point $e$ by performing an orthogonal range successor query in $N_{\lfloor y_1/ b \rfloor}\cap Q$, which requires $O(\lg\lg n)$ time by Lemma~\ref{ors_on_block}.
	Then we retrieve its original coordinates in $N$, which is $(e.x, b\lfloor y_1/ b \rfloor+e.y)$.
	Otherwise, let $N_s,\dots, N_e$ denote the blocks interacting $[y_1, y_2]$, where $s=\lfloor y_1/ b\rfloor$ and $e=\lfloor y_2/ b\rfloor$. 
	We sequentially look for the lowest point in $A_1=N_s \cap [x_1, x_2]\times[y_1 \mod b, +\infty]$, $A_2=(N_{s+1}\cup \cdots \cup N_{e-1})\cap [x_1, x_2]\times [-\infty, +\infty]$, and $A_3=N_e \cap [x_1, x_2]\times[0, y_2 \mod b]$.
	Once a point $e$ is returned, we retrieve the original coordinates of $e$ in $N$ and terminate the the query procedure.
	Both the cases $A_1$ and $A_3$ can be answered in $O(\lg \lg n)$ time by Lemma~\ref{ors_on_block}.
	It remains to find the lowest point in $A_2$, which can be implemented by querying in $O(\lg \lg n)$ time over $\hat{N}$ for the lowest point in $Q$ using Lemma~\ref{theorem_range_successor}.
	Overall, the query procedure requires $O(\lg \lg n)$ time.

	To bound the storage costs, by Lemma~\ref{ors_on_block}, the orthogonal range successor structure over each $N_i$ uses $O(2^{2\sqrt{\lg n}}\lg \lg n+w\cdot 2^{\sqrt{\lg n}})$ bits.
	Thus, the orthogonal range successor structures over $N_0, N_1, \ldots, N_{n'/b-1}$ occupy $O((n'/b)\times(2^{2\sqrt{\lg n}}\sqrt{\lg n}\lg \lg n+w\cdot 2^{\sqrt{\lg n}})) = O(n'\sqrt{\lg n}\lg \lg n + n'w/2^{\sqrt{\lg n}})$.
	As there are at most $n'/2^{\sqrt{\lg n}}$ points in $\hat{N}$, by Lemma~\ref{theorem_range_successor}, the range successor structure for $\hat{N}$ occupies $O(n'\lg \lg n \lg n/2^{\sqrt{\lg n}}) = o(n')$ bits.
	Thus the  space costs of all structures add up to $O(n^{\prime}\sqrt{\lg n}\lg \lg n+n'w/2^{\sqrt{\lg n}})$ bits.
	Note that the above analysis assumes $n' > b$. Otherwise, $O(n^{\prime}\sqrt{\lg n}\lg \lg n+w\cdot 2^{\sqrt{\lg n}})$ bits are needed, so we use  $O(n^{\prime}\sqrt{\lg n}\lg \lg n+w(n'/2^{\sqrt{\lg n}}+2^{\sqrt{\lg n}}))$ bits as the space bound on both cases. 

	Regarding construction time,  observe that the point sets $N_0, N_1, \ldots, N_{n'/b-1}$ and $\hat{N}$, can be computed in $O(n')$ time.
	By Lemma~\ref{theorem_range_successor}, The range successor structure for $\hat{N}$ can be built in $O(n'/b \times \lg n') = o(n')$ time.
	Finally, the total construction time of the range successor structures for $N_0, N_1, \ldots, N_{n/b-1}$ is $O({n^{\prime}}/{2^{2\sqrt{\lg n}}}\times(2^{2\sqrt{\lg n}}+\sqrt{\lg n}\times 2^{\sqrt{\lg n}}/\lg \lg n))=O(n^{\prime})$, which dominates the total preprocessing time of all our data structures.
	When $n' \le b$, the construction time is $O(n^{\prime}+\sqrt{\lg n}\cdot 2^{\sqrt{\lg n}}/\lg \lg n)$ by Lemma~\ref{ors_on_block}, so we use $O(n^{\prime}+\sqrt{\lg n}\cdot 2^{\sqrt{\lg n}}/\lg \lg n)$ as the upper bound on construction time in both cases.
\end{proof}

\subsubsection{Orthogonal Range Successor Queries in a $n\times n$ Grid}
Finally, we give the complete proof of Theorem \ref{theorem_range_successor_new}. 

Let the sequence $X[0, n-1]$ denote the point set $N = \{(X[i], i)| 0\le i \le n-1\}$.	We build a $2^{\sqrt{\lg n}}$-ary wavelet tree $T$ upon $X[0, n-1]$ with support for ball inheritance using part (a) of Lemma~\ref{ball_inheritance_large_d}.
	Recall that each node $u$ of $T$ is associated with 
	the value array $A(u)$ and the index array $I(u)$ (these arrays are not stored explicitly); $A(u)$ and $I(u)$ contain the $x$- and $y$-coordinates of ${N}(u)$,  where $N(u)$ is the list of points from $N$ whose $x$-coordinates are within the range of $u$, and points in $N(u)$ are ordered by their $y$-coordinates. 
	Furthermore, $u$ is associated with another sequence $S(u)$ drawn from alphabet $[2^{\sqrt{\lg n}}]$, in which $S(u)[i]$ encodes the rank of the child of $u$ that contains $N(u)[i]$ in its ordered list.
	Let $\hat{S}(u)$ denote the point set $\{(S(u)[i], i)| 0\le i \le |S(u)|-1\}$, and we use Lemma~\ref{lemma_range_successor_core_summary} to build a structure $RS_{ds}(u)$  supporting orthogonal range successor queries over $\hat{S}(u)$.
	Let $\hat{N}(u)$ denote the point set $\hat{N}(u)= \{(A(u)[i], i)| 0\le i \le |A(u)|-1\}$, and we use Lemma~\ref{lemma_three_sided_next_point_big} to build a structure $TS_{ds}(u)$  supporting three sided next point queries over $\hat{N}(u)$.
	The query procedure is exactly the same as in the proof of Lemma \ref{ors_on_block} and  requires $O(\lg \lg n)$ time.

	Now we analyze the space usage.
	$T$ with support for ball inheritance uses $O(n\lg n \lg \lg n)$ bits.
	For each internal node $u$, since $w = \Theta(\lg n)$, $RS_{ds}(u)$ over $\hat{S}(u)$ uses $O(|S(u)|\sqrt{\lg n}\lg \lg n+\lg n\times (|S(u)|/2^{\sqrt{\lg n}}+2^{\sqrt{\lg n}}))=O(|S(u)|\sqrt{\lg n}\lg \lg n+\lg n \times 2^{\sqrt{\lg n}})$ bits of space
	This subsumes the cost of storing $TS_{ds}(u)$ over $\hat{N}(u)$, which is $O(|S(u)|\lg \lg n)$ bits.
	As $T$ has $O(n/2^{\sqrt{\lg  n}})$ internal nodes,
	the total cost of storing these structures at all internal nodes is 
	$\sum_u O(|S(u)|\sqrt{\lg n}\lg \lg n+\lg n \times 2^{\sqrt{\lg n}})= O(n\lg n\lg \lg n+n\lg n)=O(n\lg n\lg \lg n)$ bits of space.
	Therefore, all the data structures occupy $O(n\lg n\lg \lg n)$ bits of space. 

	Finally, we analyze the construction time.
	As shown in part (a) of Lemma~\ref{ball_inheritance_large_d}, the tree $T$ with ball inheritance structures can be constructed in $O(n\sqrt{\lg n})$ time.
	For each internal node $u$ of $T$,  $TS_{ds}(u)$ can be constructed in linear time and $RS_{ds}(v)$ can be constructed in $O(|S(v)|+2^{\sqrt{\lg n}}\times\sqrt{\lg n}/\lg \lg n)$ time.
	As $T$ has $O(n/2^{\sqrt{\lg  n}})$ internal nodes, these structures over all internal nodes can be built in $\sum_u O(|S(v)|+2^{\sqrt{\lg n}}\times\sqrt{\lg n}/\lg \lg n)=O(n\sqrt{\lg n}+n\sqrt{\lg n}/\lg \lg n)=O(n\sqrt{\lg n})$ time.
	The preprocessing time of all data structures is thus $O(n\sqrt{\lg n})$.

\section{Optimal Orthogonal Sorted Range Reporting with Fast Preprocessing}
\label{app:sorted}

%
In this section we study the orthogonal sorted range reporting over $n$ points in 2d rank space. 
In our methods for three-sided sorted reporting and orthogonal sorted range reporting problems, we adopt the same strategy as shown in Section \ref{sect: fast_orthogonal_range_succ} which is to reduce a big point set $N$ into blocks of small point sets and sample several special points from each block. 
Both three-sided sorted reporting and orthogonal sorted range reporting queries over a block will take $O(\lg \lg n+\occ)$ time.
However, the points in the query range possibly distribute among different blocks.
As the $\lg \lg n$-term might subsumes the number of reported points from some block, we can not afford the $\lg \lg n$-term in the query time unless there are at least $\lg \lg n$ points reported from that block.
In this way, the $\lg \lg n$-term can be dismissed.

For the three-sided next point problem with the query range $[a, +\infty]\times[c, d]$, we sample the point with the maximum $x$-coordinate of each block.
Then for its counterpart problem three-sided sorted reporting, we need to sample $\lg \lg n$ points with largest $x$-coordinates from each block.
Similarly, for the orthogonal range successor problem, we sample the points from each block with the smallest $y$-coordinate for each distinct $x$-coordinates.
Then for its counterpart problem, we need to sample the points from each block with $\lg \lg n$ smallest $y$-coordinate for each distinct $x$-coordinates.
This sample strategy makes sure that if the number of points reported from the sampled point set that belongs to the same block $B$ is less than $\lg \lg n$, all points in $B\cap Q$ have been reported from the query over the sampled point set, where $Q$ denote the query range.
Otherwise, there are at least $\lg \lg n$ points in $B\cap Q$. 
And we can afford to query over the data structure built upon $B$.

Given the same query range, an answer to the orthogonal range successor is always the first point reported among the reported points from the orthogonal sorted range reporting query.
Our methods between the orthogonal range successor and orthogonal sorted range reporting are almost the same, apart from the sampling strategy described above.
In addition, our data structures can work in an online fashion: points within the query range $Q$ are reported in increasing order of $x$- or $y$-coordinates until the query procedure is terminated or all points in $Q$ are reported.

\subsection{Fast Construction of the Three-Sided Sorted Reporting Structures}
Now, we show how to efficiently construct data structures for three-sided sorted reporting.
Let $N$ be a set of $n$ points in 2d rank space. 
Given a query range $Q=[a, +\infty]\times[c, d]$, we define three-sided sorted reporting query to be the problem of reporting points in $N\cap Q$ in a increasingly sorted order by $y$-coordinates.
The methods to be shown in Lemmas \ref{lemma_three_sided_reporting_big} and \ref{lemma_three_sided_reporting_small} are under the indexing model.
The following previous result will be adopted in our method:

\begin{lemma}[{\cite[Lemma 5]{nekrich2012sorted}}]
	\label{lemma-three-sided-reporting-base}
	There exists a data structure of $O(n\lg^3 n)$-bit space constructed upon a set of $n$ points in 2d rank space in $O(n\lg^2 n)$ time, which supports three-sided reporting query in $O(\lg \lg n+\occ)$ time, where $\occ$ denotes the number of reported points.
\end{lemma}

\begin{lemma}[{\cite{nekrich2012sorted}}]
	\label{lemma_three_sided_reporting_block}
	Let $N$ be a set of $\lg^3 n$ points in rank space. 
	Given packed sequences $X$ and $Y$ respectively encoding the $x$- and $y$-coordinates of these points where $Y[i] = i$ for any $i \in [0, \lg^3 n-1]$, a data structure using $O(\lg^3 n\lg \lg n)$ bits of space constructed over $N$ in $o(\lg^3 n/\sqrt{\lg n})$ time that answers the three-sided sorted reporting query in $O(\lg \lg n+\occ)$ time. 
	The query algorithm requires access to a universal table of $o(n)$ bits.
\end{lemma}
\begin{proof}
	The proof is similar to the one shown in Lemma \ref{lemma_three_sided_next_point_block}.
	We construct almost the same data structure, apart from that upon the sampled point $\hat{N}$ from each block, we build data structure $DS(\hat{N})$ for three sided sorted reporting by Lemma \ref{lemma-three-sided-reporting-base}.
	Let $Q=[a, +\infty]\times[c, d]$ denote the query range, and $N_s$ and $N_e$ denote the blocks containing $c$ and $d$, where $s=\lfloor c/\lg^{3/4} n\rfloor$ and $e=\lfloor d/\lg^{3/4} n\rfloor$.
	If $e$ is equal to $s$, then the points in the answer to the query reside in the same subset $N_{s}$, and we can retrieve the target points in constant time by performing lookups with a universal table $U$ of $o(n)$ bits.
	$U$ has an entry for each possible set $(\alpha, \beta, \gamma, a', b', c', d')$, where $\alpha$ or $\beta$ is a packed sequence of length at most $\lg^{3/4} n$ drawn from $[\lg^3 n]$ denoting the $x$- or $y$-coordinates of the points, $\gamma$ is an integer $\in [0..(\lg^{3/4} n)-1]$ denoting the number of points, and $a', b', c', d'$ each is an integer $\in [0..\lg^3 n-1]$ such that all $a', b', c', d'$ together denote the query range.
	This entry stores a sorted point set of $\gamma$ points in the range $[a', b']\times[c', d']$.
	As $U$ has $O(2^{(\lg^{3/4} n)\times(6\lg \lg n)}\times \lg^{3 \times 4} n \times \lg^{3/4} n)$ entries and each entry stores a point set of at most $(6\lg \lg n)\times \lg^{3/4} n$ bits, $U$ uses $o(n)$ bits of space.
	If $s<e$, we sequentially check $A_1=N_s \cap Q$, $A_2=(N_{s+1} \cup \cdots \cup N_{e-1}) \cap [a, +\infty]\times [-\infty, +\infty]$, and $A_3=N_e \cap Q$, and report points in a sorted order in each of the three cases.
	Among them, points in $A_1$ and $A_3$ can be reported in contant time by performing lookups with $U$.
	It remains to report pioint in $A_2$.
	We query over $DS(\hat{N})$ in $O(\lg \lg n)$ time and retrieve all the blocks that each contains at least one point in $Q$.
	For each reported block $B$, we report points by performing lookups with $U$.
	Overall, the query time is $O(\lg \lg n+\occ)$. 
\end{proof}

\begin{lemma}[{\cite{nekrich2012sorted}}]
	\label{lemma_three_sided_reporting_big}
	Let the sequence $A[0..n'-1]$ of distinct elements drawn from $[n]$ denote a point set $N = \{(A[i], i)| 0\le i \le n'-1\}$, where $n^{\prime}\leq n$. 
	There exists a data structure using $O(n^{\prime}\lg \lg n)$ bits of extra space constructed over $N$ in $O(n^{\prime})$ time that answers three-sided sorted reporting query in $O(\lg \lg n+ t \times \occ)$ time, given that reporting the $x/y$-coordinate of a certain point of $A$ takes $O(t)$ time after the construction of the data structure. 
\end{lemma}
\begin{proof}
	The proof is similar to the one shown in Lemma \ref{lemma_three_sided_next_point_big}.
	We divide the points along $y$-axis of $N$ into blocks of length $\lg^{3} n$ each. 
	Within each block, we retrieve the $\lceil \lg \lg n \rceil$ points with largest $x$-coordinates into the point set $\hat{N}$.
	The capacity of $\hat{N}$ is $\lceil n'/\lg^3 n \rceil\times \lceil \lg \lg n \rceil$.
	We build in $O(|\hat{N}|\lg^{2} |\hat{N}|)=O(n' \lg^2 n'/\lg^3 n \times \lg \lg n)=o(n')$ time the data structure $DS(\hat{N})$ of $O(|\hat{N}|\lg^3 |\hat{N}|)=o(n')$ bits for three-sided sorted reporting by Lemma \ref{lemma-three-sided-reporting-base}.
	Over each block $N_i$ of points in rank space, we build in $o(\lg^3 n/\sqrt{\lg n})$ time the data structure $TS(N_i)$ of $O(\lg^3 n \lg \lg n)$ bits of space for three-sided reporting by Lemma \ref{lemma_three_sided_reporting_block}.
	The remaining data structures to be built are all the same as the proof of Lemma \ref{lemma_three_sided_next_point_big}.
	
	Let $Q=[a, +\infty]\times[c,d]$ denote the query range, and $N_s$ and $N_e$ denote the blocks containing $c$ and $d$, respectively.
	If $s$ is equal to $e$, we perform $\succ(a)$ over the index data structure of sorted $x$-coordinates of points from $N_s$ to retrieve $\hat{a}$ in rank space and perform a three-sided sorted reporting query in $N_s\times[\hat{a}, +\infty]\times[c \mod \lg^3 n, d \mod \lg^3 n]$, which requires $O(\lg \lg n+t\cdot \occ)$ time by Lemma~\ref{lemma_three_sided_reporting_block}.
	Note that once a point $e$ is reported from a block, we can compute its original $y$-coordinate by $i\times \lg^3 n+e.y$, where $i$ denotes the block index. 
	Then, we can retrieve its original $x$-coordinate by the computed $y$-coordinate.
	We dismiss the details here, but assume that the original $x$- and $y$-coordinates of $e$ can be retrieved in $O(t)$ time.
	Otherwise, we sequentially report points from $A_1=N_s\cap [a, +\infty]\times [c \mod \lg^3 n, +\infty]$, $A_2=(N_{s+1} \cup \cdots \cup N_{e-1})\cap [a, +\infty]\times [-\infty, +\infty]$, and $A_3=N_s\cap [a, +\infty]\times [0, d \mod \lg^3 n]$.
	We first take $O(\lg \lg n+t\cdot \occ_1)$ time to report points in $A_1$ following the similar way as we did when $s=e$, where $\occ_1=|A_1|$.
	Then, we query over $DS(\hat{N})$ for points in $A_2$. 
	If there are consecutive $\lceil \lg \lg n \rceil$ points reported from the same block $N_i$, it means that there are at least $\lceil \lg \lg n \rceil$ points in $N_i\cap Q'$, where $Q'=N_i \cap [a, +\infty]\times [-\infty, +\infty]$ and $s<i<e$.
	Then we query over the data structure $TS(N_i)$ for points in $N_i \cap Q'$ in $O(\lg \lg n+t \cdot \occ_i)=O(t \cdot \occ_i)$ time, where $\occ_i$ denotes the number of points in $N_i\cap Q'$.
	If some block $N_i$ contains less than $\lceil \lg \lg n \rceil$ points in $N_i \cap Q$, then all points in $N_i\cap Q$ are reported when performing queries over $DS(\hat{N})$ and we do not need to check $TS(N_i)$.
	Thus reporting points in $A_2$ requires $O(\lg \lg n+t\cdot \occ_2)$ time, where $\occ_2=|A_2|$.
	Finally, we query over $TS(N_e)$ for points in $A_3$ using $O(\lg \lg n+t\cdot \occ_3)$ time.
	Overall, the points in $N\cap Q$ can be reported in a sorted order along $y$-axis in $O(\lg \lg n+ t\cdot occ)$ time.
\end{proof}

More interestingly, when the $x$- and $y$-coordinates of the points are stored in the packed form, we can solve the three-sided sorted reporting with a data structure built in sublinear time. Here, we allow the point set $N$ to have duplicated $x$-coordinates.
\begin{lemma}
	\label{lemma_three_sided_reporting_small}
	Let $N$ be a set of $n'$ points with distinct $y$-coordinates in a $2^{\sqrt{\lg n}} \times n'$ grid, where $n^{\prime}=O(2^{c\sqrt{\lg n}})$ for any constant integer $c$. 
	Given packed sequences $X$ and $Y$ respectively encoding the $x$- and $y$-coordinates of these points where $Y[i] = i$ for any $i \in [0, n'-1]$, a data structure using $O(n^{\prime}\lg \lg n)$ bits of extra space constructed over $N$ in $O({n^{\prime}}/{\sqrt{\lg n}}\times \lg \lg n)$ time that answers a three-sided sorted reporting query in $O(\lg \lg n+t\times \occ)$ time, given that reporting $x/y$-coordinate of a certain point of $A$ takes $O(t)$ time after construction. 
\end{lemma}
\begin{proof}
	The proof is similar to the one shown in Lemma \ref{lemma_three_sided_next_point_small}.
\end{proof}

\subsection{Fast Construction of Orthogonal Sorted Range Reporting Structures}
Let $N$ be a set of $n$ points in 2d rank space. 
Given a query range $Q=[a, b]\times[c, d]$, we define the orthogonal sorted range reporting to be the problem of reporting points in $N\cap Q$ in a increasingly sorted order by $y$-coordinates.
In this subsection, we consider the orthogonal range queries in three different cases: on a $\lg^{1/4}\times n'$ small narrow grid, on a $2^{\sqrt{\lg n}}\times n'$ medium narrow grid, and eventually on an $n\times n$ grid.



\subsubsection{Orthogonal Sorted Range Reporting on a Small Narrow Grid}
First, we consider a special case such that the number of points is less than $\lg n$.
\begin{lemma}
	\label{lemma_orthogonal_sorted_range_reporting_base}
	Let $N$ be a set of $n'$ points with distinct $y$-coordinates in a $\lg^{1/4} n \times n'$ grid, where $n^{\prime}<\lg n$. 
	Given packed sequences $X$ and $Y$ respectively encoding the $x$- and $y$-coordinates of these points where $Y[i] = i$ for any $i \in [0, n'-1]$, a data structure of $O(n^{\prime}\lg \lg n)$ bits can be built in $o(n^{\prime}/\sqrt{\lg n})$ time over $N$ to answer orthogonal sorted range reporting in $O(\occ)$ time, where $\occ$ is the number of the reported points. 
\end{lemma}
\begin{proof}
	The data structure for orthogonal range successor queries shown in Lemma \ref{lemma_orthogonal_range_successor_base} can also be used for sorted range reporting queries.
	Therefore, we only show the query algorithm.
	Let $Q=[a, b]\times[c, d]$ denotes the query range, and $N_s$ and $N_e$  denote the block contain $c$ and $d$, respectively.
	We sequentially check $A_1=N_s \cap Q$, $A_2=(N_{s+1} \cup \cdots \cup N_{e-1}) \cap Q$ and $A_3=N_e \cap Q$ and report points in a sorted order.
	Both points in $A_1$ and $A_3$ can be reported in constant time with a universal table $U$ of $o(n)$ bits, similar to $U$ in the proof of Lemma \ref{lemma_three_sided_reporting_block}.
	As $\hat{N}$ has at most $\sqrt{\lg n}$ of points, querying over $\hat{N}$ can be also achieved by performing lookups with $U$.
	To report points in $A_2$, we query over $\hat{N}$ to find all blocks that contains at least one point in the query range.
	Then we iterate each reported block from left to right, and report the target points in a sorted order.
	Overall, the query time is $O(\occ)$.
\end{proof}

Next, we give the method for a point set with any number of points. 
As the data structure show in the proof of Lemma \ref{theorem_orthogonal_range_successor_small} can be used for both orthogonal range successor and sorted range reporting, we only show the query algorithm for orthogonal sorted range reporting in the proof of Lemma follows:

\begin{lemma}
	\label{theorem_orthogonal_reporting_small}
	Let the packed sequence $A[0..n'-1]$ drawn from $[\lg^{1/4} n]$ denote a point set $N = \{(A[i], i)| 0\le i \le n'-1\}$, where $n'\le n$.
	There is a data structure of $O(n^{\prime}\lg^2 \lg n+w\times\lg^{1/4} n)$ bits can be built in $O(n^{\prime}/\sqrt{\lg n})$ time over $N$ to answer orthogonal sorted range reporting in $O(\lg \lg n+\occ)$ time, where $\occ$ is the number of the reported points. 
\end{lemma}

\begin{proof}
	Lemma \ref{lemma_orthogonal_sorted_range_reporting_base} already subsumes this lemma when $n'< \lg n$, so it suffices to assume that $n' \ge \lg n$ in the rest of the proof.

	Let $Q=[a, b]\times[c,d]$ denote the query range and $\pi_a/\pi_b$ denote the path from $\lca(a, b)$ to $a/b$-th leaf, where $\lca(a, b)$ denotes the lowest common ancestor of the $a$-th and $b$-th leaves. 
	For each node $u$ on $\pi_a$ we mark the right child of $u$ if it exists and it does not stay on the path $\pi_a$. 
	For each node $u$ on $\pi_b$ we mark the left child of $u$ if it exists and it does not stay on the path $\pi_b$. 
	In addition, we mark the $a$-th and $b$-th leaves. 
	The points on the marked node have the $x$-coordinate in the range $[a, b]$. 
	As the height of $T$ is $O(\lg \lg n)$, there are in total $O(\lg \lg n)$ nodes marked. 
	
	The points at all marked nodes within the query range $Q$ can be identified in total $O(\lg \lg n)$ time. 
	Let $[c_v, d_v]$ denote the range such that $I(v)[c_v..d_v]$ within $[c, d]$.
	Recall that $I(v)$ is an index sequence associated with (but not explicitly stored at) each node $v$.
	Clearly, the range $[c_v, d_v]$ can be retrieved by operating $\rankop$ query over the sequence $S(u)$ where $u$ is the parent of $v$, i.e., $[c_v, d_v]=[\rankop_0(S(u), c_u),\rankop_0(S(u), d_u)]$ if $v$ is the left child of $u$. Otherwise, $[c_v, d_v]=[\rankop_1(S(u), c_u),\rankop_1(S(u), d_u)]$.
	As traversing each internal node $v$ on the path from the root node to the $a$-th leaf ($b$-th leaf, respectively), we keep operating $\rankop$ queries. And if a marked node $v$ is identified, we can find the range $[c_v, d_v]$ by $\rankop$ queries over $S(u)$ where $u$ is the parent of $v$.
	
	Note that the points associated with each node of $T$ are increasingly sorted by $y$-coordinates. 
	Once the marked nodes and the index range at each of them are identified, we can use the $O(\lg \lg n)$-way merge sorting algorithm to merge all points in $N\cap Q$ in a sorted order along $y$-axis.
	However, the merging algorithm takes $O(\lg \lg n\times \occ)$ time. 
	To speed up the querying efficiency, we can apply the $Q$-heap data structure \cite{FredmanW94} which supports to find the minimum among $O(\lg \lg n)$ $y$-coordinates of points in constant time.
	By combining the $Q$-heap with the $O(\lg \lg n)$-way merge sorting algorithm, reporting target points at marked nodes in a sorted order take $O(\lg \lg n+\occ)$ time.
	Note that the $\lg \lg n$-term spends on filling elements into Q-heap.
	Overall, the orthogonal sorted range reporting can be answered in $O(\lg \lg n+\lg \lg n+\occ)=O(\lg \lg n+\occ)$ time. 
\end{proof}

\subsubsection{Orthogonal Sorted Range Reporting on a Medium Narrow Grid}
In this subsection, we solve the orthogonal sorted range reporting over a point set $N$ on a  $2^{\sqrt{\lg n}}\times n^{\prime}$ grid with fast processing data structures.
We consider two cases depending on whether $2^{\sqrt{\lg n}}\leq n^{\prime} \le 2^{2{\sqrt{\lg n}}}$ or $2^{2\sqrt{\lg n}}< n^{\prime}$.

\begin{lemma}
	\label{orthogonal_sorted_range_reporting_on_block}
	Let $N$ be a set of $n'$ points with distinct $y$-coordinates in a $2^{\sqrt{\lg n}} \times n'$ grid, where $2^{\sqrt{\lg n}}\leq n^{\prime}\le 2^{2{\sqrt{\lg n}}}$. 
	Given packed sequences $X$ and $Y$ respectively encoding the $x$- and $y$-coordinates of these points where $Y[i] = i$ for any $i \in [0, n'-1]$, a data structure of $O(n'{(\lg n)}^{(1+\epsilon)/2}+w\times 2^{\sqrt{\lg n}})$ bits can be built over $N$ in $O(n'+2^{\sqrt{\lg n}}\times\sqrt{\lg n}/\lg \lg n)$ time to answer orthogonal sorted range reporting  in $O(\lg \lg n+\occ)$ time, where $\epsilon$ is any positive constant and $\occ$ is the number of the reported points. 
\end{lemma}
\begin{proof}
	The proof is similar to the one shown in Lemma \ref{ors_on_block}. 
	Here, we only discuss the different data structures that required to build.
	We build a $\lg^{1/4}$-ary wavelet tree $T$ upon $X[0, n'-1]$ and $Y[0, n'-1]$ with support for ball inheritance using Lemma~\ref{ball_inheritance_point}.
	We construct the data structures over the sequences associated with each internal node $u$ as follows:
	\begin{itemize}
		\item $TS_{ds}(u)$ over $\hat{N}(u)$ by Lemma \ref{lemma_three_sided_reporting_small} for three-sided sorted reporting;
		\item $RS_{ds}(u)$ over $\hat{S}(u)$ by Lemma \ref{theorem_orthogonal_reporting_small} for orthogonal sorted range reporting;
	\end{itemize}
	The overall space usage and the construction time are both dominated by the data structure for ball inheritance, which is $O(n'{(\lg n)}^{(1+\epsilon)/2}+w\times 2^{\sqrt{\lg n}})$ bits of space and $O(n'+2^{\sqrt{\lg n}}\times\sqrt{\lg n}/\lg \lg n)$ time, respectively.
	
	Let $Q$ denote the query range $[a, b]\times[c, d]$. 
	Similarly, we retrieve the lowest common ancestor $v$ of the $a$- and $b$-th leaf. 
	Let $v_s$ and $v_e$ each denote the child of $v$ on the path from $v$ to the $a$- and $b$-th leaf. 
	Note that $s$ and $e$ denote the child indexes, and $s\leq e$.
	At node $v_s$ and $v_e$, we query over $TS_{ds}(v_s)$ and $TS_{ds}(v_e)$ for reporting points in $\hat{N}(v_s)\cap [a, +\infty]\times[c_{v_s}, d_{v_s}]$ and $\hat{N}(v_e)\cap [0, b]\times[c_{v_e}, d_{v_e}]$ in a sorted order by $y$-coordinates of points, respectively, where $[c_{v_s}, d_{v_s}]=\noderange(c, d, v_s)$ and $[c_{v_e}, d_{v_e}]=\noderange(c, d, v_e)$.
	At node $v$, we query over $RS_{ds}(v)$ for reporting points in $\hat{N}(u)\cap [s+1, e-1]\times[c_v, d_v]$ in a sorted order, where $[c_{v}, d_{v}]=\noderange(c, d, v)$.
	Finally, we use 3-way merge-sorting algorithm to merge three sorted point lists reported into a single sorted list.
	As $\point(v,i)$ takes $O(1)$ time and $\noderange(c, d, v)$ takes $O(\lg \lg n)$ time, the overall query time is $O(\lg \lg n+\occ)$.
\end{proof}

Next, we provide a method when the capacity of the point set is more than $2^{2\sqrt{\lg n}}$.
Our method will use the following result:
\begin{lemma}[{\cite[Theorem 2]{nekrich2012sorted}}]
	\label{theorem_sorted_reporting}
	There exists a data structure of $O(n\lg^{1+\epsilon} n)$ bits constructed upon a set of $n$ points in 2d rank space in $O(n\lg n)$ time that supports the sorted range reporting queries in $O(\lg \lg n+occ)$ time, where $occ$ is the number of the reported points.
\end{lemma}

\begin{lemma}
	\label{lemma_range_sorted_reporting_core}
	Let $N$ be a set of $n'$ points with distinct $y$-coordinates in a $2^{\sqrt{\lg n}} \times n'$ grid,  where $2^{2\sqrt{\lg n}}< n'$. 
	Given packed sequences $X$ and $Y$ respectively encoding the $x$- and $y$-coordinates of these points where $Y[i] = i$ for any $i \in [0, n'-1]$, a data structure of $O(n'{(\lg n)}^{(1+\epsilon)/2}+w\times n'/2^{\sqrt{\lg n}})$ bits can be built over $N$ in $O(n^{\prime})$ time to answer orthogonal sorted range reporting in $O(\lg \lg n+\occ)$ time, where $\epsilon$ is any positive constant and $\occ$ is the number of reported points.
\end{lemma}
\begin{proof}
	The proof is similar to the one shown in Lemma \ref{lemma_range_successor_core_summary}.
	Let $b$ denote $2^{2\sqrt{\lg n}}$.
	We divide $N$ into $n'/b$ subsets, and for each $i \in [0, n'/b-1]$, the $i$-th subset, $N_i$, contains points in $N$ whose $y$-coordinates are in $[i b, (i+1)b-1]$.
	Let $p$ be a point in $N_i$. We call its coordinates $(p.x, p.y)$ {\em global coordinates}, while $(p.x', p.y') = (p.x, p.y \bmod b)$ its {\em local coordinates} in $N_i$; the conversion between global and local coordinates can be done in constant time.
	Hence the points in $N_i$ with their local coordinates can be viewed as a point set in a $2^{\sqrt{\lg n}}\times 2^{2\sqrt{\lg n}}$ grid, and we apply Lemma~\ref{orthogonal_sorted_range_reporting_on_block} to construct an orthogonal sorted range reporting structure $RS(N_i)$ over $N_i$.
	We also define a point set $\hat{N}$ in a $2^{\sqrt{\lg n}} \times n'$ grid. 
	For each set $N_i$ where $i \in [0, n'/b-1]$ and each $j \in [0, 2^{\sqrt{\lg n}}-1]$, if there exists $\geq \lceil \lg \lg n\rceil$ points in $N_i$ whose $x$-coordinate is $j$, we store $\lceil \lg \lg n\rceil$ points among them with smallest $y$-coordinates into $\hat{N}$.
	Otherwise, we store all points in $N_i$ whose $x$-coordinate is $j$ into $\hat{N}$.
	Thus the number of points in $\hat{N}$ is at most $\lceil \lg \lg n\rceil \times 2^{\sqrt{\lg n}} \times n' /b = \lceil \lg \lg n\rceil n'/2^{\sqrt{\lg n}}$. 
	Finally, we build the data structure $\hat{RS}_{ds}(\hat{N})$ for orthogonal sorted range reporting over $\hat{N}$ by Lemma~\ref{theorem_sorted_reporting}.
	
	Given a query range $Q=[x_1, x_2]\times[y_1, y_2]$, we first check if $\lfloor y_1/b \rfloor$ is equal to $\lfloor y_2/b\rfloor$. If it is, then the points in the answer to the query reside in the same subset $N_{\lfloor y_1/b\rfloor}$, and we can report points in $N_{\lfloor y_1/ b \rfloor}\cap Q$ in a sorted order by querying over $RS(N_{\lfloor y_1/ b \rfloor})$, which requires $O(\lg\lg n+\occ)$ time  by Lemma~\ref{orthogonal_sorted_range_reporting_on_block}.
	Otherwise, let $N_s,\dots, N_e$ denote the blocks interacting $[y_1, y_2]$, where $s=\lfloor y_1/ b\rfloor$ and $e=\lfloor y_2/ b\rfloor$. 
	We sequentially look for points in $A_1=N_s \cap [x_1, x_2]\times[y_1 \mod b, +\infty]$, $A_2=(N_{s+1}\cup \cdots \cup N_{e-1})\cap [x_1, x_2]\times [-\infty, +\infty]$, and $A_3=N_e \cap [x_1, x_2]\times[0, y_2 \mod b]$.
	Both the cases $A_1$ and $A_3$ can be answered in $O(\lg \lg n+\occ_1)$ and $O(\lg \lg n+\occ_3)$ time by Lemma~\ref{orthogonal_sorted_range_reporting_on_block}, where $\occ_1=|A_1|$ and $\occ_2=|A_2|$, respectively.
	It remains to show how to report points in $A_2$.
	We query over $\hat{RS}_{ds}(\hat{N})$ to report points in $Q\cap \hat{N}$ by Lemma~\ref{theorem_sorted_reporting}.
	If consecutive $\lceil \lg \lg n \rceil$ points reported by querying over $\hat{RS}_{ds}(\hat{N})$ are from a same bock $N_i$, then we immediately query over $RS_{ds}(N_i)$ for the remaining points in $N_i$, where $s<i< e$.
	It requires $O(\lg \lg n+\occ_i)=O(\occ_i)$ time, as $\lceil \lg \lg n \rceil\leq \occ_i$, where $\occ_i$ is the number of points in $N_i\cap Q$.
	Otherwise, if some block $N_i$ contains less $\lceil \lg \lg n \rceil$ points in $N_i \cap Q$, then all points in $N_i \cap Q$ has been reported by querying over $\hat{RS}_{ds}(\hat{N})$.
	Overall, the query procedure requires $O(\lg \lg n+\occ)$ time.
	
	To bound the storage costs, by Lemma~\ref{orthogonal_sorted_range_reporting_on_block}, the orthogonal range reporting structure over each $N_i$ uses $O(2^{2\sqrt{\lg n}}\lg^{1/2+\epsilon} n+w\cdot 2^{\sqrt{\lg n}})$ bits.
	Thus, the range reporting structures over $N_0, N_1, \ldots, N_{n/b-1}$ occupy $O((n'/b)\times(2^{2\sqrt{\lg n}}\lg^{1/2+\epsilon} n+w\cdot 2^{\sqrt{\lg n}})) = O(n'\lg^{1/2+\epsilon} n + n'w/2^{\sqrt{\lg n}})$.
	As there are at most $\lceil \lg \lg n\rceil n'/2^{\sqrt{\lg n}}$ points in $\hat{N}$, by Lemma~\ref{theorem_sorted_reporting}, the sorted reporting structure for $\hat{N}$ occupies $O((\lg \lg n) n'\lg^{1+\epsilon} n/2^{\sqrt{\lg n}}) = o(n')$ bits.
	Thus the  space costs of all structures add up to $O(n^{\prime}\lg^{1/2+\epsilon} n+n'w/2^{\sqrt{\lg n}})$ bits.
	
	Regarding construction time, observe that the point sets $N_0, N_1, \ldots, N_{n'/b-1}$ and $\hat{N}$ can be computed in $O(n')$ time.
	By Lemma~\ref{theorem_sorted_reporting}, the sorted range reporting structure for $\hat{N}$ can be built in $O(n'/b \times \lg n) = o(n')$ time.
	Finally, the total construction time of the sorted range reporting structures for $N_0, N_1, \ldots, N_{n/b-1}$ is $O({n^{\prime}}/{2^{2\sqrt{\lg n}}}\times(2^{2\sqrt{\lg n}}+\sqrt{\lg n}\times 2^{\sqrt{\lg n}}/\lg \lg n))=O(n^{\prime})$, which dominates the total preprocessing time of all our data structures.
\end{proof}


Combining Lemma \ref{lemma_range_sorted_reporting_core} and Lemma \ref{orthogonal_sorted_range_reporting_on_block}, we have the following result on the orthogonal sorted range reporting:
\begin{lemma}
	\label{lemma_range_sorted_reporting_core_summary}
	Let $N$ be a set of $n'$ points with distinct $y$-coordinates in a $2^{\sqrt{\lg n}} \times n'$ grid,  where $2^{\sqrt{\lg n}}\leq n'\leq n$. 
	Given packed sequences $X$ and $Y$ respectively encoding the $x$- and $y$-coordinates of these points where $Y[i] = i$ for any $i \in [0, n'-1]$, a data structure of $O(n'{(\lg n)}^{(1+\epsilon)/2}+w\times (n'/2^{\sqrt{\lg n}}+2^{\sqrt{\lg n}}))$ bits can be built over $N$ in $O(n^{\prime}+2^{\sqrt{\lg n}}\times\sqrt{\lg n}/\lg \lg n)$ time to answer orthogonal sorted range reporting in $O(\lg \lg n+\occ)$ time, where $\epsilon$ is any positive constant and $\occ$ is the number of reported points.
\end{lemma}

\subsubsection{Orthogonal Sorted Range Reporting  on a $n\times n$ Grid}
Finally, we show the data structure with the fast construction algorithm for $n$ points in 2d rank space. 
\begin{theorem}
	Let $N$ denote a set of $n$ points in $2$d rank space. 
	A data structure of $O(n\lg^{1+\epsilon} n)$ words of space can be built over $N$ in $O(n\sqrt{\lg n})$ time to answer orthogonal sorted range reporting in $O(\lg \lg n+\occ)$ time, where $\occ$ is the number of reported points and $\epsilon$ is any small positive constant.
\end{theorem}
\begin{proof}
	Let the sequence $X[0, n-1]$ denote the point set $N$ such that $N = \{(X[i], i)| 0\le i \le n-1\}$.
	We build a $2^{\sqrt{\lg n}}$-ary wavelet tree $T$ upon $X[0, n-1]$ with support for ball inheritance using part (b) of Lemma~\ref{ball_inheritance_large_d}.
	Recall that $A(v)$ stores the $x$-coordinates of the ordered list, $N(v)$, of points from $N$ whose $x$-coordinates are within the range represented by $v$, and these points are ordered by $y$-coordinate. 
	Furthermore, $v$ is associated with another sequence $S(v)$ drawn from alphabet $[2^{\sqrt{\lg n}}]$, in which $S(v)[i]$ encodes the rank of the child of $v$ that contains $N(v)[i]$ in its ordered list.
	We regard $A(u)$ at each internal $u$ as a point set $\hat{N}(u)= \{(A(u)[i], i)| 0\le i \le |A(u)|-1\}$ and construct the data structure $TS_{ds}(u)$ over $\hat{N}(u)$ for three-sided sorted reporting by Lemma \ref{lemma_three_sided_reporting_big}. 
	We regard $S(u)$ at each internal $u$ as a set $\hat{S}(u))= \{(S(u)[i], i)| 0\le i \le |S(u)|-1\}$ and construct the data structure $RS_{ds}(u)$ over $\hat{S}(u)$ for orthogonal sorted range reporting by Lemma \ref{lemma_range_sorted_reporting_core_summary}.
	
	Given a query range $Q=[a, b]\times[c, d]$,  we first locate the lowest common ancestor $u$ of $l_{a}$ and $l_{b}$ in constant time, where $l_{a}$ and $l_{b}$ denote the $a$-th and $b$-th leftmost leaves of $T$, respectively.
	Let $u_i$ denote the $i$-th child of $u$, for any $i \in [0, 2^{\sqrt{\lg n}}-1]$.
	We first locate two children, $u_{a'}$ and $u_{b'}$, of $u$ that are ancestors of $l_a$ and $l_b$, respectively.
	They can be found in constant time by simple arithmetic as each child of $u$ represents a range of equal size.
	Then the answer, $Q\cap N$, to the query can be partitioned into three point sets $A_1=Q\cap N(u_{a'})$, $A_2=Q\cap (N(u_{a'+1})\cup N(u_{a'+2})\cup\ldots N(u_{b'-1}))$ and $A_3=Q\cap N(u_{b'})$.
	At node $u_{a'}$, we query over $TS_{ds}(u_{a'})$ for reporting all points in $[a, +\infty]\times[c_{u_{a'}}, d_{u_{a'}}]$ in a sorted order by $y$-coordinates in $O(\lg \lg n+\occ_0)$ time, where $\occ_0$ is the number of the reported points and $[c_{u_{a'}}, d_{u_{a'}}]=\noderange(c,d, {u_{a'}})$.
	At node $u_{b'}$, we query over $TS_{ds}(u_{b'})$ for reporting all points in $[0, b]\times[c_{u_{b'}}, d_{u_{b'}}]$ in a sorted order by $y$-coordinates in $O(\lg \lg n+\occ_1)$ time, where $\occ_1$ is the number of the reported points and $[c_{u_{b'}}, d_{u_{b'}}]=\noderange(c,d, {u_{b'}})$.
	At node $u$, we query over $RS_{ds}(u)$ for reporting all points in $[a'+1, b'-1]\times[c_{u}, d_{u}]$  in a sorted order by $y$-coordinates in $O(\lg \lg n+\occ_2)$ time, where $\occ_2$ is the number of the reported points and $[c_{u}, d_{u}]=\noderange(c,d, {u})$.
	With constant-time support for $\point$, we can retrieve the original $x$- and $y$-coordinates of each reported point in contant time.
	Finally, we use a 3-way merge-sorting algorithm to merge three sorted list into one sorted list in $O(\occ_0+\occ_1+\occ_2)=O(\occ)$ time.
	Therefore, the overall query time for orthogonal sorted range reporting is $O(\lg \lg n+\occ)$ time.
	
	Now we analyze the space costs.
	$T$ with support for ball inheritance uses $O(n\lg^{1+\epsilon'}n+n\lg n)=O(n\lg^{1+\epsilon'}n)$ bits for any positive $\epsilon'$.
	For each internal node $v$, since $w = \Theta(\lg n)$, the data structure $RS_{ds}(u)$ for orthogonal sorted range reporting over $\hat{S}(v)$ uses $O(|S(u)|\lg^{1/2+\epsilon''} n+2^{\sqrt{\lg n}}\lg n+|S(u)|\lg n/2^{\sqrt{\lg n}})$ bits for any positive $\epsilon''$. 
	This subsumes the cost of the data structure $TS_{ds}(v)$ for three-sided sorted reporting over $\hat{N}(v)$ which is $O(|A(v)|\lg \lg n)$ bits. 
	As $T$ has $O(n/2^{\sqrt{\lg  n}})$ internal nodes,
	the total cost of storing these structures at all internal nodes is $\sum_u O(|S(u)|\lg^{1/2+\epsilon''} n+2^{\sqrt{\lg n}}\lg n+|S(u)|\lg n/2^{\sqrt{\lg n}}) = O(n\lg n/\sqrt{\lg n} \times \lg^{1/2+\epsilon''}n + n \lg n) = O(n\lg n\lg^{\epsilon''} n)$.
	Setting $\epsilon = max(\epsilon', \epsilon'')$, the total space bound turns to be $O(n\lg^{1+\epsilon}  n)$ bits. 
	Overall, the data structures occupy $O(n\lg^{1+\epsilon} n)$ bits.
	
	Finally, we analyze the construction time.
	As shown in Lemma~\ref{ball_inheritance_large_d}, $T$ with support for ball inheritance can be constructed in $O(n \sqrt{\lg n})$ time.
	For each internal node $v$ of $T$, constructing $TS_{ds}(v)$ over $\hat{N}(v)$ and the orthogonal sorted range reporting structure $RS_{ds}(v)$ over $\hat{S}(v)$ requires $O(|A(u)|+|S(u)|+\sqrt{\lg n}\cdot 2^{\sqrt{\lg n}}/\lg \lg n)=O(|S(u)|+\sqrt{\lg n}\cdot 2^{\sqrt{\lg n}}/\lg \lg n)$ time.
	As $T$ has $O(n/2^{\sqrt{\lg  n}})$ internal nodes, these structures over all internal nodes can be built in $\sum_u O(|S(u)|+\sqrt{\lg n}\times2^{\sqrt{\lg n}}) = O(n\lg n/\sqrt{\lg n}+n \sqrt{\lg n}/\lg \lg n) = O(n \sqrt{\lg n})$ time.
	Therefore, the preprocessing time of all data structures is hence $O(n\sqrt{\lg n})$.
\end{proof}

\end{document}